\documentclass[noshowframe]{imsart}

\RequirePackage{amsthm,amsmath,amsfonts,amssymb}
\RequirePackage[numbers]{natbib}
\usepackage{multibib}
\newcites{Supp}{References for Supplementary}
\RequirePackage[colorlinks,citecolor=blue,urlcolor=blue]{hyperref}
\RequirePackage{graphicx}
\usepackage{subfigure}
\usepackage[linesnumbered,ruled,lined,boxed]{algorithm2e}
\usepackage{float}
\usepackage{bbm}
\usepackage{rotating}
\usepackage{booktabs}
\usepackage{xr}
\usepackage{soul}
\usepackage{rotating}
\usepackage{array}

\startlocaldefs
\theoremstyle{plain}

\newtheorem{proposition}{\bf Proposition}[section]
\theoremstyle{definition}

\usepackage{multirow}
\def\qmo{``}
\def\qmc{''}
\def\qmcsp{'' }

\theoremstyle{remark}

\endlocaldefs

\begin{document}
\begin{frontmatter}
\title{Fast QR updating methods for statistical applications}
\runtitle{Fast QR methods for statistical applications}

\begin{aug}
\author[A]{\fnms{Mauro}~\snm{Bernardi}\ead[label=e1]{mauro.bernardi@unipd.it}\orcid{0000-0001-5759-9892}},
\author[A]{\fnms{Claudio}~\snm{Busatto}\ead[label=e2]{claudio.busatto@unifi.it}\orcid{0009-0008-6629-903X}}
\and
\author[A]{\fnms{Manuela}~\snm{Cattelan}\ead[label=e3]{manuela.cattelan@unipd.it}\orcid{0000-0002-5631-5293}}
\address[A]{Department of Statistical Sciences,
University of Padova\printead[presep={,\ }]{e1}}

\runauthor{M. Bernardi et al.}
\end{aug}

\begin{abstract}
This paper introduces fast R updating algorithms specifically designed for statistical applications, including regression, filtering, and model selection, where data structures change frequently. Although traditional QR decomposition is essential for matrix operations, it becomes computationally intensive when dynamically updating the design matrix in statistical models. The proposed algorithms efficiently update the R matrix without the need for recalculation of Q, thereby significantly reducing computational costs in practical computational scenarios. The provision of scalable solutions for high-dimensional regression models is a key strength of these algorithms, enhancing the feasibility of large-scale statistical analyses and model selection in data-intensive fields. A thorough simulation study and the analysis of real-world data demonstrate that the methods achieve a substantial reduction in computational time without compromising accuracy. The discussion illustrates the benefits of these algorithms across a wide range of models and applications in statistics and machine learning.
\end{abstract}



\end{frontmatter}

\section{Introduction}

\noindent The QR decomposition is a foundational technique in computational statistics and machine learning, providing an efficient method for matrix factorization that guarantees numerical stability. Its versatility encompasses a wide range of applications,  including  Cholesky factorization \citep{van_loan.1997}, eigenvalue computation \citep{jolliffe_cadima.2016},  canonical correlation analysis \citep{monahan.2001}, and, most importantly, solving linear systems and computing least squares estimates, which are cornerstones of most statistical analyses \citep{golub_van_loan.2013}.  For instance, in stepwise regression and model selection \citep{hastie_etal.2020}, the QR decomposition significantly improves the reliability of recursive least squares by effectively addressing multicollinearity and minimizing numerical precision errors. In filtering theory and sequential learning, where error propagation can lead to instability, QR-based square root Kalman filtering \citep{tracy.2022} offers a fast and robust approach for recursive calculations. Furthermore, the QR decomposition is essential for maintaining computational speed and accuracy when fitting complex structured regressions and generalized additive models, especially with large datasets and complex smoothing functions \citep{wood_etal.2016}, where alternative methods could become computationally intensive or prone to error.\par 
The QR decomposition breaks an $N\times p$ matrix, where $N>p$, into an $N \times N$ orthogonal matrix Q and an upper trapezoidal $N\times p$ matrix R. Although the QR factorization is computationally demanding---with a complexity of $\mathcal{O}(Np^2)$ for operations and $\mathcal{O}(Np)$ for storage \citep[see, e.g.][]{golub_van_loan.2013}---it is often unnecessary to recompute both Q and R from scratch when the matrix is modified by adding or removing elements. This is particularly advantageous in large-scale statistical methods, such as Bayesian model selection, hypothesis testing, routinely applied, for example, to infer graph structures \citep{colombo.2014}, iterative regression techniques \citep{efron_etal.2004}, and optimization via data augmentation \citep{boyd_etal.2011}, where frequent QR updates are required.\par  
Methods for updating matrix factorizations have a long-established history in statistics, particularly in solving least squares problems and computing covariance matrices along with their inverses \citep{hager.1989}. Block downdating has also been extensively explored, with significant contributions from \cite{yanev_kontoghiorghes.2004}, which have developed robust techniques for adjusting factorizations. Recent advancements by \cite{wei_etal.2020} have refined these methods, especially in recursive least squares applications, yielding significant improvements in computational efficiency. A major limitation of these techniques lies in their need to compute and store both the Q and R matrices, which can place significant strain on memory and computational resources. While both matrices are essential to the overall algorithm, in most statistical applications the Q matrix is calculated primarily to update R, adding unnecessary complexity and increasing resource demands. To overcome this challenge, we propose a streamlined approach to QR updating and downdating that emphasizes direct updates to the R matrix, thereby eliminating the need to store and compute the Q matrix. By updating the R matrix incrementally, the computational cost can be significantly reduced, avoiding the need for full matrix re-decomposition each time a change occurs \citep[see, e.g.][]{bjorck.2015}.  This efficiency is vital for handling large datasets and enabling real-time model adjustments. To support this, we provide a detailed analysis of the computational costs associated with various algorithms, empowering users to make well-informed decisions and select the most effective approach for their specific requirements.\par
We validate the proposed R updating algorithms through a combination of simulation studies and real-world data experiments. The simulation studies assess the performance of the algorithms in the challenging context of Bayesian model selection for high-dimensional regression. By simulating datasets with varying number of covariates, observations, and levels of correlation among explanatory variables, we evaluate how effectively the methods address large-scale problems. The results demonstrate a significant reduction in computational time while maintaining the same accuracy in posterior inference, compared to traditional QR decomposition methods. Furthermore, the proposed methods achieve up to 1500-fold  speed improvements compared to state-of-the-art algorithms, even in worst-case scenarios. Real-data experiments further illustrate the broad applicability of the R updating and downdating algorithms, underscoring their practical value in complex statistical modeling scenarios.\par 
The proposed algorithms efficiently address the computational challenges of high-dimensional data and solution path computation, enabling effective parameter tuning, model selection and validation, and parameter estimation. By providing access to the full solution path and demonstrating strong practical performance across diverse applications, they integrate seamlessly into modern data analysis workflows and advance contemporary statistical methodology.\par
The remainder of the paper is structured as follows. Section \ref{sec:qrupdate_algo} provides a comprehensive review of classical QR updating and downdating methods, detailing their mathematical formulations and algorithmic implementations. Section \ref{sec:thinqrupdate_algo} introduces the proposed R updating and downdating algorithms. Section \ref{sec:computational_costs} presents a precise analysis of the computational costs, quantified in terms of floating-point operations (FLOPS), for each algorithm. Sections \ref{sec:simulations} and \ref{sec:applications} offer empirical results from simulations and real-world case studies, demonstrating the effectiveness of R updates in various statistical applications. 
Section \ref{sec:discussion} provides a comprehensive analysis of how the proposed methods can improve the efficiency of existing techniques in statistics and machine learning, while also identifying promising directions for future research. The paper concludes with Section \ref{sec:conclu}. An extensive supplementary material document provides additional details on the proofs of the main results, algorithmic procedures, computational cost analyses, and further results from the simulation studies and real data analyses. Building on the methods developed in this paper, an open-source R package, \qmo\texttt{fastQR}\qmc, has been released. This is available on CRAN. The package provides efficient functions for building, updating and downdating QR decompositions, including the simultaneous modification of multiple rows and columns, facilitating practical application of the proposed algorithms in high-dimensional regressions and model selection.
%
\subsection{Notation}
%
Let $\mathbf{X}\in\mathbb{R}^{N \times p}$ denote a generic matrix of dimension $N\times p$, with $ N > p$ and real entries $x_{ij}$ and let $\mathbf{X}^\top$ denote its transpose. $\mathbf{X}\left[r_1:r_2,\,\right]$ and $\mathbf{X}\left[\,,c_1:c_2\right]$ denote the sub-matrices which include rows from  $r_1$ to $r_2$ or columns from $c_1$ to $c_2$, respectively, while $\mathbf{X}\left[r_1:r_2,c_1:c_2\right]$ is the block of the matrix which includes entries $x_{ij}$ such that $r_1 \leq i \leq r_2$ and $c_1 \leq j \leq c_2$. $\mathbf{I}_p$ denotes the identity matrix of dimension $p$ and $\mathbf{0}_{N,p}$ is the matrix of zero elements of dimension $N\times p$. A column vector of zeros of dimension $p$ is denoted as $\mathbf{0}_{p}$, while a column vector of ones of dimension $p$ is $\boldsymbol{\iota}_p$.
The square permutation matrix, that moves row $k$ to position $l$ is denoted by $\mathcal{P}(k, l)$. 
Finally, a Givens matrix $\mathbf{G}_k(i,j)$, that zeroes element $j$ in column $k$ of matrix $\mathbf X$, is an $N \times N$ identity matrix, with specific non-zero elements in cells $(i,i)$, $(i,j)$, $(j,i)$ and $(j,j)$, see \cite{golub_van_loan.2013} for details. The QR decomposition can be computed by applying a sequence of Givens rotations to sequentially set to zero all elements under the diagonal of $\mathbf X$, and matrix $\mathbf Q$ will be the product of such Givens matrices.

\section{QR updating algorithms}
\label{sec:qrupdate_algo}
%
Let $\mathbf{X} \in \mathbb{R}^{N \times p}$, with $N \ge p$, be of full column rank, with decomposition $\mathbf{X}= \mathbf{Q} \mathbf{R}$ where $\mathbf{Q}\in\mathbb{R}^{N\times N}$ is an orthogonal matrix and $\mathbf{R}\in\mathbb{R}^{N\times p}$ is an upper trapezoidal matrix. The aim is to update the matrices $\mathbf{Q}$ and  $\mathbf{R}$ when one or more rows (or columns) are added to (or deleted from) the matrix $\mathbf{X}$. We assume that the matrix $\mathbf{X}$ remains of full column rank after the modification to the rows or columns. Hereafter, we report the results for adding or deleting one row (or column), while the extension to a block of rows (or columns) is reported in Appendix \ref{secA1}. The algorithms for implementing the updates described  here, in Section \ref{sec:thinqrupdate_algo}, and in Appendix \ref{secA1}, are reported in Section \ref{appB} of the supplementary material.
%
\subsection{Adding and deleting rows}
\label{subsec:add_del_rows}
%
Matrices  $\mathbf{Q}$ and  $\mathbf{R}$ can be updated when one or more rows are added to or deleted from the matrix $\mathbf{X}$.
First, consider the addition of one row, $\mathbf{x}_{\star}\in\mathbb{R}^p$, in position $k$, such that 
\begin{equation*}
\mathbf{X}^{+} = \begin{bmatrix} \mathbf{X}\left[ 1 : (k-1), \;\; \right] \\ \mathbf{x}_{\star} ^{\top} \\ \mathbf{X}\left[ k : N, \; \right] \end{bmatrix} \quad \text{ and } \quad
\mathcal{P}(k, N+1) \mathbf{X}^{+} = \begin{bmatrix} \mathbf{X} \\ \mathbf{x}_{\star} ^{\top} \end{bmatrix},
\end{equation*}
where $\mathbf{X}^{+} \in \mathbb{R}^{(N+1) \times p}$ and the permutation matrix moves the vector $\mathbf{x}_{\star} ^{\top}$ to the bottom of matrix $\mathbf{X}^{+}$. Then,
\begin{equation}
\begin{bmatrix} \mathbf{Q}^{\top} & \mathbf{0}_N\\ \mathbf{0}_N^\top & 1 \end{bmatrix} \mathcal{P}(k, N+1) \mathbf{X}^{+} = \begin{bmatrix} \mathbf{R} \\ \mathbf{x}_{\star} ^{\top} \end{bmatrix} = \widetilde{\mathbf{R}}.
\label{eq:12}
\end{equation}
The new matrix $\mathbf{R}^{+}$, such that $\mathbf{X}^{+}= \mathbf{Q}^{+}\mathbf{R}^{+}$, is obtained by sequentially setting to zero the last row of $\widetilde{\mathbf{R}}$ though a sequence of Givens rotations, see Figure \ref{fig:WD_add_1_row} for a graphical representation of the procedure, while $\mathbf{Q}^{+}$ is recovered by applying the same set of Givens rotations 
\begin{align}
\label{eq:addingonerow_1}
\mathbf{R}^{+} &= \mathbf{G}_p (p, N+1)^\top \cdots \mathbf{G}_1(1, N+1)^\top\widetilde{\mathbf{R}}\\
\label{eq:addingonerow_Q}
\mathbf{Q}^{+}&=(\mathcal{P}(k, N+1))^\top\begin{bmatrix} \mathbf{Q} & \mathbf{0}_N\\ \mathbf{0}_N^\top & 1 \end{bmatrix}\mathbf{G}_1(1, N+1)\cdots\mathbf{G}_p (p, N+1).
\end{align}
Now, consider the deletion of one row, so that
\begin{equation*}
\mathbf{X}^{-} = \begin{bmatrix} \mathbf{X}\left[ 1 : (k-1), \;\right] \\ \mathbf{X}\left[ (k+1) : N, \;\; \right] \end{bmatrix} \quad \text{ and } \quad \mathcal{P}(k,1) \mathbf{X} = \begin{bmatrix} \mathbf{x}_k^\top \\ \mathbf{X}^{-} \end{bmatrix}.
\end{equation*}
Thus, it is necessary to find $\mathbf{Q}^{-}$ and $\mathbf{R}^{-}$ such that
\begin{equation*}
\mathcal{P}(k,1)\mathbf{X} = \begin{bmatrix} \alpha & \mathbf{0}_{N-1}^\top \\ \mathbf{0}_{N-1} & \mathbf{Q}^{-} \end{bmatrix} \begin{bmatrix} \mathbf{z}_k^\top \\ \mathbf{R}^{-} \end{bmatrix}=\mathcal{P}(k,1)\mathbf{Q}\mathbf{R}=\mathbf{Q}_{p}\mathbf{R},
\end{equation*}
where $\mathbf{Q}_{p}=\mathcal{P}(k,1)\mathbf{Q}$. The $(N-1)$ elements of the first row of $\mathbf{Q}_p$ can be zeroed  leveraging a sequence of Given rotations
\begin{equation} 
\label{eq:deleteonerow_givens}
\mathbf{G}_1(1, 2)^{\top} \cdots \mathbf{G}_1(N-1, N)^{\top}  \mathbf{Q}_p^\top = \begin{bmatrix} \alpha & \mathbf{0}_{N-1}^\top \\ \mathbf{0}_{N-1} & (\mathbf{Q}^{-})^\top \end{bmatrix}.
\end{equation}
The same set of Givens rotations applied to $\mathbf{R}$ yields
\begin{equation*} 
\mathbf{G}_1(1, 2)^{\top} \cdots\mathbf{G}_1(N-1, N)^{\top}  \mathbf{R}  = \begin{bmatrix} \mathbf{z}_k^\top \\ \mathbf{R}^{-} \end{bmatrix},
\end{equation*}
where $\mathbf{ R}^{-}$ is upper trapezoidal, so
\begin{align*}
\mathcal{P}(k,1)\mathbf{X} &= \Big[\mathbf{Q}_p  \mathbf{G}(N-1, N)  \cdots \mathbf{G}(1, 2)\Big]
 \Big[ \mathbf{G}(1, 2)^{\top} \cdots \mathbf{G}(N-1, N)^{\top} \mathbf{R}\Big]  \nonumber \\
&= \begin{bmatrix} \alpha & \mathbf{0}_{N-1}^\top \\ \mathbf{0}_{N-1} & \mathbf{Q}^{-} \end{bmatrix} \begin{bmatrix} \mathbf{z}_k^\top \\ \mathbf{R}^{-} \end{bmatrix}.
\end{align*}
Hence, the new matrices $\mathbf{R}^{-}$ and $\mathbf{Q}^{-}$ are obtained by appropriate multiplication of $\mathbf{R}$ and $\mathbf{Q}$ by the sequence of Givens matrices used to erase $\mathbf{Q}_p[1,]$. 
%
\begin{figure}[!t]
\begin{center}
\subfigure[{\scriptsize Add $1$ row}\label{fig:WD_add_1_row}]{\includegraphics[height=2.4cm]{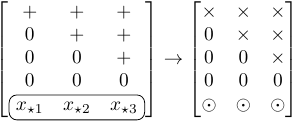}}
\hspace*{0.05em}
\subfigure[{\scriptsize Add 1 column}\label{fig:WD_add_col}]{
\raisebox{-0.5mm}{\includegraphics[height=2.6cm]{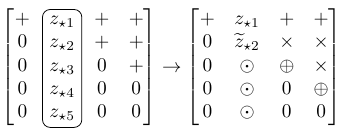}}}
\subfigure[{\scriptsize Delete 1 column}\label{fig:WD_del_col_red}]{
\includegraphics[height=2.5cm]{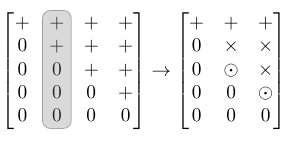}} \quad
\caption{Graphical representation of the effect of Givens rotations on $\widetilde{\mathbf{R}}$ in order to obtain the new $\mathbf{R}$. The boxed row/column are added cells, the column in gray is deleted. $\odot$ indicates cells that are zeroed, $\oplus$ indicates cells that were $0$ and after the update assume values different from $0$, finally $\times$ indicate cells whose value is modified from the starting one $(+)$ as a consequence of the update.}
\label{fig:R_modifications}
\end{center}
\end{figure}
%
%
\subsection{Adding and deleting columns}
\label{subsec:add_del_cols}
%
\noindent Now consider the update of $\mathbf{Q}$ and $\mathbf{R}$ when one or more columns are added to or deleted from the matrix $\mathbf{X}$.
First, consider the inclusion of a column, $ \mathbf{x}_\star\in\mathbb{R}^N$, in generic positions $k$, so that
$\mathbf{X}^{+}= \begin{bmatrix} \mathbf{X}\left[ ,\; 1 : (k-1) \right] & \mathbf{x}_\star & \mathbf{X}\left[,\;  k : p\right] \end{bmatrix}$ where $\mathbf{X}^{+} \in \mathbb{R}^{N \times (p+1)}$.
Leveraging the QR decomposition and defining $\mathbf{z}_\star = \mathbf{Q}^{\top} \mathbf{x}_\star$,  we get
\begin{equation}
\label{eq:addonecol}
\mathbf{Q}^{\top} \mathbf{X}^{+} = \begin{bmatrix} \mathbf{R}\left[ , \;\; 1:(k-1) \right] & \mathbf{z}_\star &  \mathbf{R}\left[ , \;\; k:p \right]\end{bmatrix} = \widetilde{\mathbf{R}}.
\end{equation}
Computing $\mathbf{R}^{+}$ requires setting to zero elements $\mathbf{z}_\star \left[ (k+1) : N \right]$ and filling entries $\mathbf{R}\left[i+1, \;\; i \right]$, $i = k, \dots, p$.
This can be achieved by applying $N-k$ Givens rotations to $\widetilde{\mathbf{R}}$, see Figure \ref{fig:WD_add_col} for a graphical representation of the procedure. The new matrices $\mathbf{R}^{+}$ and $\mathbf{Q}^{+}$ are then obtained as
\begin{align*}
\mathbf{R}^{+} &=\mathbf{G}_k (k, k+1)^{\top} \cdots \mathbf{G}_k(N-1, N)^{\top} \widetilde{\mathbf{R}}, \\ \nonumber 
\mathbf{Q}^{+}&=\mathbf{Q} \mathbf{G}_k(N-1, N) \cdots \mathbf{G}_k (k, k+1). \nonumber 
\end{align*}
Then, we obtain $\mathbf{X}^{+} = \mathbf{Q}^{+} \mathbf{R}^{+}$. \par
Suppose column $k$ is deleted from matrix $\mathbf{X}$, so that the new matrix is $\mathbf{X}^{-}= \begin{bmatrix} \mathbf{X}\left[, \; 1 : (k-1) \right] & \mathbf{X}\left[,  \; (k+1) : p \right] \end{bmatrix}$, then the updated matrix $\mathbf{R}$ is obtained by setting to zero the elements on the diagonal of the last  $p-k$ columns and then removing column $k$, see Figure \ref{fig:WD_del_col_red}. 
Again, exploiting Givens rotations, new matrices $\mathbf{R}^{-}$ and $\mathbf{Q}^{-}$ are obtained as 
\begin{align}
\label{eq:delonecol}
\mathbf{R}^{-}&=\mathbf{G}_p(p-1, p)^{\top} \cdots \mathbf{G}_{k+1}(k, k+1)^{\top} \widetilde{\mathbf{R}}, \\ 
\mathbf{Q}^{-}&=\mathbf{Q} \mathbf{G}_{k+1}(k, k+1) \cdots \mathbf{G}_p(p-1, p).\nonumber
\end{align}
where $\widetilde{\mathbf{R}}$ is the matrix $\mathbf{R}$ with column $k$ removed. The update yields  $\mathbf{X}^{-} = \mathbf{Q}^{-} \mathbf{R}^{-}$.

\section{R updating algorithms}
\label{sec:thinqrupdate_algo}
%
In this section, we focus exclusively on updating matrix $\mathbf{R}$, rather than updating the entire QR decomposition. We anticipate that the computational cost of updating just one matrix will be lower than the cost of updating both. Furthermore, the elimination of the matrix $\mathbf{Q}$ results in a considerable reduction in storage requirements, given its size of $N \times N$. 
%
\subsection{Thin QR decomposition}
%
\noindent Let $\mathbf{X} \in \mathbb{R}^{N\times p}$ be a full column rank matrix, with $N \ge p$. Following the work of \cite{golub_van_loan.2013}, it can be shown that, for $k \in \left\{1, \dots, p\right\}$ it yields
\begin{equation*}
\text{span} \left\{\mathbf{x}_1, \dots, \mathbf{x}_k\right\} = \text{span} \left\{\mathbf{q}_1, \dots, \mathbf{q}_k\right\},
\end{equation*}
where $\mathbf{q}_1, \dots, \mathbf{q}_k$, are the first $k$ columns of the matrix $\mathbf{Q}$ from the QR decomposition of $\mathbf{X}$, and form an orthonormal basis for the same subspace. Here, $\text{span} \left\{S\right\}$ denotes the smallest linear subspace containing the set $S$. This result allows a reduced QR decomposition such that
\begin{equation*}
\begin{bmatrix} \mathbf{Q}_1 & \mathbf{Q}_2 \end{bmatrix}^\top \mathbf{X} = \begin{bmatrix} \mathbf{R}_1 \\ \mathbf{0}_{N-p,p} \end{bmatrix},
\end{equation*}
where $\mathbf{Q}_1 \in \mathbb{R}^{N \times p}$ is a matrix with orthonormal columns, $\mathbf{Q}_2 \in \mathbb{R}^{N \times \left(N-p\right)}$ and $\mathbf{R}_1 \in \mathbb{R}^{p \times p}$. It is straightforward to show that $\mathbf{X} = \mathbf{Q}_1 \mathbf{R}_1$. 
In some statistical applications only matrix $\mathbf{R}_1$ may be necessary. For example, quantity $\mathbf{X}^\top \mathbf{X}$, or its inverse, can be conveniently computed using matrix $\mathbf{R}_1$. In such contexts, consistent computational savings may be possible. The following section contains the details in case one row or column is added to or deleted from the matrix $\mathbf{X}$, while the extension to more rows or columns is detailed in Appendix \ref{secA1}. 
%
\subsection{Adding and deleting rows}
%
\noindent It is rather straightforward to see that when the update of the QR factorization is based only on entries of matrices $\mathbf{R}$ and $\mathbf{X}$, the update of matrix $\mathbf{R}_1$ is not problematic. This happens when an addition of rows or a deletion of columns is required. The other instances need matrix $\mathbf{Q}$, so alternative updating methods are necessary.\par
Consider the addition of one row in position $k$ and the QR update as described in equations \eqref{eq:addingonerow_1}--\eqref{eq:addingonerow_Q}. Since the Givens rotations are computed on the basis of the previous matrix $\mathbf{R}$ and the added row, then  an updated matrix $\mathbf{R}_1$ can be obtained as follows
\begin{equation}
\label{eq:addingonerow_R}
\begin{bmatrix}\mathbf{R}_1^{+} \\
\mathbf{0}^{\top}_p
\end{bmatrix}= \mathbf{G}_p \left(p, p+1\right)^\top\cdots \mathbf{G}_1\left(1, p+1\right)^\top\widetilde{\mathbf{R}}_1,
\end{equation}
where $\widetilde{\mathbf{R}}_1 = \begin{bmatrix} \mathbf{R}_1 \\ \mathbf{x}_{\star}^\top \end{bmatrix}$. 
The complete proof is given in Section \ref{appA} of the supplementary material.\par
The update of matrix $\mathbf{R}_1$ after the deletion of one row is more challenging as, in the full QR decomposition, it requires matrix $\mathbf{Q}$, see Section \ref{subsec:add_del_rows}.
An alternative method that avoids the computation of quantities related to matrix $\mathbf{Q}$ is based on equation \eqref{eq:addingonerow_R}, in which matrix $\mathbf{R}_1^{+}$ is now known and $\widetilde{\mathbf{R}}_1$ should be recovered. This corresponds to solving
\begin{equation*}
\begin{bmatrix}\mathbf{R}_1 \\
\mathbf{0}^{\top}_p
\end{bmatrix}= \mathbf{G}_p \left(p, p+1\right)^\top\cdots \mathbf{G}_1\left(1, p+1\right)^\top\begin{bmatrix}\mathbf{R}_{1}^{-}\\ \mathbf{x}_k^\top \end{bmatrix},
\end{equation*}
where $\mathbf{R}_1$ and $\mathbf{x}_{k}^\top$ are known. Entries of matrix $\mathbf{R}^{-}_1$ can thus be computed iteratively by applying Algorithm \ref{alg:thinQRdeleterow} in Section \ref{appB} of the supplementary material. This algorithm essentially reverses the approach for adding one row and requires an iterative procedure.\par
%
\subsection{Adding and deleting columns}
%
\noindent In this Section, we consider the case of updating matrix $\mathbf{R}_1$ after the addition or deletion of one column from $\mathbf{X}$, analogously to Section \ref{subsec:add_del_cols}. 
Assume that a column is added at the end of matrix $\mathbf{X}$, i.e. $\mathbf{X}^{+}= \begin{bmatrix} \mathbf{X} & \mathbf{x}_\star  \end{bmatrix}$ where $\mathbf{X}^{+} \in \mathbb{R}^{N \times (p+1)}$, then equation \eqref{eq:addonecol} becomes 
\begin{equation*}
\begin{bmatrix} \mathbf{Q}^\top_1 \\  \mathbf{Q}^\top_2 \end{bmatrix} \mathbf{X}^{+} = \begin{bmatrix} \mathbf{R}_1 & \mathbf{z}_{\star 1} \\   \mathbf{0}_{N-p, p} & \mathbf{z}_{\star 2} \end{bmatrix} = \widetilde{\mathbf{R}}^{+},
\end{equation*}
where $\mathbf{z}_{\star 1}= \mathbf{Q}^\top_{1} \mathbf{x}_{\star}$ and $\mathbf{z}_{\star 2}= \mathbf{Q}^\top_{2} \mathbf{x}_{\star}$. Matrix $\mathbf{R}_1^{+}$ can be obtained by setting to zero the last $N-p-1$ elements of the last column of $\widetilde{\mathbf{R}}^{+}$ through a sequence of Givens matrices. However, this procedure requires the evaluation of matrices $\mathbf{Q}_{1}$ and $\mathbf{Q}_{2}$. In order to avoid such computation, it is possible to exploit the relation $(\mathbf{X}^{+})^{\top} \mathbf{X}^{+}= (\mathbf{R}_{1}^{+})^{\top} \mathbf{R}_{1}^{+}$, so 
\begin{equation*}
\mathbf{R}^{+}_1 = \begin{bmatrix} \mathbf{R}_1 & \mathbf{R}^{-\top}_1 \mathbf{X}^\top \mathbf{x}_{\star} \\
0 & (\mathbf{x}^\top_{\star}\mathbf{x}_{\star} - \mathbf{z}_{\star 1}^{\top} \mathbf{z}_{\star 1})^{1/2} \end{bmatrix},
\end{equation*}
where $\mathbf{z}_{\star 1}=\mathbf{R}_1^{-\top}\mathbf{X}^\top\mathbf{x}_\star$, 
see Section \ref{appA} of the supplementary material.\par
We note that, while the update of the full QR factorization can be obtained when a column is added in any position of the matrix $\mathbf{X}$, the proposed update of the $\mathbf{R}_1$ matrix can be performed only when the new columns are added at the right end of the matrix. Although this may appear a limit of the proposed methodology, we have realized that in statistical applications of the QR factorization this is not an actual limitation, as it typically occurs that new columns can be added at the right end of the matrix.\par
The update of the QR factorization when deleting a column does not require knowledge of matrix $\mathbf{Q}$, hence the extension to the update of matrix $\mathbf{R}_{1}$ is quite straightforward. 
Let $\mathbf{X}^-$ be the reduced form of $\mathbf{X}$ after the deletion of column $k$, then $\widetilde{\mathbf{R}}_1 = \begin{bmatrix}\textbf{R}_1[, 1:(k-1)] & \mathbf{R}_1[, (k+1):p] \end{bmatrix}$ is the $\mathbf{R}_1$ matrix without column $k$. Updated matrix $\mathbf{R}_1^-$ can be obtained by setting to $0$ the elements on the sub-diagonal of the last $p-k$ columns of $\widetilde{\mathbf{R}}_1$. So, similarly to the full QR update, see equation \eqref{eq:delonecol}, this can be done by applying a set of Givens rotations as follows
\begin{equation*}
\begin{bmatrix}\mathbf{R}_1^{-} \\
\mathbf{0}^{\top}_{p-1}
\end{bmatrix} = \mathbf{G}_p(p-1, p)^\top  \cdots \mathbf{G}_{k+1}(k, k+1)^\top \widetilde{\mathbf{R}}_1.
\end{equation*}
%

\section{Computational costs}
\label{sec:computational_costs}
%
This section investigates the computational efficiency of several QR updating strategies by combining theoretical floating-point operations analyses with empirical timing experiments. Our goal is to quantify the gains obtained by updating existing decompositions—either the full QR factors or only the upper-triangular factor R—rather than recomputing them from scratch. We begin by summarizing the  computational costs associated with the various update and downdate operations, and then examine how these theoretical differences translate into practical runtime improvements across a broad range of matrix sizes and update scenarios. Additional experiments illustrating the computational benefits of these strategies in concrete statistical applications are presented in Section \ref{sec:simulations}.\par
Table \ref{tab:QRcosts} reports the leading term of the number of floating-point operations (FLOPS) \cite[see,][Ch. 1]{golub_van_loan.2013} required to add or remove rows or columns using QR-based and R-based updating algorithms, respectively.
Proofs for these expressions and the exact computational costs are given in Section \ref{appC} of the supplementary material. 
A comparison of the dominant terms shows that updating the R factor is significantly more economical than updating the full QR decomposition. For example, adding a single row requires 
$6Np$ operations for a QR update but only $3p^2$ for an R update---a substantial reduction, since $N \ge p$. Removing rows displays an even larger disparity. Indeed, the cost for the QR update increases, whereas the R update remains essentially unchanged. Similar gains hold when adding or removing blocks of $m>1$ rows. Notably, the complexity of adding $m\geq 2$ rows sequentially, is $\mathcal{O}(mNp)$, with a dominant term of $6mNp$, which exceeds the complexity of the algorithm for adding $m$ rows in bulk, $4mNp$. Moreover, the cost of column removal depends on the position of the updated column for both approaches, with deletions near the right boundary being the least expensive.
We note, however, that the addition of columns is allowed only on the rightmost edge of the matrix $\mathbf{X}$ when employing the R updates.
%
\begin{table}[!t]
\caption{Arithmetic operations (FLOPS) required by QR and R updating algorithms for performing addition or deletion of rows or columns of an $N \times p$ matrix.}
\centering
\begin{tabular}{p{4.6cm} >{\centering\arraybackslash}p{2.6cm} >{\centering\arraybackslash}p{2.6cm}}
\toprule
Description & \multicolumn{2}{c}{Most relevant term} 	\\
\cmidrule{2-3}
& QR & R\\
\midrule
add $1$ row	&		$6Np$ & 	$3p^2$		\\
add the $k$-th column	&	$8N^2$ & --		\\
add the $(p+1)$-th column	& $8N^2$ &	$2Np$	\\
remove $1$ row	&	$6N^2$	& $3p^2$		\\
remove the $k$-th column	&	$6N(p-k)$ &	$3(p-k)^2$	\\												
add $m$ rows	&	$4mNp$	&	$2mp^2$	\\
add $m$ columns	&	$8mN^2$	&	$2mNp$	\\
remove $m$ rows	&	$6mN^2$	&	$2mp^2$	\\
remove $m$ columns	&	$4mNp$  &	$2mp^2$	\\					
\bottomrule
\end{tabular}
\label{tab:QRcosts}
\end{table}
%
It is important to emphasize that the correct baseline for evaluating the improvement offered by the proposed method is the cost of updating both
$\mathbf{Q}$ and $\mathbf{R}$, as reported in Table~\ref{tab:QRcosts}. From the perspective of this article, any fair comparison must account for the cost of maintaining the entire QR factorization under repeated modifications of the design matrix. To illustrate this point, consider the simple case of appending a column to the end of $\mathbf{X}$ (see Algorithm~\ref{alg:qr_add_one_col} in Section~\ref{appB} and Proposition~\ref{prop:qrupdate_add_one_col} in Section~\ref{appC} of the supplementary material). Appending the vector $\mathbf{x}$ yields 
$\mathbf{X}^{+}=\begin{bmatrix}\mathbf{X}&\mathbf{x}\end{bmatrix}$, and updating the triangular factor 
$\mathbf{R}$ requires computing $\mathbf{R}^{+}=\begin{bmatrix}\mathbf{R} & \mathbf{Q}^\top \mathbf{x}\end{bmatrix}$. Although this operation is algebraically simple, it depends explicitly on the orthogonal factor $\mathbf{Q}$, whose update cost cannot be ignored in an iterative setting where the matrix $\mathbf{X}$ is modified repeatedly. 
Table~\ref{tab:QRcosts} reports a baseline cost of $\mathcal{O}\left(N^2\right)$ for multiplying by the $N \times N$ matrix $\mathbf{Q}$, corresponding to a dense-matrix representation, although in practice $\mathbf{Q}$ is stored and applied implicitly via Givens rotations or Householder reflections. From a single-factorization perspective, an implicit representation of $\mathbf{Q}$ allows $\mathbf{Q}^{\top}$ to be applied to a vector at a reduced cost, of order $\mathcal{O}(pN-p^2/2)$, regardless of whether the orthogonal factor is expressed through Givens rotations or Householder reflections.
However, this perspective does not reflect the computational setting considered here. The reduced cost $\mathcal{O}\left(pN-p^2 / 2\right)$ assumes that the implicit representation of $\mathbf{Q}$ is already available and can be applied directly. In contrast, our focus is on sequentially maintaining a valid QR factorization under repeated modifications of $\mathbf{X}$. From this point of view, the representation of $\mathbf{Q}$ itself must be updated consistently at each step, requiring the accumulation, storage, and application of additional orthogonal transformations. \par
\begin{figure}[!t]
\begin{center}
\subfigure{
\includegraphics[width=0.45\textwidth]{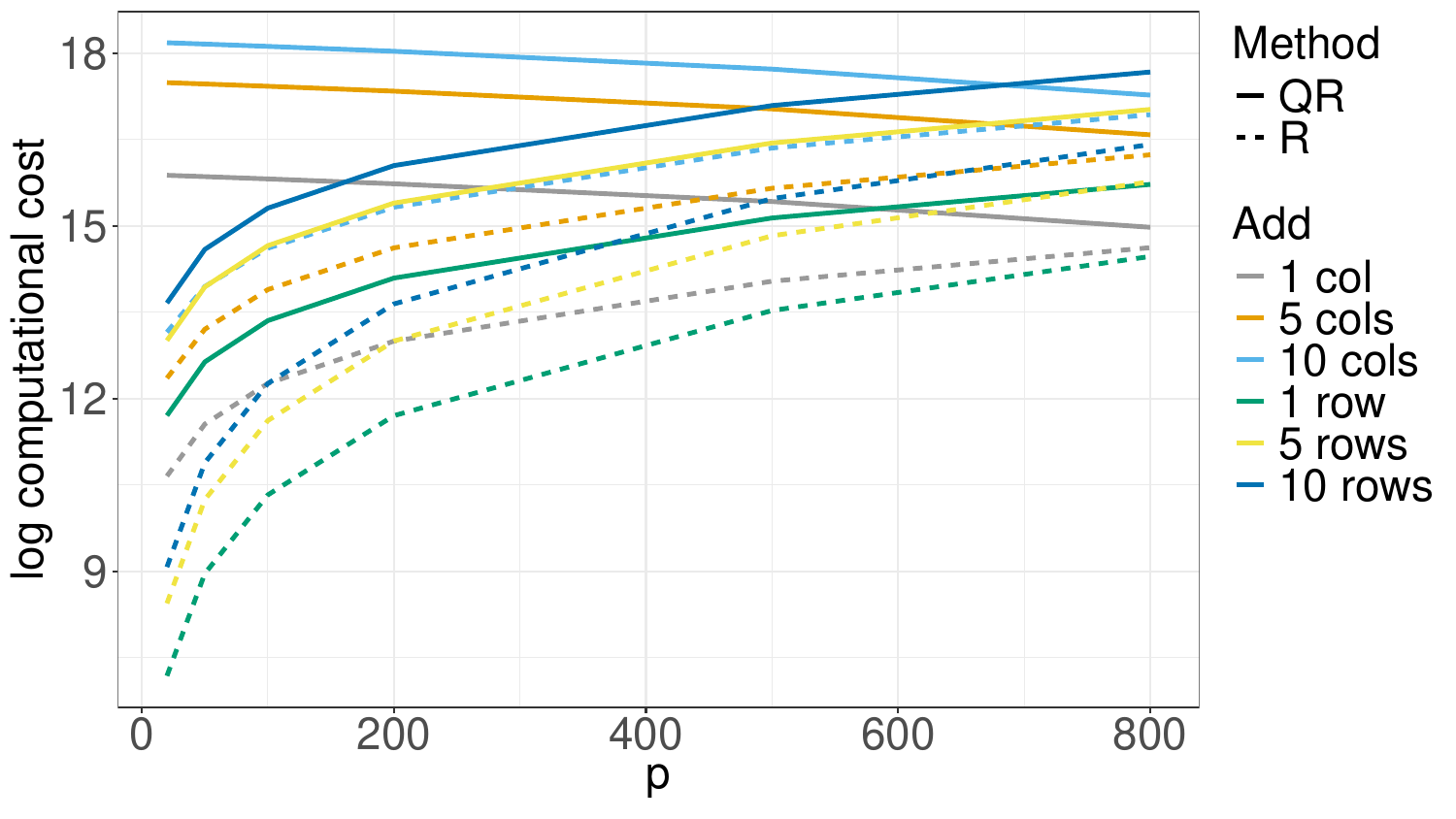}}\qquad
%
\subfigure{
\includegraphics[width=0.45\textwidth]{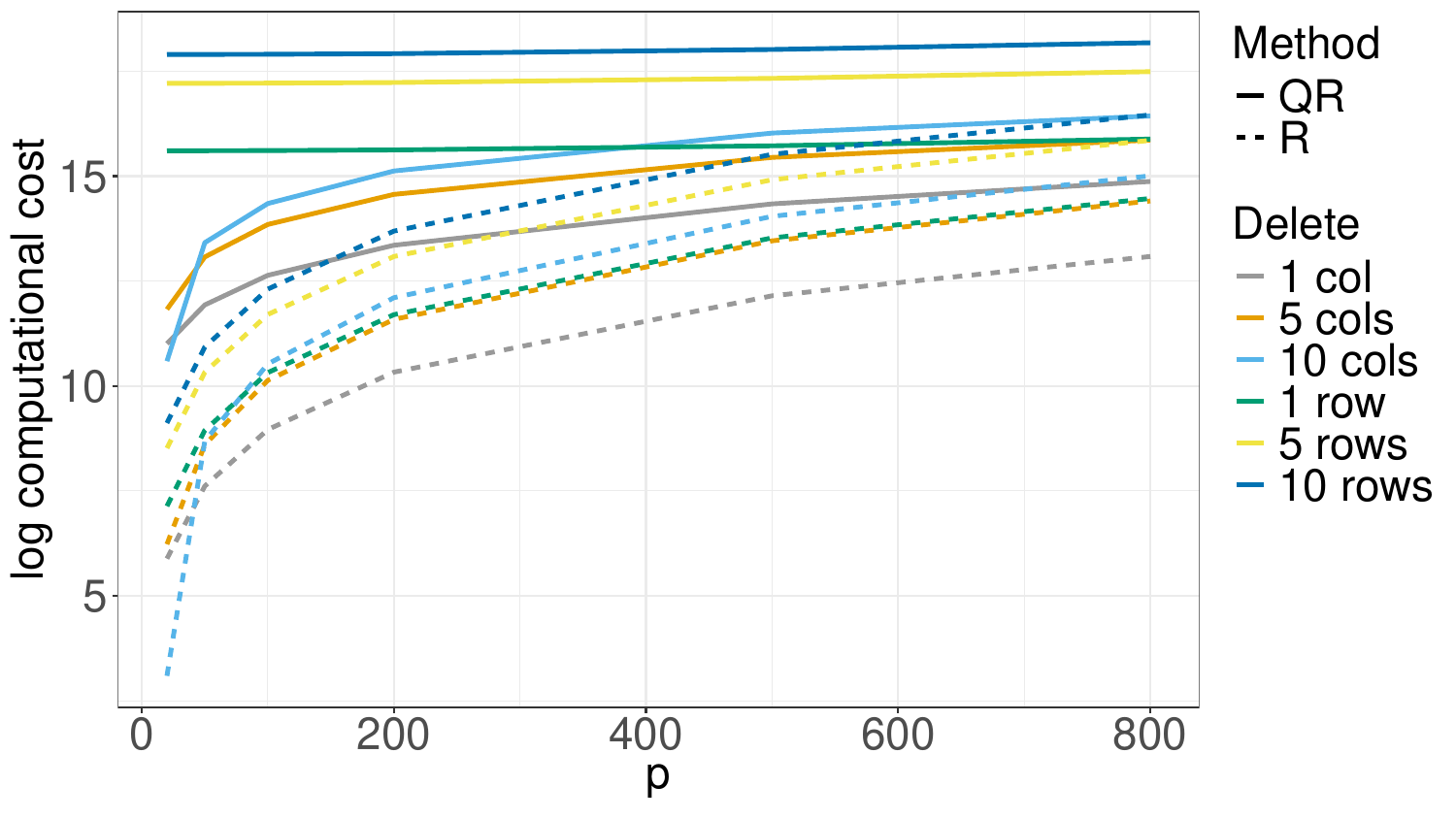}}\\
%
\subfigure{
\includegraphics[width=0.45\textwidth]{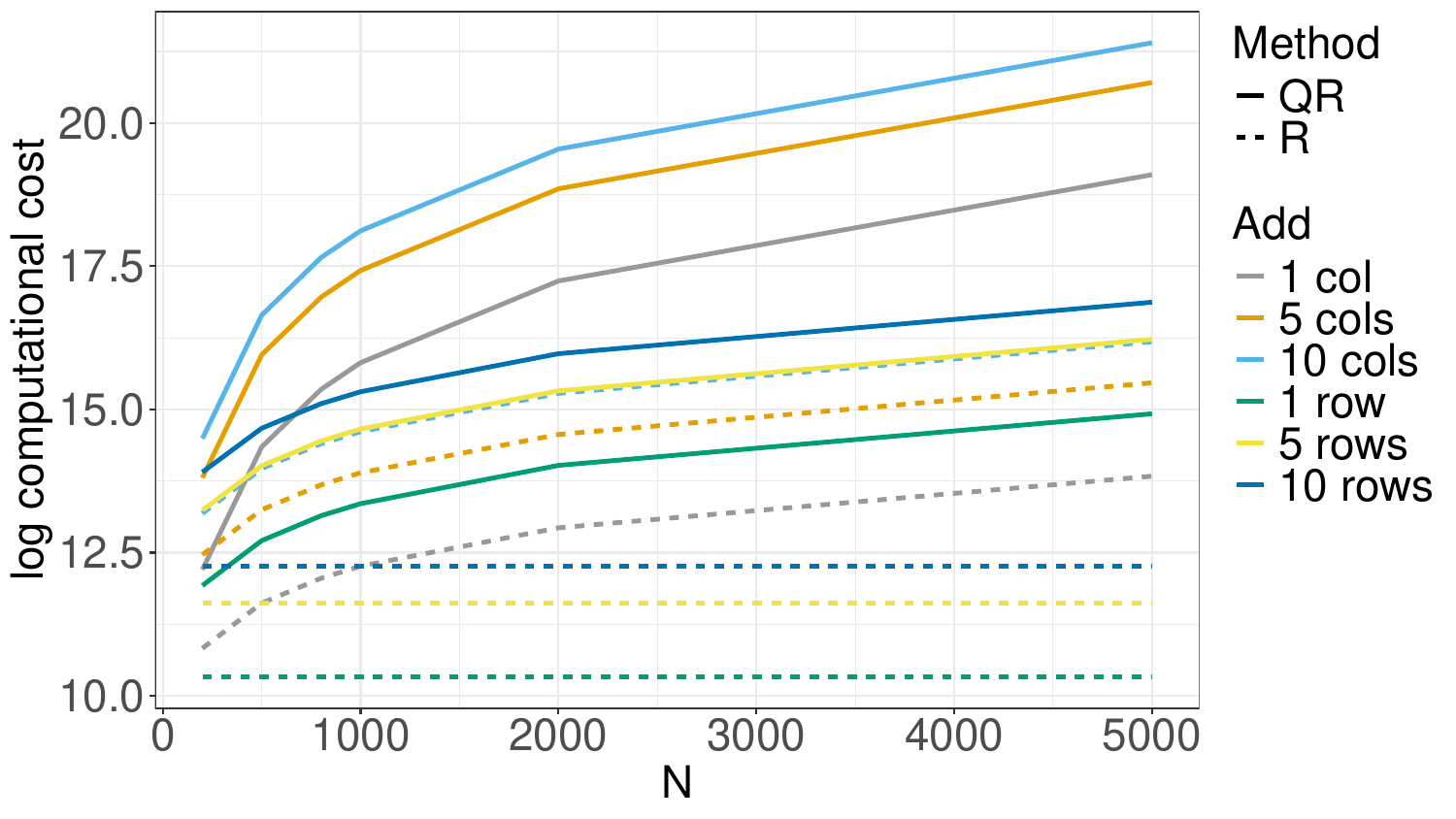}}\qquad
\subfigure{
\includegraphics[width=0.45\textwidth]{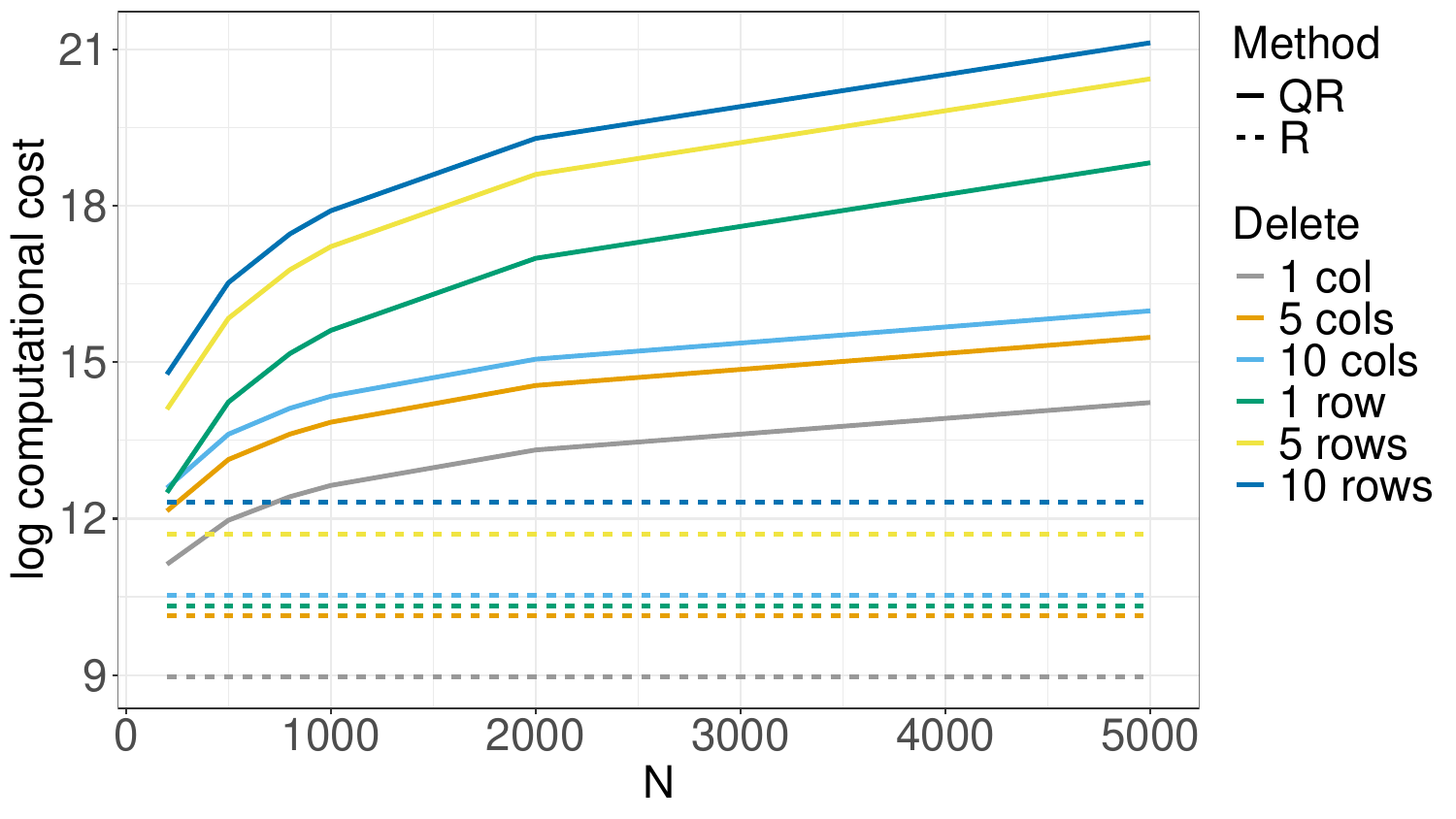}}
\caption{Logarithm of the exact computational costs of adding (left panels) or deleting (right panels) 1, 5 or 10 columns or rows. Top row: $N=1000$ and $p \in \{20, 50, 100, 200, 500, 800\}$. Bottom row: $p=100$ and $N \in \{200, 500, 800, 1000, 2000, 5000\}$.}
\label{fig:computational_costs}
\end{center}
\end{figure}
%
Figure \ref{fig:computational_costs} visualizes the exact FLOPS counts on a logarithmic scale, illustrating how the costs evolve with $N$, $p$, (with $N>p$ by construction) and the number of rows or columns modified. 
We consider the addition or deletion of $1$, $5$, or $10$ rows or columns. The costs are calculated under the assumption that columns are appended to the rightmost side of the matrix and that, when columns are deleted, they are removed from the central part of the matrix. We adopt this convention because, in the practical applications discussed in Sections \ref{sec:simulations} and \ref{sec:applications}, the columns represent specific statistical variables. Although new columns can be added in any position without affecting interpretation, any deletions must target columns that correspond precisely to the variables to be removed.
When columns are appended on the right, the QR update cost decreases as 
$p$ grows because fewer elements must be zeroed, eventually approaching the cost of the R update. Notably, R-based updates depend on $N$ only when adding columns. Across all scenarios, updating the R factor is consistently less expensive than its QR-based counterpart. It is important to note that these observations focus purely on computational costs and do not incorporate memory advantages associated with avoiding explicit storage of the Q factor. Sections \ref{appB} and \ref{appC} of the supplementary material also detail the algorithms and costs for updating the QR and R matrices when deleting $m$ non-adjacent columns.\par
%
\begin{table}[t]
\caption{Execution time (in seconds per 10 replications) for removing 10 rows or columns from the matrix $\mathbf{X}\in\mathbb{R}^{N\times p}$, stratified by matrix dimensions $(N,p)$. Times for reconstructing the upper-triangular factor $\mathbf{R}$ using the implementations in the R packages \texttt{base}, \texttt{Matrix}, and \texttt{fastQR}, or downdating the full QR decomposition (\texttt{QR downdate}) or the $\mathbf{R}$ factor alone (\texttt{R downdate}), both provided by the \texttt{fastQR} package.}
\setlength{\tabcolsep}{8pt}   
\begin{tabular}{lrrrrrr}\toprule
{$N$} & {$p$} & {\texttt{base}} & {\texttt{Matrix}} & {\texttt{fastQR}} & {\texttt{QR downdate}} & {\texttt{R downdate}} \\
\midrule
\multicolumn{7}{c}{\textit{Rows}}\\
10000 & 1000 & 50.38 & 50.88 & 20.43 & 61.37 & 0.14 \\ 
10000 & 3000 & 446.89 & 441.03 & 161.54 & 69.69 & 1.46 \\ 
10000 & 5000 & 1178.04 & 1155.34 & 415.52 & 83.81 & 4.06 \\ 
15000 & 1000 & 80.32 & 80.72 & 31.44 & 134.90 & 0.14 \\ 
15000 & 3000 & 693.65 & 685.86 & 248.54 & 145.13 & 1.46 \\ 
15000 & 5000 & 1867.08 & 1825.81 & 651.65 & 177.36 & 4.55 \\ 
20000 & 1000 & 108.39 & 107.80 & 42.79 & 245.56 & 0.14 \\ 
20000 & 3000 & 956.05 & 958.59 & 341.12 & 271.42 & 1.45 \\ 
20000 & 5000 & 2555.47 & 2429.76 & 779.00 & 236.74 & 3.70 \\ 
\midrule
\multicolumn{7}{c}{\textit{Columns}}\\
10000 & 1000 & 44.19 & 43.25 & 19.02 & 18.15 & 0.12 \\ 
10000 & 3000 & 392.36 & 391.44 & 144.67 & 20.02 & 1.39 \\ 
10000 & 5000 & 1146.06 & 1143.46 & 415.26 & 33.78 & 4.30 \\ 
15000 & 1000 & 74.48 & 74.54 & 30.31 & 25.12 & 0.13 \\ 
15000 & 3000 & 683.77 & 696.93 & 245.01 & 35.58 & 1.37 \\ 
15000 & 5000 & 1858.47 & 1836.10 & 662.18 & 53.52 & 4.25 \\ 
20000 & 1000 & 100.42 & 101.61 & 41.41 & 55.58 & 0.14 \\ 
20000 & 3000 & 933.94 & 940.96 & 340.65 & 59.85 & 1.41 \\ 
20000 & 5000 & 2508.69 & 2546.73 & 909.69 & 82.68 & 4.24 \\ 
\bottomrule
\end{tabular}
\label{tab:QR_downdate_rows_columns}
\end{table}
%
%
To complement the theoretical analysis, we benchmark several implementations across a range of matrix dimensions, block sizes, and update positions. This empirical study is essential not only for quantifying the practical behavior of the proposed update-based routines but also for situating them relative to standard QR implementations. Comparing against widely used baseline routines is necessary to assess whether the methodological gains translate into meaningful computational advantages in real settings. It also highlights additional practical costs---such as storage requirements---and clarifies the absolute computational burden of each approach. Relative improvements alone can be misleading: saving a large percentage of an already inexpensive operation may matter less than a modest improvement on a computationally heavy one. Understanding the actual magnitudes involved is therefore crucial for evaluating efficiency in practice.
The experiments compare the following methods:
\begin{itemize}
\item[{\it (i)}] full recomputation of the upper-triangular factor $\mathbf{R}$ matrix using the \texttt{base}, \texttt{Matrix}, and \texttt{fastQR} packages;
\item[{\it (ii)}] update-based routines for the full QR decomposition (\texttt{updateQR}), leveraging the \texttt{fastQR} package;
\item[{\it (iii)}] update-based routines for the upper-triangular factor alone (\texttt{updateR}), leveraging the \texttt{fastQR} package.
\end{itemize}
Tables \ref{tab:QR_downdate_rows_columns} and \ref{tab:QR_update_rows_columns} report runtimes for deleting and adding blocks of $10$ columns, averaged over $10$ repetitions and stratified by matrix size. Results are shown for update positions starting at $k=\lfloor p/4 \rfloor$ for columns and $k=\lfloor N/4 \rfloor$ for rows, with analogous behavior observed at other positions (e.g., $k=\lfloor p/2 \rfloor$, $k=\lfloor N/2 \rfloor$, $k=p-10$, or $k=N-10$). The empirical findings closely follow the theoretical expectations: full recomputation via \texttt{base} and \texttt{Matrix} is consistently the slowest, while update-based routines provide substantial speedups. Among these, \texttt{updateR} achieves the largest gains—often by orders of magnitude for large matrices—whereas the \texttt{fastQR} full decomposition remains an efficient choice when both Q and R are required.\par
%
\begin{table}[t]
\caption{Execution time (in seconds per 10 replications) for adding 10 rows or columns to the $\mathbf{X}\in\mathbb{R}^{N\times p}$ matrix, stratified by matrix dimensions $(N,p)$. 
Time for reconstruction of the upper-triangular factor $\mathbf{R}$ using the implementations provided by the R packages \texttt{base}, \texttt{Matrix}, and \texttt{fastQR}, or updating the full QR decomposition (\texttt{QR update}) or the $\mathbf{R}$ factor alone (\texttt{R update}), both provided by the \texttt{fastQR} package.}
\setlength{\tabcolsep}{8pt}   
\begin{tabular}{lrrrrrr}\toprule
{$N$} & {$p$} & {\texttt{base}} & {\texttt{Matrix}} & {\texttt{fastQR}} & {\texttt{QR update}} & {\texttt{R update}} \\
\midrule
\multicolumn{7}{c}{\it Rows}\\
10000 & 1000 & 41.37 & 39.50 & 17.34 & 16.15 & 0.71 \\ 
10000 & 3000 & 407.52 & 403.94 & 137.38 & 23.52 & 3.55 \\ 
10000 & 5000 & 1051.23 & 1051.88 & 348.85 & 31.61 & 7.63 \\ 
15000 & 1000 & 71.72 & 72.27 & 27.23 & 37.49 & 1.02 \\ 
15000 & 3000 & 635.29 & 645.79 & 213.09 & 53.83 & 4.59 \\ 
15000 & 5000 & 1680.68 & 1683.92 & 551.67 & 70.89 & 9.18 \\ 
20000 & 1000 & 95.18 & 95.24 & 37.15 & 64.62 & 1.56 \\ 
20000 & 3000 & 865.46 & 872.33 & 291.54 & 98.26 & 5.66 \\ 
20000 & 5000 & 2380.73 & 2341.71 & 764.84 & 131.87 & 11.15 \\ 
\midrule
\multicolumn{7}{c}{\it Columns}\\
10000 & 1000 & 52.39 & 55.80 & 22.24 & 62.05 & 0.63 \\ 
10000 & 3000 & 440.22 & 438.65 & 139.05 & 50.00 & 2.84 \\ 
10000 & 5000 & 1007.13 & 1005.63 & 330.20 & 42.34 & 5.52 \\ 
15000 & 1000 & 67.26 & 66.79 & 26.56 & 113.21 & 0.88 \\ 
15000 & 3000 & 610.13 & 607.19 & 203.19 & 106.03 & 3.56 \\ 
15000 & 5000 & 1617.44 & 1618.04 & 526.47 & 99.68 & 6.90 \\ 
20000 & 1000 & 91.73 & 93.26 & 36.56 & 206.98 & 1.15 \\ 
20000 & 3000 & 832.70 & 833.86 & 286.97 & 200.26 & 4.16 \\ 
20000 & 5000 & 2217.56 & 2217.70 & 724.93 & 183.86 & 8.38 \\ 
\bottomrule
\end{tabular}
\label{tab:QR_update_rows_columns}
\end{table}
%

\section{Simulation studies}
\label{sec:simulations}
%
In this section, we explore the advantages of using the R update in statistical applications. A particularly relevant area where QR decomposition methods are applied is regression analysis, where a continuous or discrete response variable $y$ is regressed on a set of covariates $\mathbf{x}\in\mathbb{R}^p$. QR methods are especially useful when comparing two or more model specifications that differ in the set of covariates included. Such comparisons are often made using techniques like automatic information criteria, likelihood ratio tests, Bayesian selection  procedures and penalized methods such as the LASSO and its extensions. Although both frequentist and Bayesian approaches can, in theory, benefit from QR updating methods, we focus on Bayesian model selection based on the spike-and-slab prior introduced by \cite{george_mcculloch.1993}  for the linear regression model with a continuous outcome. Specifically, we consider the following Gaussian univariate linear regression model for the continuous outcome $\mathbf{y}=(y_1, \ldots, y_n)^\top\in\mathbb{R}^n$:
\begin{equation}
\label{eq:gaussian_reg_model_def}
\mathbf{y}=\mathbf{X}\boldsymbol{\beta}+\boldsymbol{\varepsilon},\qquad\boldsymbol{\varepsilon}\sim\mathrm{N}\big(0,\sigma^2\mathbf{I}_n\big),
\end{equation}
where $\mathbf{X}\in\mathbb{R}^{n\times p}$, denotes the matrix of observations for the $p$ covariates and $\boldsymbol{\beta}=\big(\beta_1,\dots,\beta_p\big)^{\top}\in\mathbb{R}^p$ is the vector of regression coefficients. To induce sparse solutions we assume a Dirac spike-and-slab prior \citep[][]{george_mcculloch.1993} that relies on an auxiliary latent $p$-dimensional selection vector $\boldsymbol{\gamma}=\big(\gamma_1,\dots,\gamma_p\big)^\top$, where $\gamma_j=1$, $j=2,\dots,p$, when the $j$-th regressor is included in the model, $0$ otherwise. We assume that the intercept is always included in the model, hence $\gamma_1=1$. Moreover, we assume the following general hierarchical Dirac spike-and-slab prior distribution
\begin{align}
\label{eq:spike_slab_prior_general_def_1}
\pi\big(\boldsymbol{\beta}_{\gamma}\vert \boldsymbol{\gamma},\sigma^2\big)&\sim\phi_{p_\gamma}\big(\boldsymbol{\beta}_{\gamma}\vert  0,\sigma^2\boldsymbol{\Sigma}_{\beta_{\gamma}}\big),\qquad
\pi\big(\boldsymbol{\beta}_{-\gamma}\vert \boldsymbol{\gamma}\big)=\prod_{j=1}^p\delta\big(\beta_j,0\big)^{1-\gamma_j}, \\
\gamma_j&\sim\mathrm{Bern}(\theta),\qquad j=2,\dots,p, \nonumber
\end{align}
where $\phi_{p}(\cdot)$ and $\mathrm{Bern}(\cdot)$ denote the $p$-dimensional Gaussian  and the Bernoulli distributions, respectively, $\delta\big(x,0\big)=\mathbbm{1}_{x=0}$ denotes the Dirac function evaluated at zero, $p_\gamma=\sum_{j=1}^p \gamma_j$ denotes the number of covariates included in the regression model, $\boldsymbol{\beta}_{\gamma}\in\mathbb{R}^{p_\gamma}$ denotes the  vector consisting of all elements $\beta_j$ of $\boldsymbol{\beta}$ for which $\gamma_j=1$ for $j=1,\dots,p$, and $\boldsymbol{\beta}_{-\gamma}$ is such that $\boldsymbol{\beta}=\boldsymbol{\beta}_{\gamma}\cup\boldsymbol{\beta}_{-\gamma}$ with $\boldsymbol{\beta}_{\gamma}\cap\boldsymbol{\beta}_{-\gamma}=\emptyset$, $\boldsymbol{\Sigma}_{\beta_{\gamma}}\in\mathbb{S}_{++}^{p_\gamma}$ is a symmetric and positive-definite variance-covariance matrix and $\theta\in\big(0,1\big)$. Several alternative specification for $\boldsymbol{\Sigma}_{\beta_{\gamma}}$ are possible, we assume $\boldsymbol{\Sigma}_{\beta_{\gamma}}=\upsilon_0\mathbf{I}_{p_\gamma}$ as in \cite{george_mcculloch.1997}. Independent Bernoulli priors on the $\gamma_j$'s with a Beta hyperprior, $\theta\sim\mathrm{Be}\big(\xi,\varphi\big)$, with $\xi>0$, $\varphi>0$, are used for example by \cite{brown_etal.1998}. An attractive feature of these priors is that appropriate choices of $\theta$ that depend on the number of covariates $p$ impose an a priori multiplicity penalty, as argued in \cite{scott_berger.2010}.\par
Reversible jump \citep{green.1995} algorithms for model selection are Metropolis-type methods that require the simulation of the model indicator $\boldsymbol{\gamma}$ from its posterior distribution $m(\boldsymbol{\gamma}\vert  \mathbf{y},\mathbf{X}^{\gamma})$ which, assuming an Inverse Gamma prior for the scale parameter $\sigma^2$, i.e. $\sigma^2\sim\mathrm{IG}(\nu,\lambda)$, is proportional to
\begin{equation}
\label{eq:marginal_posterior}
m(\boldsymbol{\gamma}\vert  \mathbf{y},\mathbf{X}^{\gamma}) \propto\ell(\boldsymbol{\gamma}\vert  \mathbf{y},\mathbf{X}^{\gamma})\pi(\boldsymbol{\gamma}), \nonumber
\end{equation}
where
\begin{align}
\ell(\boldsymbol{\gamma}\vert  \mathbf{y},\mathbf{X}^{\gamma}) &\propto \vert(\widehat{\boldsymbol{\Sigma}}_{\gamma})^{-1}\vert^{-1/2} \vert\boldsymbol{\Sigma}_{\beta_\gamma}\vert^{-1/2} \bigg(\lambda + \frac{S^2_{\gamma}}{2}\bigg)^{-(\nu + n/2)}\nonumber\\
\pi(\boldsymbol{\gamma})&={p \choose p_{\gamma}}\theta^{p_{\gamma}}\big(1-\theta\big)^{p-p_{\gamma}}\nonumber,
\end{align}
with $\widehat{\boldsymbol{\Sigma}}_{\gamma}=\big((\widetilde{\mathbf{X}}^\gamma)^{\top}\widetilde{\mathbf{X}}^\gamma\big)^{-1}=\big((\mathbf{X}^{\gamma})^{\top}\mathbf{X}^{\gamma}+\boldsymbol{\Sigma}_{\beta_\gamma}^{-1}\big)^{-1}$, $S^2_{\gamma} = \mathbf{y}^{\top}\mathbf{y} - \mathbf{y}^{\top} \mathbf{X}^{\gamma} \big( (\mathbf{X}^{\gamma})^{\top} \mathbf{X}^{\gamma} + \boldsymbol{\Sigma}_{\beta_\gamma}^{-1} \big)^{-1} (\mathbf{X}^{\gamma})^{\top} \mathbf{y}$, $\mathbf{X}^{\gamma} \in \mathbb{R}^{n\times p_{\gamma}}$ is the $n\times p_\gamma$ matrix whose columns correspond to the components of $\boldsymbol{\beta}_\gamma$.
The R updating methods presented in this paper enable efficient updates to $\widetilde{\mathbf{X}}^\gamma$ as the model indicator shifts from $\boldsymbol{\gamma}$ to $\boldsymbol{\gamma}^\star$ accommodating the addition or removal of covariates in the regression model.\par
%

\begin{figure}[t]
\begin{center}
\includegraphics[width=0.95\textwidth]{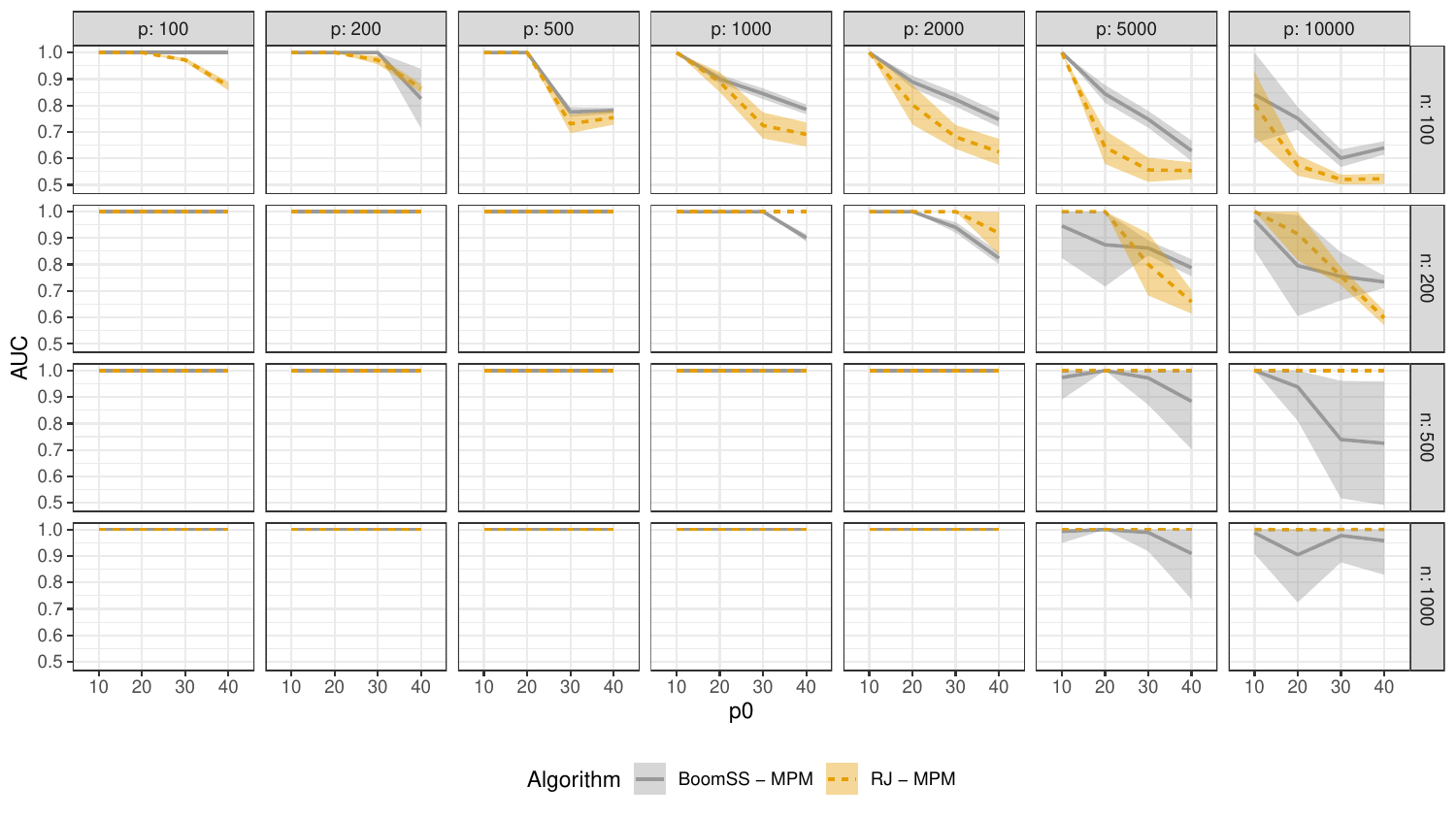}
\caption{Mean AUC with $1$ standard error bands of the MPM computed by the RJ (with R update) and the BoomSS algorithms. 40 repetitions for each setting with independent covariates.}
\label{fig:auc_Sim1}
\end{center}
\end{figure}
\begin{figure}[h]
\begin{center}
\includegraphics[width=0.95\textwidth]{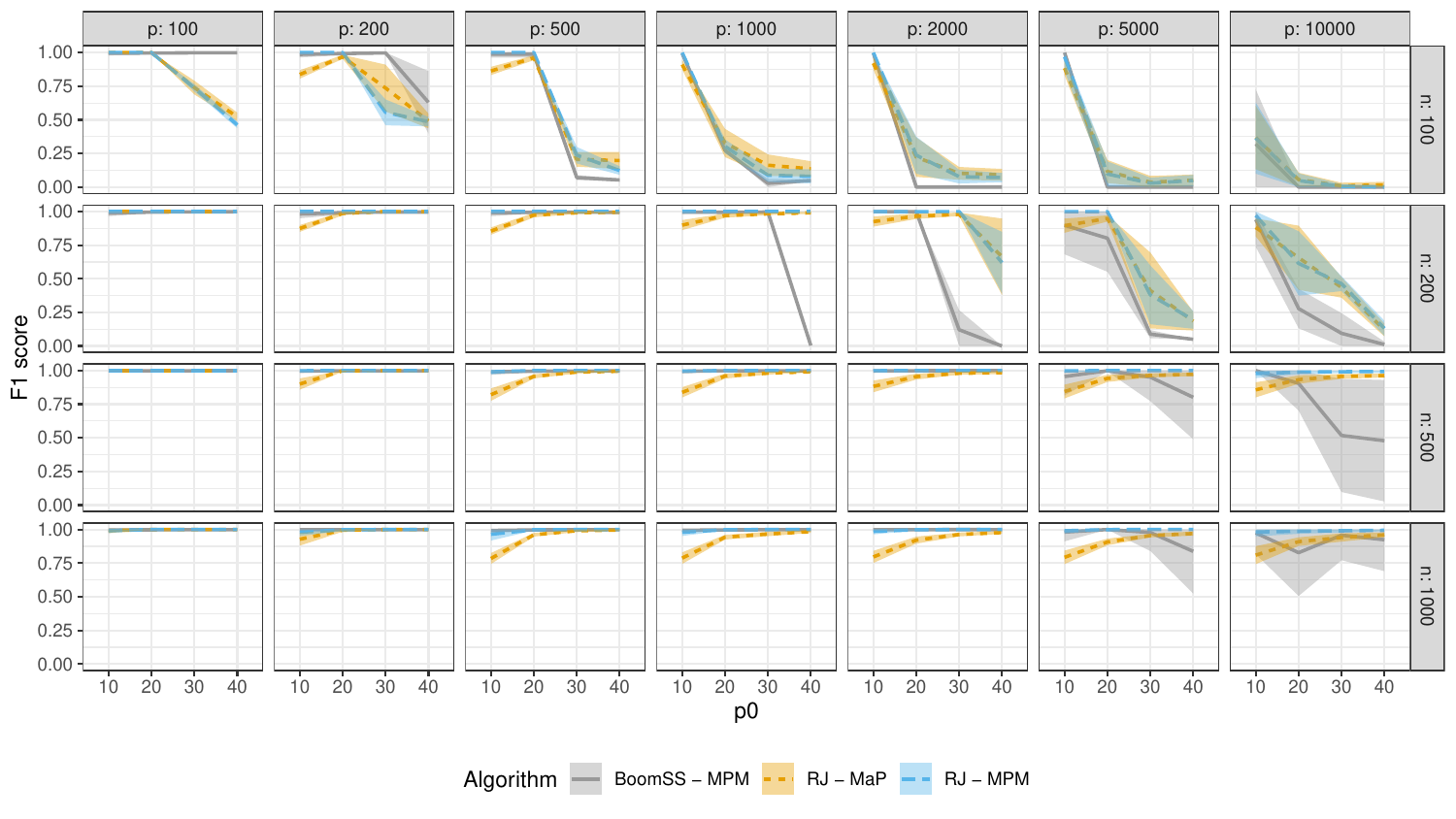}
\caption{Mean F1 score with 1 standard errors bands of the MPM and the MaP model computed by the RJ (with R update) and the BoomSS algorithms. 40 repetitions for each setting with independent covariates.}
\label{fig:F1_Sim1}
\end{center}
\end{figure}
%
%
\begin{figure}[h]
\begin{center}
\includegraphics[width=0.95\textwidth]{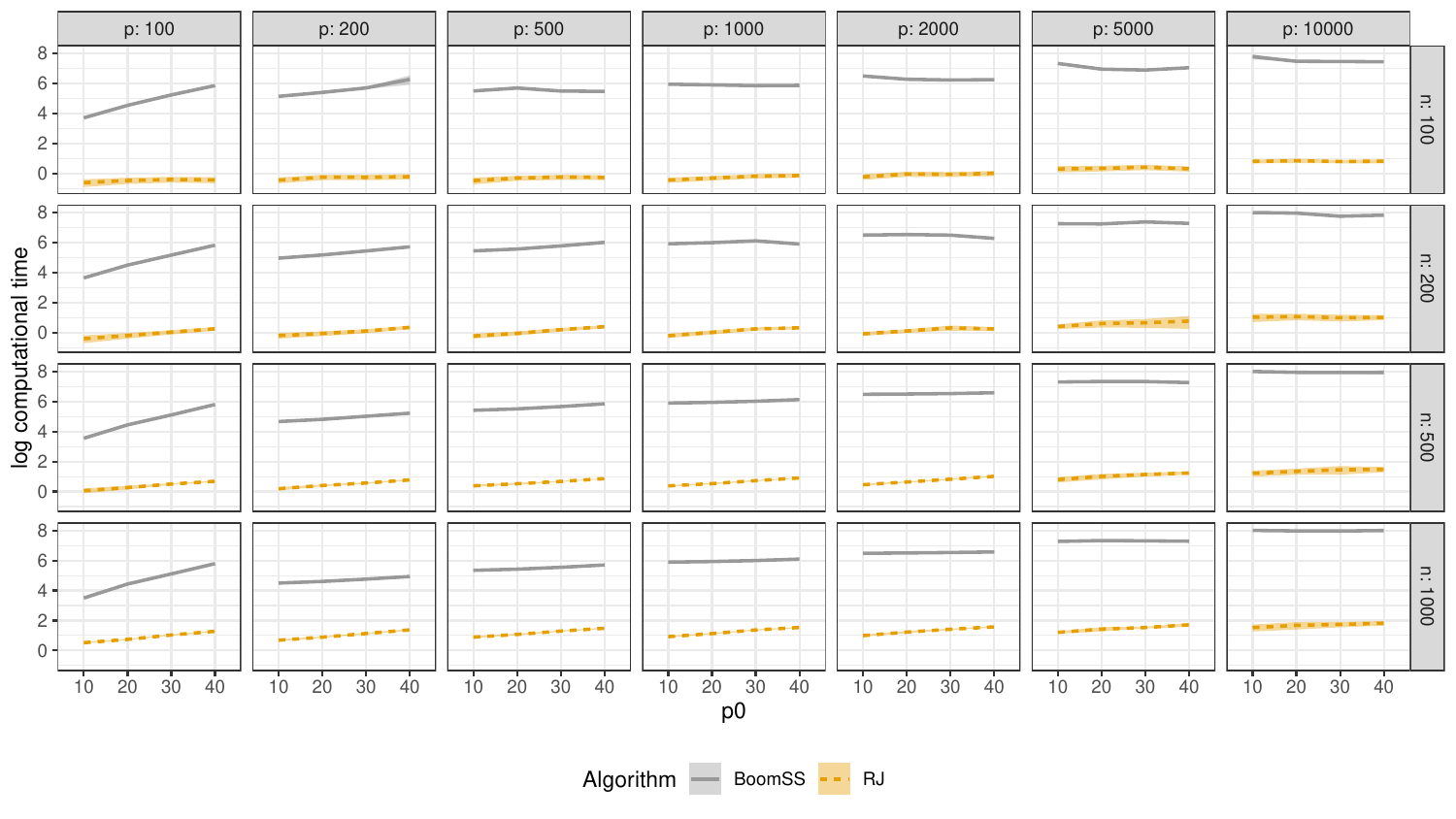}
\caption{Logarithm of the mean computational time (in seconds) with $1$ standard error bands for $50,000$ draws from the RJ (with R update) and the BoomSS algorithms. 40 repetitions for each setting with independent covariates.}
\label{fig:comp_time}
\end{center}
\end{figure}
In what follows, we consider the simulation setting of \cite{johnson_rossell.2012}. The response vector is generated as in equation \eqref{eq:gaussian_reg_model_def}, where we consider $\sigma^2 =1$ and generate the components of the design matrix $\mathbf{X}$ from a multivariate Normal distribution. In each simulation, the design matrix $\mathbf{X} = (\boldsymbol{\iota}_n, \mathbf{X}_1)$ with $\mathbf{X}_1 \sim\mathrm{N}_{p-1}(0, \boldsymbol{\Sigma}_X)$ and $\boldsymbol{\Sigma}_X$ is a $(p-1) \times (p-1)$ matrix. We consider three cases, the first study assumes independent covariates, $\boldsymbol{\Sigma}_X =\mathbf{I}_{p-1}$, the second assumes equicorrelated covariates, $\boldsymbol{\Sigma}_X=\rho\big(\boldsymbol{\iota}_{p-1}\boldsymbol{\iota}_{p-1}^\top-\mathbf{I}_{p-1}\big)+\mathbf{I}_{p-1}$, with $\rho=0.50$, while the third one assumes decreasing correlation $(\boldsymbol{\Sigma})_{i,j} = \rho^{|i-j|}$ with $\rho=0.50$. In addition, the number of predictors is $p = \{100,200,500,1000,2000,5000,10000\}$, the number of observations is $n=\{100,200,500,1000\}$ and the $p$-dimensional vector  $\boldsymbol{\beta}$ is defined as $\boldsymbol{\beta} = (\boldsymbol{\beta}_0^\top,\mathbf{0}_{p-p_0}^\top)^\top$, where $p_0=\{10,20,30,40\}$ denotes the number of non-zero coefficients and
$\boldsymbol{\beta}_0=(-1)^\mathbf{u} \Big(\frac{5\log(n)}{\sqrt{n}}+\vert\mathbf{z}\vert\Big)$, 
with $\mathbf{z}\sim\mathrm{N}_{p_0}(0,\mathbf{I}_{p_0})$ and $\mathbf{u}\sim\mathrm{Bin}(p_0,0.4)$. The details concerning the hyper-parameters $(\xi,\varphi,\nu,\lambda,\upsilon_0)$ are reported in 
Section \ref{app:prior_hyperparameters} of the supplementary material.
We implement the reversible jump algorithm with updates based on R updating algorithms (RJ) and compare the results with an alternative approach, that is the Rao-Blackwellized SSVS method by \cite{ghosh_clyde.2011} (implemented in the R package ``\texttt{BoomSpikeSlab}'' of \cite{scott.2023}, or BoomSS).
For each setting, 40 repetitions are performed and Figures \ref{fig:auc_Sim1} and \ref{fig:F1_Sim1} report the mean, with one standard error bands, of the AUC and F1 score computed for the median probability model (MPM), i.e. the model which includes all covariates with marginal inclusion probability larger than 0.5. It is evident that, as $p$ and $p_{0}$ increase, both methods have difficulties in recovering the true model. 
A comparison of the two performances reveals a certain degree of similarity, whilst it is evident that the RJ algorithm is more reliable when the sample size is greater than $200$ and $p$ is very large. Figure \ref{fig:F1_Sim1} reports the F1 score also for the maximum-a-posteriori (MaP) model, that is the model with maximum posterior probability, however this appears not to perform better than the median probability model. Note that the MaP model is readily available from the RJ algorithm while it is not directly available when the SSVS algorithm is employed. Additional results from this simulation study are provided in Section \ref{appD} of the supplementary material. Similar findings are observed when covariates are simulated with either an equicorrelated covariance matrix or a decreasing correlation structure; full results are available in Section \ref{appD} of the supplementary material.\par
The results reveal several important insights into the computational efficiency and model performance of the SSVS and RJ algorithms across different $p$ and $n$ configurations. While both methods deliver comparable performance in terms of AUC, F1 score, and other metrics---aside from some extreme cases---their computational demands  differ markedly. 
As shown in Figure \ref{fig:comp_time}, which displays the logarithmic computational time required for $50,000$ draws, the RJ algorithm with the R matrix update achieves a significant reduction in computational time relative to SSVS. 
This reduction enables longer or multiple chains, facilitating a more thorough exploration of the model space. However, both algorithms perform poorly when the number of predictors $p$ is large relative to the sample size $n$, particularly when the number of relevant predictors $p_0$ increases. High-dimensional noise complicates variable selection in these cases, limiting the effectiveness of both approaches. However, as $n$ increases, the RJ algorithm clearly outperforms SSVS, achieving higher F1 scores and AUC, while also exhibiting reduced variability. This enhanced stability suggests that the RJ algorithm’s computational approach, particularly its efficient R matrix updates, boosts its robustness in larger sample settings, enabling it to manage complexity with fewer performance fluctuations. We emphasise that selecting a model via RJ and SSVS is a standard task in Bayesian computation. The results in Figure \ref{fig:comp_time} clearly show that an RJ implementation relying on full refactoring would incur substantially higher computational costs. Taken together, the findings demonstrate that the RJ algorithm equipped with R-updates outperforms both the Rao-Blackwellised SSVS procedure and the naive full QR factorization strategy in terms of computational efficiency.
%
\section{Real data applications}
\label{sec:applications}
%
In this section, we consider some applications to real datasets that allow to highlight further advantages related to the availability of fast R updating algorithms.
The first dataset, \qmo Inflation\qmc, from \cite{bernardi_etal.2016}, focuses on predicting US inflation using quarterly changes in the consumer price index (CPIAUCSL). With $127$ observations and $21$ macroeconomic covariates covering the period from 1991 to 2023, this dataset supports two distinct analyses. The first analysis uses only covariates lagged by one quarter, enabling a direct comparison between the RJ approach and complete model enumeration, where posteriors are calculated for each model. The second analysis expands the covariate set to include lags up to four quarters, allowing us to assess the performance of the method in scenarios where $p$ approximates $n$. Details on the variables and sources can be found in Table \ref{tab:US_inflation_data} in Section \ref{sec:dataset_description} of the supplementary material.
The second dataset, \qmo Bardet-Biedl\qmc, involves microarray gene expression data from eye tissue in $120$ laboratory rats. Initially used by \cite{scheetz_etal.2006} for mammalian eye disease research, this dataset focuses on identifying genes linked to trim32, a gene associated with Bardet-Biedl syndrome, a multi-organ disorder with retinal implications \citep{chiang_etal.2006}. Of the $31,042$ probe sets available, $18,976$ exhibited sufficient signal and variation for reliable analysis, including trim32 and potentially influential genes. This additional dataset was chosen to represent a challenging scenario for variable selection, enabling a rigorous evaluation of method robustness in high-dimensional contexts where $p \gg n$. \par
The specification of the prior hyper-parameters $(\nu,\lambda,\xi,\varphi,\upsilon_0)$ is detailed in Section~\ref{app:prior_hyperparameters} of the supplementary material. Among these, $\upsilon_0$ is particularly influential, as it governs the trade-off between model fit and regularization. Several strategies for its calibration have been proposed, including asymptotic rules \citep{narisetty_he.2014}, empirical Bayes methods \citep{george_foster.2000}, and predictive criteria such as the cross-validation approach of \cite{lamnisos_etal.2012}. In this work, we adopt the latter, taking advantage of efficient posterior sampling algorithms that render its implementation both practical and computationally efficient. The adaptation of the approach by \cite{lamnisos_etal.2012} to leverage the proposed algorithms to quickly update the QR decomposition is discussed in Section~\ref{sec:model_prior_appl} of the supplementary material.
%
\subsection{Inflation}
%
\noindent The first application involves predicting  the relative change in the consumer price index from the previous quarter, using the full set of regressors listed in Table \ref{tab:US_inflation_data}. 
The first dataset, referred to as the \qmo small\qmcsp dataset, includes covariates from Table \ref{tab:US_inflation_data}, lagged by one quarter, and is used to predict inflation one quarter ahead. The second dataset expands on this by incorporating covariates for lags up to one year prior to the predicted quarter, resulting in a total of $84$ covariates. In the small dataset the exact posterior is calculated through the complete enumeration of alternative models and the computation of the marginal likelihood, which has a closed-form expression. This comparison underscores the effectiveness and accuracy of the RJ algorithm in approximating the true posterior distribution. 
%
%
\begin{figure}[!t]
\begin{center}
\subfigure[small: log predictive score]{\label{fig:ex_infl_lag1_CV_score}\includegraphics[width=0.45\textwidth]{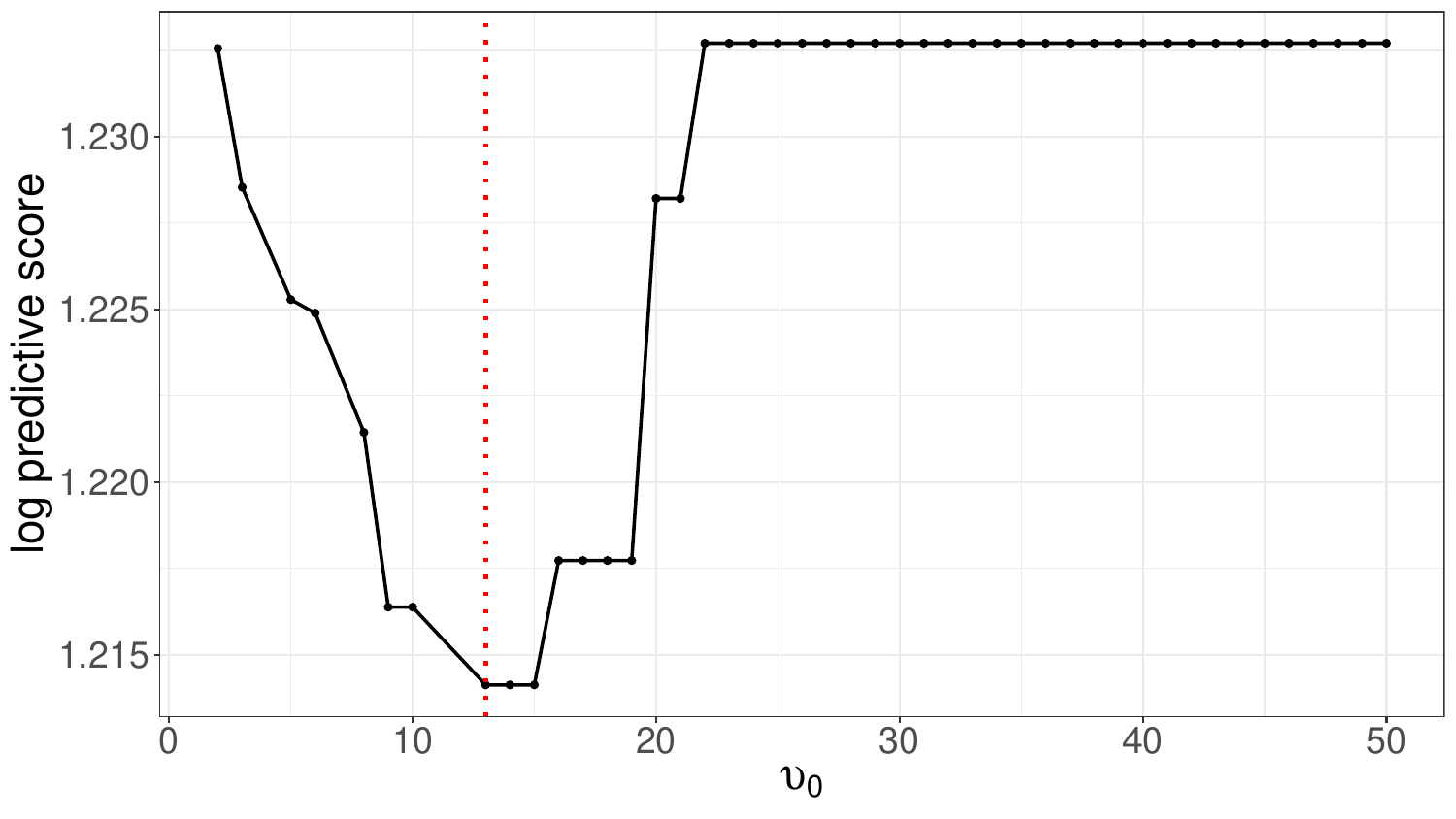}}\quad
\subfigure[small: MIP]{\label{fig:post_boxplots_inflation_lag1_MIP}\includegraphics[width=0.45\textwidth]{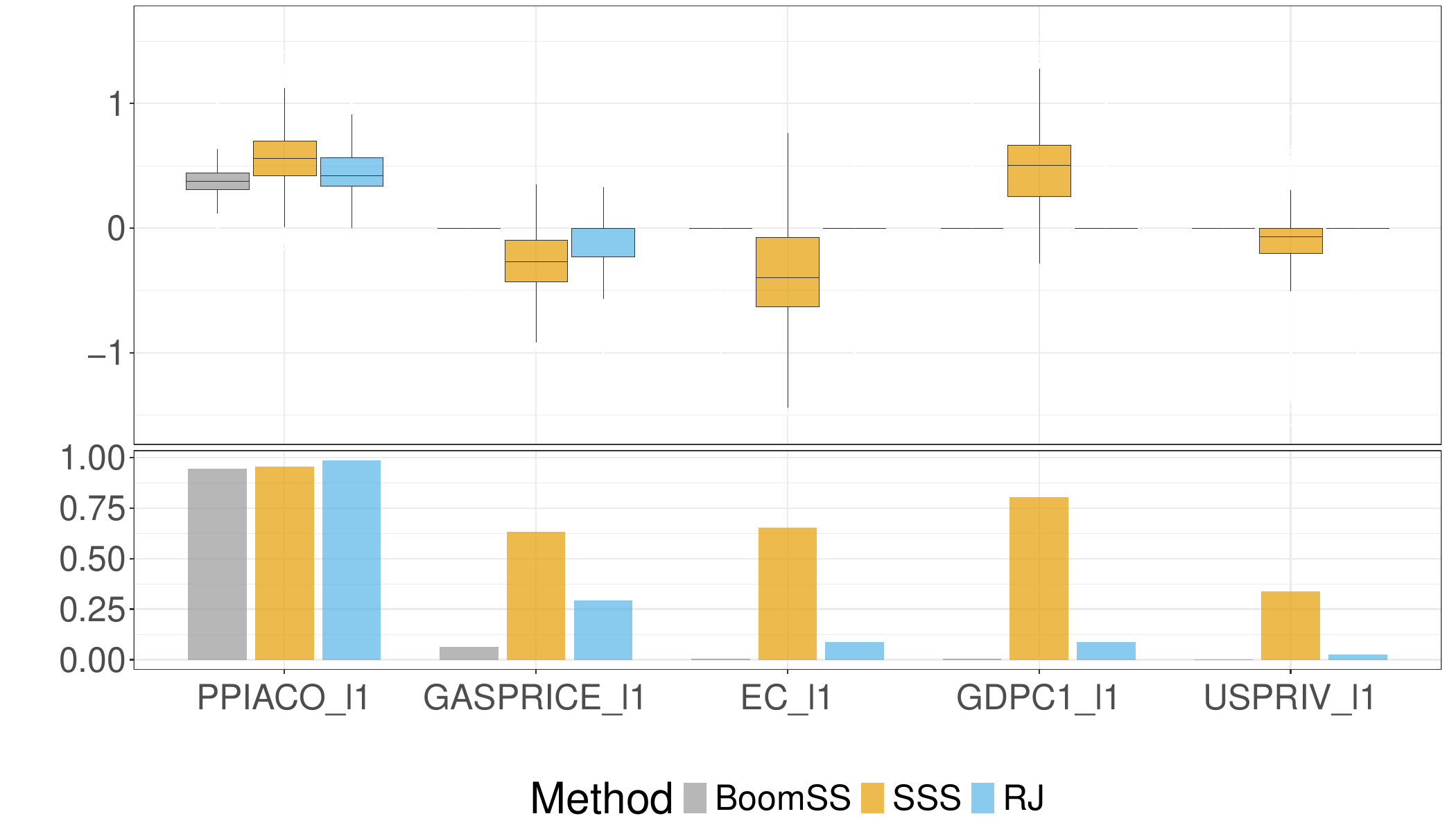}}\\
\subfigure[large: log predictive score]{\label{fig:ex_infl_lag4_CV_score}\includegraphics[width=0.45\textwidth]{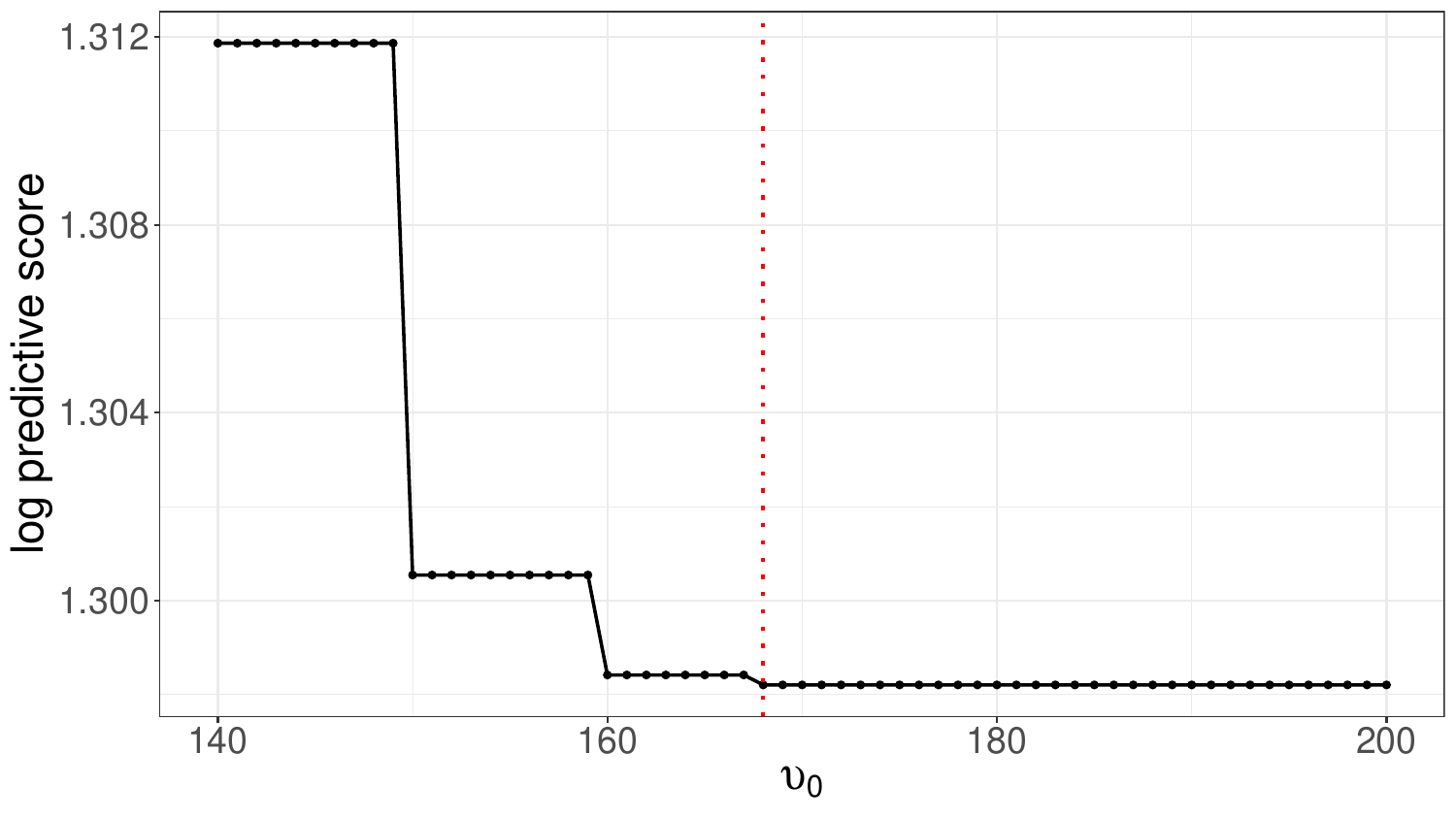}}\quad
\subfigure[large: MIP]{\label{fig:post_boxplots_inflation_lag4_MIP}\includegraphics[width=0.45\textwidth]{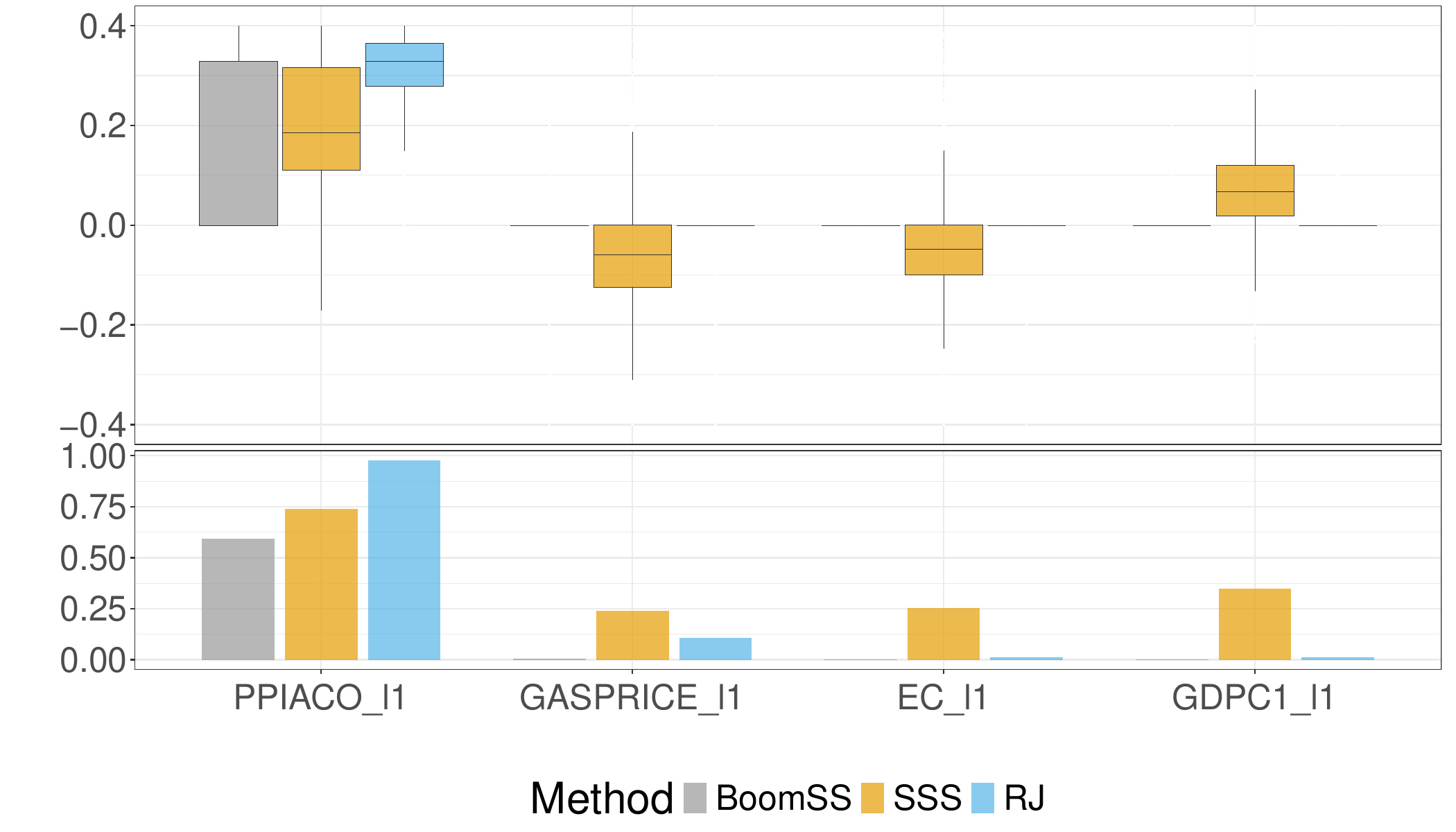}}
%
\caption{Inflation example. Predictive scores and summary of the posterior draws for the linear model fitted to the entire sample using various methods. The top panels consider the small dataset while the bottom panels consider the large one. The panels on the left show MCMC estimates of the log predictive score (its minimum is denoted by a vertical dotted red line). The right-hand panels show bar plots of the MIP and boxplots of the posterior distributions for the relevant regression coefficients.}
\label{fig:ex_inflation_post_summary}
\end{center}
\end{figure}
%
Figure \ref{fig:ex_inflation_post_summary} displays the posterior distributions of the regression coefficients and marginal inclusion probabilities (MIP) for variables selected in more than 20\% of iterations by at least one of the competing methods. The comparison includes results from the RJ algorithm, the Rao-Blackwellized SSVS method by \cite{ghosh_clyde.2011} and the Scalable Spike-and-Slab method by \cite{biswas_etal.2022} (implemented in the R package \texttt{ScaleSpikeSlab}, or SSS). These posterior distributions are derived from 200,000 post-burn-in samples. The two left panels of Figure \ref{fig:ex_inflation_post_summary} also present cross-validation results for the scale parameter $\upsilon_0$, evaluated using the log predictive score for each dataset. Notably, the results from all approaches are consistent, suggesting that each method performs similarly in capturing the underlying relationships within the data. However, the SSS method deviates slightly, showing some differences in its performance relative to the other approaches.\par
Posterior exploration of the space of $2^p$ competing models is another crucial feature of RJ algorithms. Table \ref{tab:inflation_posterior_models}  reports the posterior model probability (PMP), i.e. the probability that a given model $\boldsymbol{\gamma}$ is visited, for the top five models. These top five models have been identified as those with the highest relative visitation frequency reported by the MCMC chain, as shown in the column labeled $\widehat{\text{PMP}}$.
The primary drawback of both the BoomSS and SSS algorithms is their inability to provide the marginal likelihood. Therefore, we compare the $\widehat{\text{PMP}}$ with the \qmo true\qmcsp PMP as calculated by complete model enumeration which can be done because the number of covariates is small. Thus, if the RJ algorithm effectively explores the entire space of competing models, then $\widehat{\text{PMP}}_k$ is an unbiased estimator of $\text{PMP}_k$ for all $k = 1, \dots, 2^p$.
%
\begin{table}[!t]
\caption{For the three datasets, the first column  provides the model rank.  The following columns report the relative frequencies of visiting the corresponding model ($\widehat{\text{PMP}}$), the true normalized posterior delivered by complete enumeration  ($\text{PMP}$, available only in the small dataset), and the model dimension ($\text{dim}$), i.e. the number of covariates included in the model.}
%
\centering
\begin{tabular*}{\textwidth}{@{\extracolsep\fill}rccccccc}
\toprule
& \multicolumn{3}{c}{Inflation, small} & \multicolumn{2}{c}{Inflation, large} & \multicolumn{2}{c}{Bardet-Biedl}\\
\cmidrule{2-4} \cmidrule{5-6} \cmidrule{7-8}
Rank & $\widehat{\text{PMP}}$ & $\text{PMP}$ & $\text{dim}$ &  $\widehat{\text{PMP}}$ & dim  & $\widehat{\text{PMP}}$ & dim\\  
\midrule
  1  &  0.461 & 0.462 & 1 & 0.528 & 1 & 0.107 &   12 \\ 
  2  &  0.213 & 0.218 & 2 & 0.074 & 2 & 0.056 &   13\\ 
  3  &  0.034 & 0.033 & 3 & 0.027 & 2 & 0.048 &   15\\ 
  4  &  0.022 & 0.023 & 4 & 0.021 & 2 & 0.042 &   13\\ 
  5  &  0.021 & 0.021 & 2 & 0.016 & 2 & 0.036 &   14\\ 
%
\bottomrule 
\end{tabular*}
\label{tab:inflation_posterior_models} 
\end{table} 
%
\noindent Comparing the posterior distribution of the models selected by the RJ algorithm ($\widehat{\text{PMP}}$) with the distribution provided by complete model enumeration ($\text{PMP}$) in Table \ref{tab:inflation_posterior_models}, we note that the two columns almost coincide, which indicates that the RJ algorithm provides excellent exploration of the space of competing models.\par
%
%
We now assess the predictive accuracy of the RJ with R updating algorithm using a rolling window approach. We implemented a rolling window of $4\times 25$ quarters, meaning that at each step, the model was trained on the past $100$ quarters of data, and tested on the subsequent quarter. This process was repeated iteratively, sliding the window one quarter forward each time, resulting in a series of predictions and corresponding errors. To evaluate the model predictive performance, we calculated the root mean squared prediction error (RMSPE)  for each prediction.
This rolling window approach enables a dynamic evaluation of the model's performance, capturing potential shifts in the time series structure across different periods. The results are presented in Table \ref{tab:inflation_posterior_models_RW}. In summary, the analysis compares four primary approaches: RJ utilizing the R factorization, BoomSS, SSS and the Bayesian model averaging leveraging the RJ method of \cite{zeugner_feldkircher.2015} (as implemented in the R package \texttt{BMS}, or BMS). To summarize the posterior draws, we employ Bayesian model averaging (BMA), the median probability model \citep{barbieri_berger.2004} (MPM), and the MaP, where applicable. For each method, we contrast two key approaches: one that fixes the scale parameter $\upsilon_0$ based on the recommendation by \cite{narisetty_he.2014}, and another using the cross-validation (CV) method suggested by \cite{lamnisos_etal.2012}. These approaches are then evaluated against the full OLS estimate, stepwise OLS, regularized methods like ridge, lasso, adaptive-lasso, and elastic net \citep{hastie_etal.2009} and best subset selection \citep{hastie_etal.2020}. The adaptive weights for the adaptive-lasso are the reciprocal of the absolute values of the OLS-estimated coefficients.
The results indicate that the RJ method that leverages the R decomposition generally performs better than other approaches, particularly when using the BMA and selecting $\upsilon_0$ by CV techniques. Specifically, RJ BMA CV achieves the best performance with a value of RMSE of $0.067$ and $0.063$, for the small and large dataset, respectively. Among other RJ methods, the RJ MaP with $\upsilon_0$ selected by CV also shows good performance with $0.075$ for the large dataset. In comparison, Boom SS BMA and SSS BMA methods yield slightly higher values, but SSS BMA performs quite well with a RMSE of $0.075$ on the small dataset. 
On the other hand, traditional methods like OLS and stepwise-OLS show the worst performance, with values as high as $1.769$ and $7.587$ for the small and large dataset, respectively. Regularization techniques, such as ridge, adaptive-lasso, and elastic net, also demonstrate competitive performance, with values around $0.068$-$0.071$, positioning them as strong alternatives to the Bayesian methods, with the only exception of the ridge for the larger dataset. Finally, the best subset approach stands out with the best performance on the small dataset, achieving a value of $0.051$. However, it lacks comparable results for the larger datasets.\par
In conclusion, the RJ method utilizing the R factorization demonstrates superior predictive performance in forecasting the Consumer Price Index, particularly when employing Bayesian model averaging and cross-validation for hyper-parameter selection, outperforming traditional approaches and alternative Bayesian methods in both small and large datasets. Moreover, the provided R updating algorithms serve the purpose of comprehensive model building by significantly enhancing the speed of both model selection and cross-validation.
%
\begin{table}[!t]
\setlength{\tabcolsep}{10 pt}
\caption{Average RMSPE for the Inflation and the Bardet-Biedl Syndrome data analysis. 
\qmo CV\qmcsp refers to the case where $\upsilon_0$ is selected through cross-validation, while \qmo BMA\qmcsp denotes the Bayesian model averaging estimate, \qmo MPM\qmcsp and \qmo MaP\qmcsp (available only for the RJ algorithm) are the median probability model and the maximum-a-posteriori estimates of the coefficients, respectively.  For the Bardet-Biedl dataset, the last column report either the  number of selected probe sets  or their average (for RJ and BoomSS) when we fit the different models to the complete data set.}
%
\centering
\resizebox{\columnwidth}{!}{
\begin{tabular}{llcccc}\\
\toprule
\multirow{2}{*}{Algorithm}&\multirow{2}{*}{Method}&\multicolumn{2}{c}{Inflation}&\multicolumn{2}{c}{Bardet-Biedl}\\
\cmidrule(l){3-4}\cmidrule(l){5-6}
&&\textit{small dataset} & \textit{large dataset}& \multirow{2}{*}{RMSPE}& \multirow{2}{*}{\# of probes} \\  
&&RMSPE&RMSPE\\
\midrule
RJ  &BMA & 0.086 &0.093 & 0.402&15\\ 
RJ &MPM & 0.130 &0.135& 0.438 &17\\ 
RJ & MaP & 0.199  &0.282& 0.432&14 \\ 
RJ&BMA CV & 0.067& 0.063 & 0.334&11 \\
RJ &MPM CV &0.162 &0.120 &  0.432&12\\
RJ &MaP CV& 0.161 & 0.075 & 0.425&10 \\
BoomSS &  BMA & 0.096 &  0.134 &  0.897&4 \\ 
BoomSS&MPM & 0.136 &0.164& 0.985&5  \\ 
SSS &BMA & 0.075  &0.153 & & \\ 
SSS & MPM & 0.112&  0.136& & \\ 
BMS& BMA & 0.111 &1.281& & \\ 
\midrule
&OLS & 1.769  &7.587 & & \\ 
& stepwise& 1.769 &7.587& & \\ 
   \midrule
&  ridge & 0.068 & 0.982 &1.145  &\\ 
&  adaptive-lasso & 0.071 &0.071  & 1.241 &65 \\ 
&  elastic net & 0.071 & 0.071 & 1.191&32 \\ 
    \midrule
&  best subset & 0.051 & & \\ 
\bottomrule 
\end{tabular}
} 
\label{tab:inflation_posterior_models_RW} 
\end{table}
%
%
\subsection{Bardet–Biedl syndrome gene expression study}
%
\noindent We now analyze a microarray data set consisting of gene expression measurements from the eye tissue of $120$ laboratory rats. As for the inflation examples, Figure \ref{fig:ex_Bardet-Beidl_post_summary} summarizes the results concerning the posterior distribution  and the MIP of the regression coefficients for variables selected in more than 20\% of iterations by at least one of the competing methods, as delivered by the RJ algorithm and the Rao-Blackwellized SSVS method of \cite{ghosh_clyde.2011}. For the RJ algorithm, posterior distributions are approximated using $1,000,000$  post burn-in posterior draws, while only $200,000$ posterior draws are considered for the alternative method. The scalable spike-and-slab approach of \cite{biswas_etal.2022} was excluded due to its inefficiency when $p$ is very large. BoomSS and RJ  methods show completely different outcomes. This discrepancy may be due to the differing inference strategies of each method: while BoomSS builds on an enhanced version of the traditional SSVS model, RJ might be influenced by a more scalable approach with flexible parameters, leading to more varied variable selection.\par
%
\begin{figure}[!t]
\begin{center}
\subfigure[log predictive score]{\label{fig:Bardet-Beidl_CV}\includegraphics[width=0.45\textwidth]{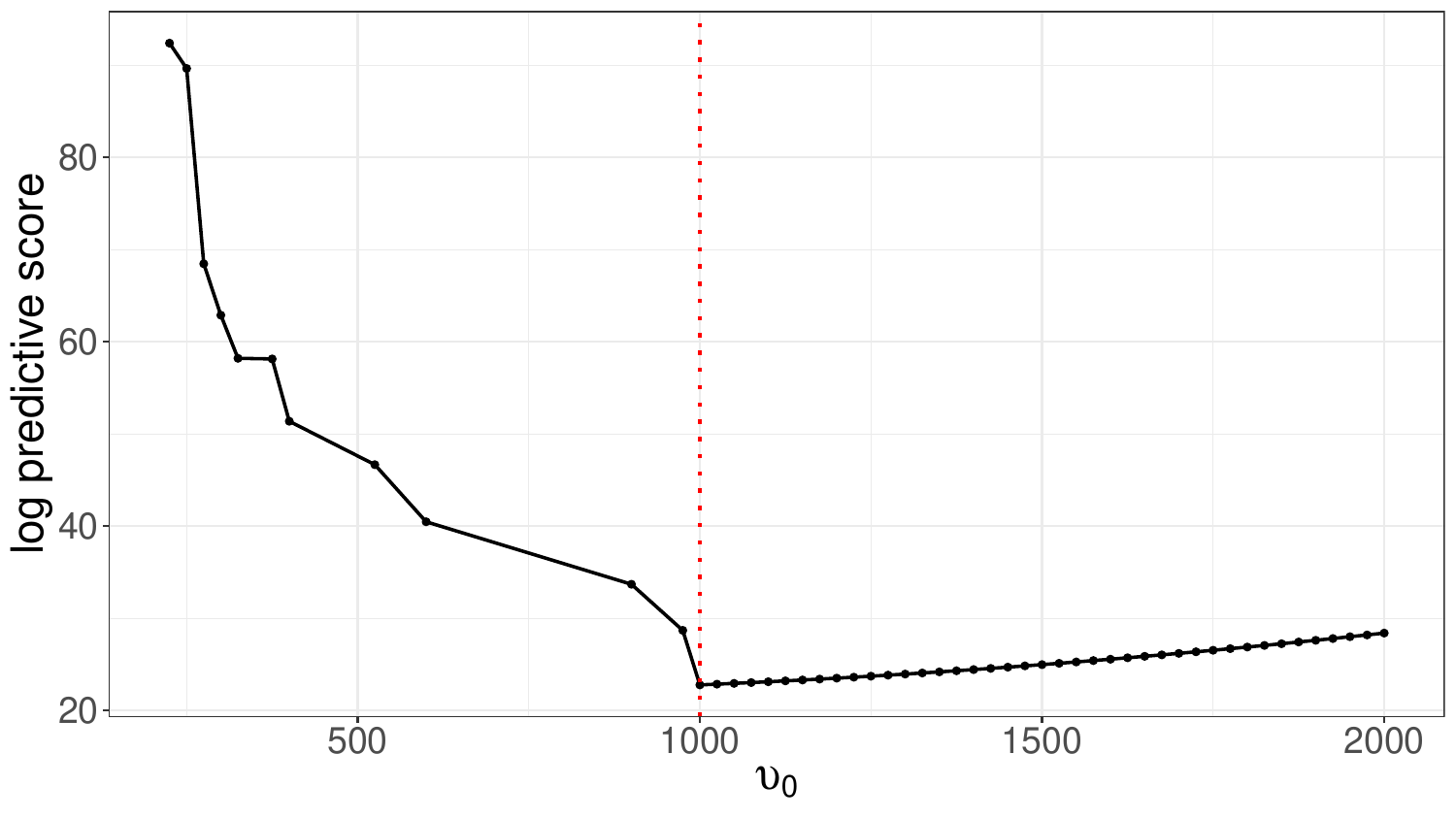}}\quad
\subfigure[Coefficients]{\label{fig:post_boxplots_Bardet-Beidl_coef}\includegraphics[width=0.45\textwidth]{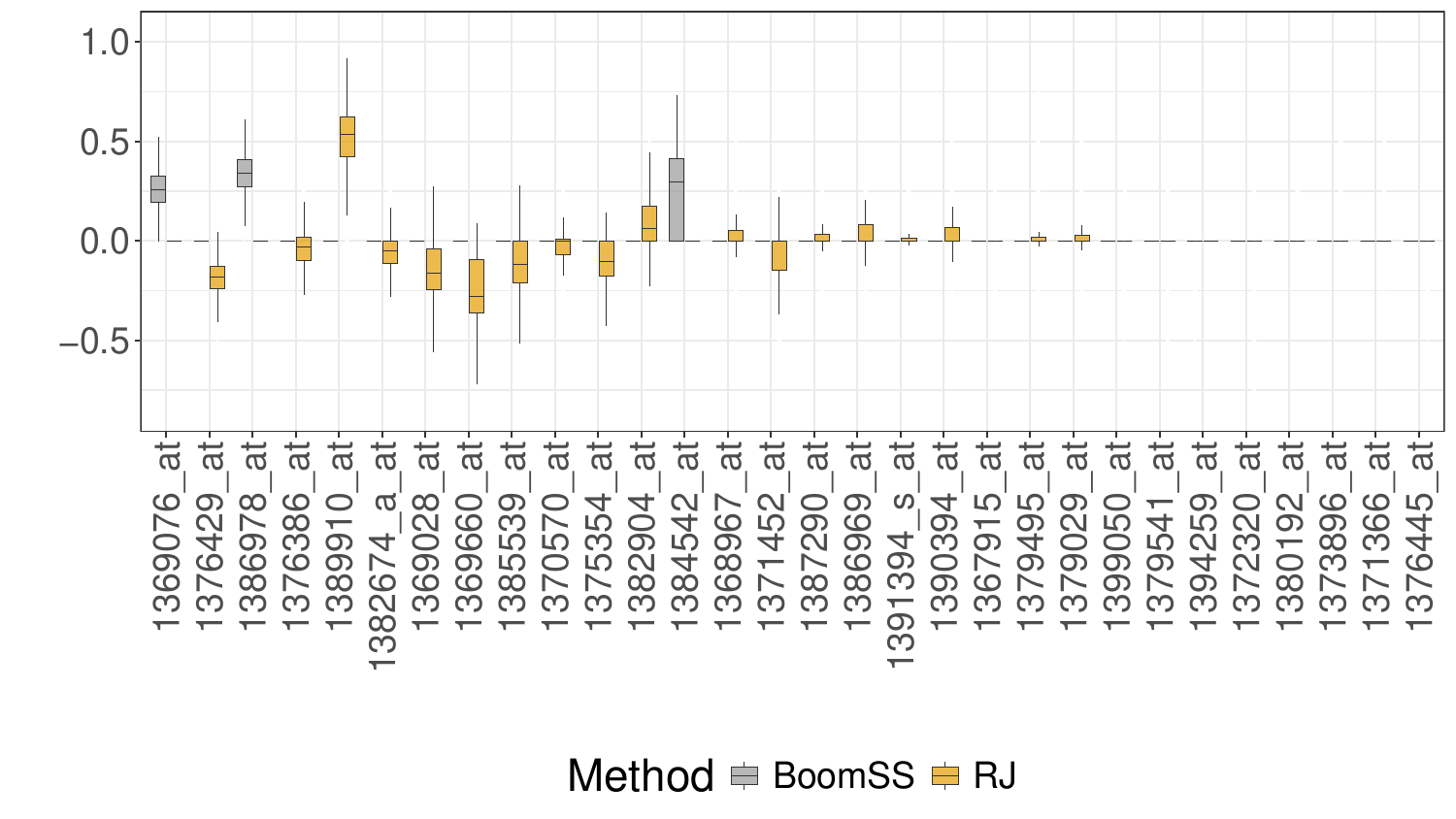}}\\
\subfigure[MIP]{\label{fig:post_boxplots_Bardet-Beidl_MIP}\includegraphics[width=0.45\textwidth]{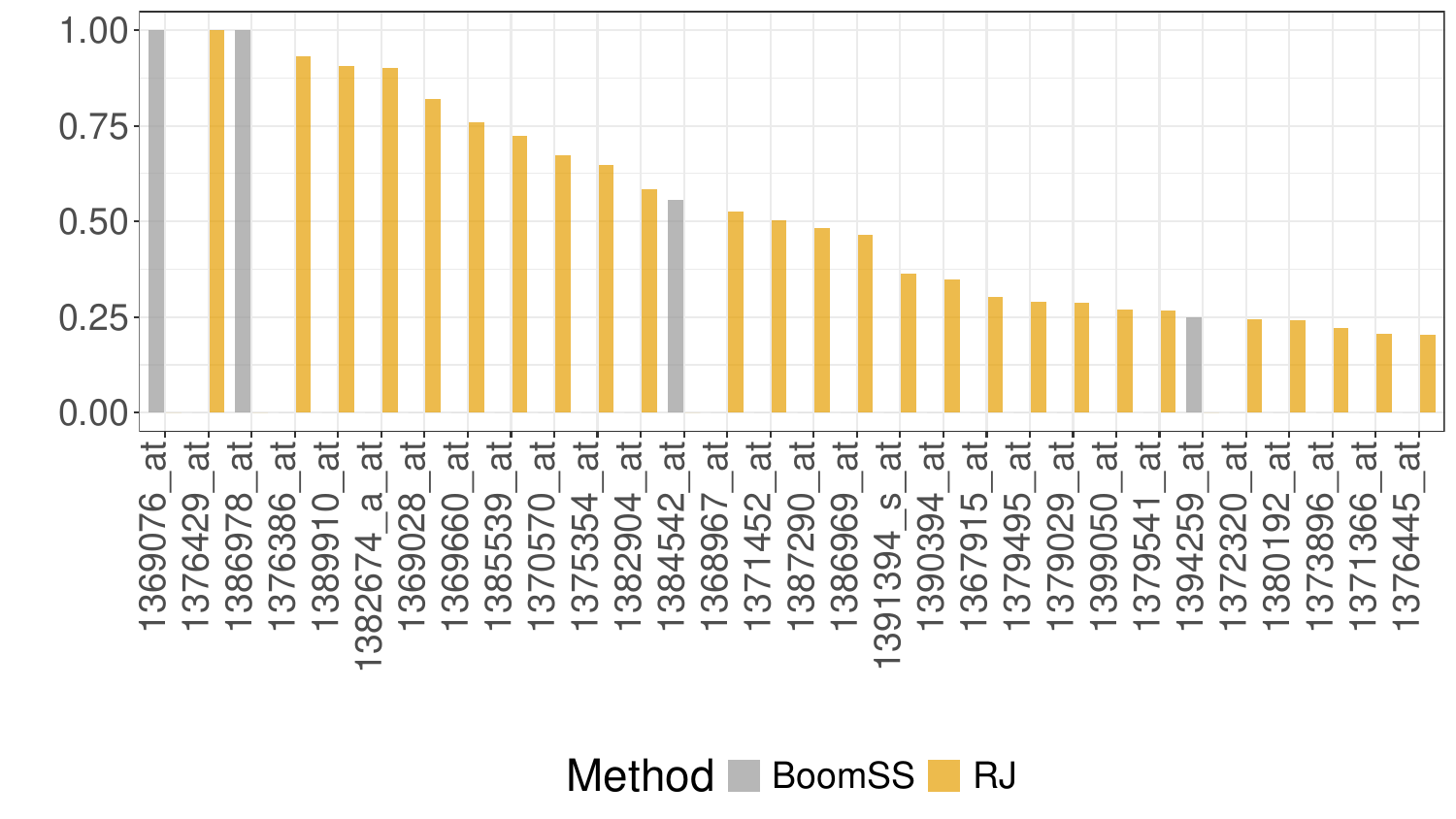}}
\caption{Bardet-Biedl syndrome gene expression study. Predictive scores and a summary of the posterior draws for the linear model fitted to the entire sample using various methods.}
\label{fig:ex_Bardet-Beidl_post_summary}
\end{center}
\end{figure}
%
%
Table \ref{tab:Bardet-Biedl_summary_stat} (postponed in Section \ref{appG} of the supplementary material to save space) reports the posterior summaries of the five best models selected by the RJ algorithm in the Bardet–Biedl syndrome gene expression study. The table provides a comprehensive comparison of model-based summaries, which is particularly valuable for assessing the contribution of different genes to the model space and for identifying those with the strongest evidence of association with the syndrome.
%
%
To evaluate the predictive accuracy of the models, we replicated the analysis conducted by \cite{bai_etal_2021}. We began by randomly splitting the dataset into $90$ training observations and $30$ test observations. We fitted the models to the training set and utilized the estimated coefficients $\widehat{\boldsymbol{\beta}}_{\text{train}}$ to compute the RMSPE. This procedure was repeated $100$ times to ensure the robustness of our results, and we computed the average RMSPE across these iterations. For the analysis, we compared the results for the RJ and Boom spike-and-slab methods, employing the same frameworks previously applied to the Inflation dataset. We excluded the scalable spike-and-slab algorithm due to its prohibitive computational demands, as well as best subset selection and ordinary least squares (OLS) methods, which were deemed infeasible in this context.\par
The results of our analysis are summarized in the last two columns of Table \ref{tab:inflation_posterior_models_RW} that report the RMSPE as well as the number of selected probe sets when we fit the different models to the complete data set. Notably, the RJ BMA CV method exhibited the lowest out-of-sample RMSPE, signifying the highest predictive power among the methods evaluated. Meanwhile, the Boom SS BMA method identified the most parsimonious model, selecting only four probe sets from the $18,976$ available. However, it is important to note that  RMSPE for Boom SS was significantly higher than that of the other approaches. In contrast, the RJ method selected from $14$ to $17$ probe sets, which is fewer than the $32$ and $71$ selected by the lasso and elastic net methods, respectively. This analysis highlights that, for this particular dataset, the RJ method achieved an admirable balance of predictive performance and model parsimony, making it the most effective choice in this instance.

\section{Discussion}
\label{sec:discussion}
%
Fast R updating and downdating algorithms are  valuable for other tasks besides least square estimation, such as multivariate regression, semi- and non-parametric modeling, graph selection and hyper-parameter tuning in penalized methods. In this section, we illustrate concrete scenarios where repeated, fine-grained updates of the R matrix are essential.\par
%
%
Many existing approaches, particularly MM algorithms \citep{lange.2016} and ADMM schemes \citep{boyd_etal.2011}, tend to rely on repeated matrix inversions, which become computationally burdensome when the number of predictors is large relative to the sample size. R updating techniques offer a highly efficient alternative, enabling the rapid recalculation of residual sums of squares and related quantities. This makes the inclusion and exclusion of predictors feasible, even in high dimensions. These techniques also play a crucial role in penalized regression procedures, such as least angle regression (LARS) \citep{efron_etal.2004}, where the evolving active set $\mathcal{A}$ requires fast updates of $(\mathbf{X}_{\mathcal{A}}^\top \mathbf{X}_{\mathcal{A}})^{-1}$.\par
%
%
In addition to regression and penalized estimation, R updates provide significant advantages for sequential hypothesis testing, best subset approaches and stepwise model selection \citep[see, e.g.][]{hastie_etal.2020}. In these scenarios, full recomputation can quickly become impractical, particularly in high-dimensional problems or when repeated regression fits are necessary. The proposed methods update the R decomposition directly, enabling the relevant test quantities to be recalculated quickly without restarting the computational pipeline. This results in significant efficiency gains in applications such as real-time monitoring, including quality control settings \citep{yao_etal.1999}, anomaly detection \citep{chandola_etal.2009}, adaptive clinical trials in which treatment decisions evolve with the accumulation of patient data \citep{zhao_etal.2013}, and graphical model inference  \citep{sanghavi_etal.2007}, in which sequential conditional independence tests guide structure learning in complex networks. 
For practitioners, this efficiency translates into the ability to explore a much wider range of candidate models in less time, leading to more reliable identification of important predictors.\par
%
%
%
%
These ideas naturally extend to multivariate and seemingly unrelated regression models \citep{zellner.1962}, which must explicitly account for correlations among multiple outcomes. In such settings, sparsity and joint covariate selection play an increasingly prominent role \citep{brown_etal.1998,bottolo_etal.2020}. As in the univariate case, R updating enables efficient exploration of these models by streamlining the required factorizations, maintaining computational feasibility even when both the number of predictors and the number of outcomes are large.\par
%
%
The advantages of fast updating methods extend beyond linear regression. Similar benefits arise in semi-parametric regression \citep{ruppert_etal.2003}, where nonlinear effects are captured through spline or penalized spline representations \citep[see, for example,][]{wood.2017,eilers_marx.2021,lang.2004}. As these models typically involve group-wise inclusion or exclusion of sets of basis functions, rapid updates are particularly valuable for iterative feature selection and the repeated smoothing parameter adjustments that are often required in estimation routines.\par 
%
%
The same principles also apply to non-parametric and structured models. Gaussian Markov random fields \citep{rue_held.2005} and state-space models \citep{durbin_koopman.2012}, for example, impose temporal or spatial dependencies that give rise to large, structured precision matrices. Updating the R factor substantially reduces the computational burden while preserving stability under sequential modifications. Non-linear machine learning approaches, such as Bayesian additive regression trees \citep{chipman_etal.2010}, also benefit from this. Updating the R decomposition enables the efficient addition or removal of tree components, thereby improving scalability without compromising accuracy.\par
%
%
A complementary set of challenges arises when rows are added or removed. This scenario frequently occurs in cross-validation \citep{konishi_kitagawa.2008}, when evaluating model goodness-of-fit on new data, in sequential learning \citep{lu_etal.2015,luo_song.2020,van_wieringen_binder.2022}, and to facilitate parallel computing \citep{adams_etal.1996,kontoghiorghes.2006,ko_etal.2022}. Methods for updating QR and R matrices tailored to row modifications can deliver substantial computational gains in these settings, especially in high-dimensional problems.\par
However, only certain algorithmic families benefit directly from these row-wise strategies, namely those whose updates require solving regularized normal equations. Prominent examples include specific MM algorithms \citep[see references in][]{lange.2016,sun_etal.2017} and, most notably, ADMM \citep{boyd_etal.2011}, whose iterations involve generalized ridge matrix inversions. 
When combined with R updates, ADMM is particularly effective for hyper-parameter tuning and model selection across the aforementioned high-dimensional regression methods. Similar gains arise in convex clustering \citep{chi_lange.2015} and graphical model estimation \citep{friedman_etal.2007}, where tuning parameters are usually chosen via cross-validation.\par
In practice, model selection and hyper-parameter tuning often account for the majority of the computational cost. Procedures such as cross-validation and generalized cross-validation (GCV) require repeated refitting on slightly perturbed datasets. Although there are analytic shortcuts for specific models, e.g. efficient leave-one-out cross-validation for penalized splines, see \cite{eilers_marx.2021}, row-wise R updates are essential in online or sequential settings where fast recalculation is needed for recursive least squares, online ridge-type methods and state-space models \citep{bottegal_pillonetto.2018}. Significant extensions to GCV for non-regular problems have also been developed \citep{jansen.2015}. A related issue arises when the goal shifts from tuning-parameter selection to evaluating a model's out-of-sample performance. Here, too, repeated row-wise modifications of the design matrix are required, and fast, stable R updates provide substantial efficiency gains.\par
Overall, the evidence highlights the versatility and scalability of R-updating algorithms. By enabling rapid and stable recomputation without full decompositions, these algorithms provide a unifying computational tool that can be used with linear, multivariate, semi-parametric, non-parametric and graphical models, as well as with tasks such as feature selection, covariance estimation, tuning and out-of-sample evaluation. The advantages of these algorithms are most evident in high-dimensional and iterative environments, where they transform otherwise impracticable procedures into efficient components of modern statistical workflows.

\section{Conclusion}
\label{sec:conclu}
This paper introduces efficient algorithms for the rapid updating and downdating of QR factorizations, delivering substantial computational gains by avoiding repeated recomputation of the full decomposition when data are modified, a common requirement in statistical applications. Building on classical QR updating techniques, we extend them in two main directions: first, by reducing computational cost through avoiding explicit recomputation of the Q factor, and second, by allowing simultaneous updates of multiple, potentially non-adjacent, rows and columns. For each extension, we formalize the corresponding algorithm and derive its exact computational cost, enabling a precise assessment of the resulting reductions in FLOPS, memory usage, and runtime. Numerical experiments on both simulated and real datasets confirm the practical benefits of the proposed methods, demonstrating improved efficiency, stability, and scalability in high-dimensional statistical applications.\par
To place these contributions in the broader context of modern numerical linear algebra, it is worth noting that recent research has focused extensively on parallel and communication-avoiding QR factorizations tailored to multicore and manycore architectures. While parallelism plays a central role in contemporary linear algebra libraries, it is not the primary objective of the present work. The proposed R-update algorithms are designed to incrementally modify an existing triangular factor rather than recomputing a full QR decomposition. As these updates rely on sequential adjustments to individual rows or columns, they offer limited scope for effective parallelization. However, this sequential structure is well matched to statistical settings in which frequent, small-scale data modifications dominate computational costs, making incremental updates more advantageous than large parallel factorizations.\par
Within this same perspective, one might also consider the Sherman-Morrison-Woodbury (SMW) identity as an alternative strategy for low-rank updates, particularly when working with Gram matrices. However, its computational cost for rank-one updates is essentially comparable to that of QR-based methods, as shown in Proposition~\ref{prop:onecolinv} in Section~\ref{appC} of the supplementary material. We therefore do not pursue SMW-based approaches for two main reasons. First, our objective is to maintain an upper-triangular factor R, akin to a Cholesky factor, whereas SMW updates generate dense intermediate quantities and do not preserve triangular structure. Second, in the statistical applications motivating this work, such as likelihood-ratio tests, model selection, and repeated evaluation of quadratic forms, the availability of a triangular and sparse factor substantially simplifies the repeated solution of linear systems, which typically dominates overall computational cost. Updating R thus aligns naturally with these downstream tasks and provides structural advantages that SMW does not offer. For completeness, we note that \cite{sturm_christensen.2012} compare naive, QR, Cholesky, and SMW updates and conclude that, although all non-naive methods are efficient, no single approach is uniformly optimal across all scenarios.\par
Overall, the proposed algorithms provide a flexible and computationally efficient framework for QR updating and downdating, well suited to modern statistical workflows involving complex and high-dimensional data.\par

\section*{Supplementary information}\par
\noindent
The supplementary material provides the proofs of the main results in Section \ref{sec:thinqrupdate_algo}; outline the algorithms for updating and downdating QR and R factorizations, as well as for updating the matrix $\boldsymbol{\Sigma}=(\mathbf{X}^\top\mathbf{X})^{-1}$; provide formal derivations of the computational costs for each algorithm in terms of floating-point operations. Furthermore, they contain the description of the procedure for selecting the prior hyper-parameters used in the simulated and real data examples of Sections \ref{sec:simulations}-\ref{sec:applications}; present additional simulation results, and provide a detailed description of the inflation dataset besides additional results for real data applications.

\section*{Funding}
The authors acknowledge support from the European Union - Next Generation EU, Mission 4 Component 2 - CUP C53D23002580006 via the MUR-PRIN grant 2022SMNNKY. Mauro Bernardi also acknowledges partial funding by the BERN BIRD2222 01 - BIRD 2022 grant from the University of Padua.

\section*{Declarations}
\textbf{Conflicts of interests}
 The authors declare no competing interests or
conflicts of interest.

\begin{appendix}

\section{Updating blocks of columns or rows}
\label{secA1}
\subsection{Householder reflections}
Householder reflections can be employed in order to zero more than on cell of a matrix at a time.
Consider $\mathbf{x}_i$, column $i$ of matrix $\mathbf{X}$, and assume that entries from $j>i$ to $k>j$ are different from zero, then they can be zeroed through premultiplication of matrix $\mathbf{X}$ by the Householder matrix $\mathbf{H}_i(j,k)$ with normal vector $\mathbf{v}_i(j,k)\in\mathbb{R}^N$ 
\begin{equation*}
\mathbf{H}_i(j,k)= \mathbf{I}_N - \tau \mathbf{v}_i(j,k) \mathbf{v}^{\top}_i(j,k), \quad \tau = \frac{2}{\|\mathbf{v}_i(j,k)\|_2^2},
\end{equation*}
where $\|\mathbf{v}_i(j,k)\|_2=\sqrt{\mathbf{v}^\top_i(j,k)\mathbf{v}_i(j,k)}$ is the $\ell_2$-norm of the vector $\mathbf{v}_i(j,k)$ and $\mathbf{v}_{i}(j,k)=\widetilde{\mathbf{v}}_{i}(j,k)/\widetilde{\mathbf{v}}_i(j,k)[i]$ where
\begin{equation}
\label{eq:householder_generic}
\widetilde{\mathbf{v}}_{i (j,k)}=  \left[ 
\begin{array}{l}
\mathbf{0}_{i-1}\\
\mathbf{x}_i[i] + \mathsf{sign}\left(\mathbf{x}_{i}\left[ i \right] \right) \Vert   \big[ \mathbf{x}_i[i]\,\,\mathbf{x}_i[j:k] \big]  \Vert_2 \\
\mathbf{0}_{j-i-1}\\
\mathbf{x}_i[j:k]\\
\mathbf{0}_{N-k}\\
\end{array}
\right].
\end{equation}
The definition of the normal vector given in equation \eqref{eq:householder_generic} is useful for computational costs calculation, as it takes advantage of the presence of zeros in the matrix, while the usual definition of Householder reflection has normal vector $\mathbf{v}_i(i+1,N)$. The Householder matrix is an $N \times N$ symmetric and orthogonal matrix, and the QR decomposition of $\mathbf{X}$ can be computed through a sequence of $p$ Householder reflections applied to $\mathbf{X}$:
\begin{equation*}
\mathbf{H}_{p}(p+1, N) \mathbf{H}_{p-1}(p, N) \cdots \mathbf{H}_{1}(2, N) = \mathbf{Q}^{\top}.
\end{equation*}

\subsection{QR updating algorithms for blocks of rows or columns}
Hereafter, we denote by $P_m(k,l)$ the permutation matrix that moves the $m$ rows starting in position $k$ to the $m$ rows starting in position $l$.\par
Consider the addition of $m$ rows, in this instance
\begin{equation}
\label{eq:qr_add_block_rows}
\mathbf{X}^{+} = \begin{bmatrix} \mathbf{X}\left[ 1 : (k-1), \;\; \right] \\ \mathbf{U}_\star \\ \mathbf{X}\left[ k : N, \;\; \right] \end{bmatrix} \quad \text{ and } \quad \mathcal{P}_{m}(k, N+1) \mathbf{X}^{+} = \begin{bmatrix} \mathbf{X} \\ \mathbf{U}_\star \end{bmatrix},
\end{equation}
which leads to
\begin{equation}
\begin{bmatrix} \mathbf{Q}^{\top} & \mathbf{0}_{N, m}\\ \mathbf{0}_{N, m}^\top & \mathbf{I}_m \end{bmatrix} \mathcal{P}_{m}(k, N+1) \mathbf{X}^{+} = \begin{bmatrix} \mathbf{R} \\ \mathbf{U}_\star \end{bmatrix} = \widetilde{\mathbf{R}}.
\label{eq:24}
\end{equation}
Then, a sequence of $p$ Householder reflections is applied to $\widetilde{\mathbf{R}}$, and the same set of reflections is employed for obtaining $\mathbf{Q}^{+}$:
\begin{align*}
\mathbf{R}^{+} = & \; \mathbf{H}_p (N+1, N+m)\cdots\mathbf{H}_1(N+1, N+m) \widetilde{\mathbf{R}}, \\
\mathbf{Q}^{+} = & \; (\mathcal{P}_{m}(k, N+1))^{\top} \begin{bmatrix} \mathbf{Q} & \mathbf{0}_{N, m}\\ \mathbf{0}_{N, m}^\top & \mathbf{I}_m \end{bmatrix} \mathbf{H}_1(N+1, N+m)  \cdots \mathbf{H}_p (N+1, N+m), \nonumber
\end{align*}
where $\mathbf{H}_i(j,k)$, $i = 1,2, \dots, p$, is the Householder matrix with normal vector $\widetilde{\mathbf{v}}_i(j, k) \in \mathbb{R}^{N+m}$ defined in equation \eqref{eq:householder_generic}. After the update, we obtain  $\mathbf{X}^{+} = \mathbf{Q}^{+} \mathbf{R}^{+}$.\par

\indent Now, consider when a block of $m$ rows starting at position $k$ is removed, then 
\begin{equation*}
\mathbf{X}^{-} = \begin{bmatrix} \mathbf{X}\left[ 1 : (k-1), \;\; \right] \\ \mathbf{X}\left[ (k+m) : N, \;\; \right] \end{bmatrix} \quad \text{ and } \quad \mathcal{P}_{m}(k,1) \mathbf{X} = \begin{bmatrix} \mathbf{U}_\star \\ \mathbf{X}^{-} \end{bmatrix}.
\end{equation*}
The updated matrices $\mathbf{Q}^{-}$ and $\mathbf{R}^{-}$ satisfy the following equation:
\begin{equation*}
\mathcal{P}_{m}(k,1) \mathbf{X} = \begin{bmatrix} \mathbf{I}_m & \mathbf{0}_{m, N-m} \\ \mathbf{0}_{m, N-m}^\top & \mathbf{Q}^{-} \end{bmatrix} \begin{bmatrix} \mathbf{Z}_\star \\ \mathbf{R}^{-} \end{bmatrix}=\mathcal{P}_{m}(k,1)\mathbf{Q}\mathbf{R}=\mathbf{Q}_p\mathbf{R},
\end{equation*}
where $\mathbf{Q}_p=\mathcal{P}_{m}(k,1)\mathbf{Q}$ and $\mathbf{Z}_\star$ is a matrix of dimension $m\times p$. Therefore, $\sum_{k=1}^{m}(N - k)$ Givens rotations are applied to $\mathbf{Q}_p$, that is equation \eqref{eq:deleteonerow_givens} is repeatedly applied to the set of rows $\mathbf{Q}_p[j,]$, for $j=1,2,\dots,m$, giving
\begin{align*}
&\mathcal{P}_{m}(k,1)\mathbf{Q} \mathbf{G}_1(N-1, N) \cdots \mathbf{G}_1(1, 2) \cdots \mathbf{G}_m(N-1, N)\cdots \nonumber\\
&\qquad\qquad\qquad\qquad\qquad\qquad\qquad\times\mathbf{G}_m(m, m+1) = \begin{bmatrix} \mathbf{I}_m & \mathbf{0}_{m, N-m} \\ \mathbf{0}_{m, N-m}^\top & \mathbf{Q}^{-} \end{bmatrix}, \\
&\mathbf{G}_m(m, m+1)^{\top} \cdots \mathbf{G}_m(N-1, N)^{\top} \cdots \mathbf{G}_1(1, 2)^{\top} \cdots \mathbf{G}_1(N-1, N)^{\top} \mathbf{R}=\begin{bmatrix} \mathbf{Z}_\star \\ \mathbf{R}^{-} \end{bmatrix}. 
\end{align*}
The update yields $\mathbf{X}^{-} = \mathbf{Q}^{-} \mathbf{R}^{-}$. Note that all the results in this section can be straightforwardly extended  to the case of non adjacent rows by employing the appropriate permutation matrix which moves the rows to be deleted on top of matrix $\mathbf{X}$ or the rows to be added at the bottom of matrix $\mathbf{X}$.\par

When a block of $m$ columns, $\mathbf{U}_\star$, is added starting at position $k$, so $\mathbf{X}^{+}= \begin{bmatrix} \mathbf{X}\left[ ,\; 1 : (k-1) \right] & \mathbf{U}_\star & \mathbf{X}\left[,\;  k : p\right] \end{bmatrix}$, then
\begin{equation*}
\mathbf{Q}^{\top} \mathbf{X}^{+} = \begin{bmatrix} \mathbf{R}\left[ , \;\; 1:(k-1) \right] & \mathbf{Z}_\star &  \mathbf{R}\left[ , \;\; k:p \right]\end{bmatrix} = \widetilde{\mathbf{R}},
\end{equation*}
where $\mathbf{Z}_\star = \mathbf{Q}^{\top}  \mathbf{U}_\star$.
Computing $\mathbf{R}^{+}$ requires setting to zero many elements of $\mathbf{Z}_\star$ and filling $\mathbf{R}\left[(i+1) : (i+m), \;\; i \right]$, $i = k, \dots, p$.
This can be achieved by sequentially applying Givens rotations to $\widetilde{\mathbf{R}}$. Then the new matrices $\mathbf{R}^{+}$ and $\mathbf{Q}^{+}$ are obtained as
\begin{align}
\label{eq:addmcols1}
\mathbf{R}^{+} &=\mathbf{G}_{k+m-1} (k+m-1, k+m)^{\top}\cdots\mathbf{G}_{k+m-1} (N-1, N)^{\top}\cdots\nonumber\\
&\qquad\qquad \times \mathbf{G}_{k} (k, k+1)^{\top} \cdots \mathbf{G}_k(N-1, N)^{\top} \widetilde{\mathbf{R}},\\ 
\label{eq:addmcols2}
\mathbf{Q}^{+}&=\mathbf{Q} \mathbf{G}_k(N-1, N) \cdots \mathbf{G}_k (k, k+1)\cdots\nonumber\\
&\qquad\qquad\times \mathbf{G}_{k+m-1} (N-1, N)\cdots\mathbf{G}_{k+m-1} (k+m-1, k+m). 
\end{align}
The update yields  $\mathbf{X}^{+} = \mathbf{Q}^{+} \mathbf{R}^{+}$. In equations \eqref{eq:addmcols1}-\eqref{eq:addmcols2} the Givens rotations can be replaced by Householder reflections.\par

Finally, if $ \mathbf{X}^{-}= \begin{bmatrix} \mathbf{X}\left[, \; 1 : (k-1) \right] & \mathbf{X}\left[,  \; (k+m) : p \right] \end{bmatrix}$, so that a block of $m$ columns is removed, then the updated QR decomposition requires to set to 0 all non-zero elements under the diagonal of matrix $\mathbf{R}$, with the columns from $k$ to $k+m-1$ removed.
New matrices $\mathbf{R}^{-}$ and $\mathbf{Q}^{-}$ are obtained by exploiting a sequence of Householder reflections. Let $\widetilde{\mathbf{R}}$ be matrix $\mathbf{R}$ with columns from $k$ to $k+m-1$ removed, then
\begin{align*}
\mathbf{R}^{-}&=\mathbf{H}_{p-m}(p-m+1, p) \cdots \mathbf{H}_{k}(k+1, k+m) \widetilde{\mathbf{R}},\nonumber\\
\mathbf{Q}^{-}&= \mathbf{Q} \mathbf{H}_{k}(k+1, k+m)^\top  \cdots \mathbf{H}_{p-m}(p-m+1, p)^\top.
\nonumber
\end{align*}
After the update, it yields  $\mathbf{X}^{-} = \mathbf{Q}^{-} \mathbf{R}^{-}$. Note that when $1 < m<p$ non-adjacent columns are removed from $\mathbf{R}$ at positions $k_{1}, \ldots, k_{m}$, $\mathbf{R}^{-}$ is computed by setting to zero all non-zero elements below the diagonal through an appropriate set of Gives rotations or Householder reflections. 

\subsection{R updating algorithms for blocks of rows or columns}

The updated matrix $\mathbf{R}_1^+$ after the addition of $m$ rows, $\mathbf{U}_\star \in \mathbb{R}^{m \times p}$, to $\mathbf{X}$, as in equation \eqref{eq:qr_add_block_rows}, can be computed as
\begin{equation*}
\begin{bmatrix}\mathbf{R}_1^{+}\\
\mathbf{0}_{m,p}
\end{bmatrix} = \mathbf{H}_p (p+1, p+m) \cdots \mathbf{H}_1(p+1, p+m) \widetilde{\mathbf{R}}_1,
\end{equation*}
where $ \widetilde{\mathbf{R}}_1=\begin{bmatrix} \mathbf{R}_1 \\ \mathbf{U}_{\star} \end{bmatrix} $, see proof in Section \ref{appA} of the supplementary material.\par
Challenges analogous to the case of deletion of one row are faced when $m$ rows are removed from matrix $\mathbf{X}$. In order to avoid computation of quantities related to matrix $\mathbf{Q}$, steps outlined in Algorithm \ref{alg:thinQRaddrows} of the supplementary material for the addition of a block of $m$ rows in matrix $\mathbf{X}$ can be retraced backwards. Specifically, we need to solve
%
\begin{equation*}
\begin{bmatrix}\mathbf{R}_1\\\mathbf{0}_{m, p}\end{bmatrix} = \mathbf{H}_p (p+1, p+m) \cdots \mathbf{H}_1(p+1, p+m) \begin{bmatrix} \mathbf{R}_1^- \\ \mathbf{U}_\star \end{bmatrix},
\end{equation*}
where $\mathbf{R}_1$ and $\mathbf{U}_{\star}$ are known. Entries of matrix $\mathbf{R}^{-}_1$ can thus be computed iteratively by applying Algorithm \ref{alg:thinQRdeleterows}, see Section \ref{appB} of the supplementary material.\par
When $m$ columns are added at the end of matrix $\mathbf{X}$, that is when updating $\mathbf{R}_1$ for $\mathbf{X}^{+} = \begin{bmatrix} \mathbf{X} & \mathbf{U}_{\star}  \end{bmatrix}$, where $\mathbf{U}_{\star} \in \mathbb{R}^{N \times m}$, then
\begin{equation*}
\begin{bmatrix} \mathbf{Q}_1^\top \\  \mathbf{Q}_2^\top \end{bmatrix} \mathbf{X}^{+} = \begin{bmatrix} \mathbf{R}_1 & \mathbf{Z}_{\star 1} \\   \mathbf{0}_{N-p, p} & \mathbf{Z}_{\star 2} \end{bmatrix} = \widetilde{\mathbf{R}}_1^{+},
\end{equation*}
where $\mathbf{Z}_{\star 1}=\mathbf{Q}_{1}^\top \mathbf{U}_{\star}$ and $\mathbf{Z}_{\star 2}= \mathbf{Q}_{2}^\top \mathbf{U}_{\star}$. However, this approach requires the evaluation of matrices $\mathbf{Q}_{1}$ and $\mathbf{Q}_{2}$. In order to avoid such computation, $\mathbf{Z}_{\star 1}$ can be determined by solving the linear system $\mathbf{R}_1^\top \mathbf{Z}_{\star 1} = \mathbf{X}^\top \mathbf{U}_{\star}$. Entries $\mathbf{R}^{+}_1\left[p+i, p+j \right]$, for $i=1, \ldots, m$ and $j \geq i, \ldots, m$ can be computed iteratively exploting the relationship $(\mathbf{X}^{+})^{\top} \mathbf{X}^{+} = (\mathbf{R}_{1}^{+})^{ \top}\mathbf{R}_{1}^{+}$, see Section \ref{appA} and Algorithm \ref{alg:thinQRaddcols} in Section \ref{appB} of the supplementary material.\par
Finally, the deletion of $m$ columns can be analogously accomplished.
Let $\mathbf{X}^-$ be the reduced form of $\mathbf{X}$ after the deletion of a block of $m$ columns starting at position $k$, then $\widetilde{\mathbf{R}}_1 = \begin{bmatrix}  \mathbf{R}_1 [, 1:(k-1)] & \mathbf{R}_1[, (k+m):p]\end{bmatrix}$.
Updated matrix $\mathbf{R}^-_1$ can be obtained through triangularization of matrix $\widetilde{\mathbf{R}}_1$. As with the full QR, see Section \ref{subsec:add_del_cols}, this can be done by applying a set of Householder reflections as follows
\begin{equation*}
\begin{bmatrix}\mathbf{R}_1^{-} \\
\mathbf{0}_{m,p-m}
\end{bmatrix}=\mathbf{H}_{p-m}(p-m+1, p) \cdots     \mathbf{H}_{k}(k+1, k+m) \widetilde{\mathbf{R}}_1.
\end{equation*}
Eventually, the upper triangular sub-matrix $\mathbf{R}_1^- \in \mathbb{R}^{(p-m) \times (p-m)}$ is selected.
Note that, similarly to the full QR update, if $\mathbf{X}^-$ is the reduced form of $\mathbf{X}$ after the deletion of $m$ columns at non adjacent positions $k_{1}, \ldots, k_{m}$, then updated matrix $\mathbf{R}_{1}^-$ can be obtained through triangularization of matrix $\mathbf{R}_{1}$ after the deletion of columns $k_{1}, \ldots, k_{m}$ by applying either Givens rotations or Householder reflections, depending on the number of elements to be zeroed below the main diagonal. 
\end{appendix}



\newpage
\clearpage
\vfill\eject

\setlength{\textwidth}{152mm}
\setlength{\topmargin}{-0mm}
\setlength{\textheight}{150mm}

\renewcommand{\theequation}{S.\arabic{equation}}
\renewcommand{\thesection}{S.\arabic{section}}
\renewcommand{\thetable}{S.\arabic{table}}
\renewcommand{\theproposition}{S.\arabic{proposition}}

\setcounter{equation}{0}
\setcounter{table}{0}
\setcounter{section}{0}
\setcounter{proposition}{0}
\setcounter{page}{1}
\setcounter{footnote}{0}

\begin{center}
{\Large Supplementary material for:}
\vskip3mm
\centerline{\Large\bf Fast QR updating methods for statistical applications}
\vskip7mm
\centerline{\normalsize\sc M. Bernardi, C. Busatto and M. Cattelan}
	\vskip5mm
	\centerline{\textit{Department of Statistical Sciences,
University of Padova}}{}
\end{center}

\noindent These supplementary materials are organized as follows. Section \ref{appA} provides the proofs of the main results in Section \ref{sec:thinqrupdate_algo}. Section \ref{appB} outlines the algorithms for updating and downdating QR and R factorizations, as well as for updating  $\boldsymbol{\Sigma}=(\mathbf{X}^\top\mathbf{X})^{-1}$, including supporting routines for Givens rotations and Householder reflections. Section \ref{appC} provides formal derivations of the computational costs for each algorithm in terms of floating-point operations (FLOPS). Section \ref{app:prior_hyperparameters} describes the selection of prior hyper-parameters used in the simulated and real data examples presented in Sections \ref{sec:simulations}–\ref{sec:applications} of the main paper. Section \ref{appD} presents additional simulation results, including full performance metrics for all correlation structures considered. Section \ref{sec:dataset_description} provides a detailed description of the inflation dataset. Section \ref{appG} provides additional results for real data applications.
%
\appendix
%
\section{Proofs}
\label{appA}
%
\subsection{Proofs of $\mathbf{R}_{1}$ updates when adding rows}
Consider the addition of one row at position $k$ and its permutation at row $p+1$, and assume the division in blocks $\mathbf{Q}_1=\begin{bmatrix}  \mathbf{Q}_{1,1} \\ \mathbf{Q}_{1,2}\end{bmatrix}= \begin{bmatrix}  \mathbf{Q}[1:p,1:p] \\ \mathbf{Q}[(p+1):N, 1:p]\end{bmatrix}$ and $\mathbf{Q}_2=\begin{bmatrix}  \mathbf{Q}_{2,1} \\ \mathbf{Q}_{2,2}\end{bmatrix}= \begin{bmatrix}  \mathbf{Q}[1:p,(p+1):N] \\ \mathbf{Q}[(p+1):N, (p+1):N]\end{bmatrix}$. In the thin QR context, equation \eqref{eq:12} becomes
\begin{equation}
\begin{bmatrix} \mathbf{Q}_{1,1}^\top & \mathbf{0}_p &  \mathbf{Q}_{1,2}^\top \\ \mathbf{0}^{\top}_p & 1 & \mathbf{0}^{\top}_{N-p} \\ \mathbf{Q}_{2,1}^\top & \mathbf{0}_{N-p}  & \mathbf{Q}_{2,2}^\top \end{bmatrix} \mathcal{P}(k,p+1) \mathbf{X}^{+} = \begin{bmatrix} \mathbf{R}_1 \\ \mathbf{x}_{\star}^\top  \\ \mathbf{0}_{N-p,p} \end{bmatrix} =  \begin{bmatrix} \widetilde{\mathbf{R}}_1  \\ \mathbf{0}_{N-p,p} \end{bmatrix}.
\label{eq:51}
\end{equation}
The QR factorization of $\mathbf{X}^+$ can be obtained by multiplying the matrices of the QR factorization of $\mathbf{X}$ by the sequence of Gives rotations that zero the $p$-th row of the right-hand side of equation \eqref{eq:51}. Specifically, matrix $\mathbf{R}_1^+$ can obtained by applying the following set of  Givens rotations 
\begin{equation*}
\begin{bmatrix}\mathbf{R}_1^{+} \\
\mathbf{0}^{\top}_p
\end{bmatrix}= \mathbf{G}_p \left(p, p+1\right)^\top \cdots \mathbf{G}_1\left(1, p+1\right)^\top\widetilde{\mathbf{R}}_1.
\end{equation*}
Note that this procedure requires only matrix $\mathbf{R}_1$ and the new row $\mathbf{x}_{\star}$.\par
When $m$ rows, $\mathbf{U}_\star \in \mathbb{R}^{m \times p}$, are added to $\mathbf{X}$, as in equation \eqref{eq:qr_add_block_rows}, the analogous of equation \eqref{eq:24} becomes
\begin{equation}
\begin{bmatrix} \mathbf{Q}_{1,1}^\top & \mathbf{0}_{p, m} & \mathbf{Q}_{1,2}^\top\\ \mathbf{0}_{p, m}^\top & \mathbf{I}_m & \mathbf{0}_{N-p, m}^\top \\ \mathbf{Q}_{2,2}^\top & \mathbf{0}_{N-p, m} & \mathbf{Q}_{2,2}^\top \end{bmatrix} \mathcal{P}_{m}(k, p+1) \mathbf{X}^{+} = \begin{bmatrix} \mathbf{R}_1 \\ \mathbf{U}_{\star} \\ \mathbf{0}_{N-p,p} \end{bmatrix} = \begin{bmatrix} \widetilde{\mathbf{R}}_1  \\ \mathbf{0}_{N-p,p} \end{bmatrix}.
\label{eq:123}
\end{equation}
Analogously to the previous case, the QR factorization of matrix $\mathbf{X}^+$ is obtained by sequentially setting to zero the $m$ rows $\mathbf{U}_{\star}$ of the right-hand side of equation \eqref{eq:123}. Specifically, new matrix $\mathbf{R}_1^+$ can be obtained as follows
\begin{equation*}
\begin{bmatrix}\mathbf{R}_1^{+}\\
\mathbf{0}_{m,p}
\end{bmatrix} = \mathbf{H}_p (p+1, p+m)\times \cdots\times \mathbf{H}_1(p+1, p+m) \widetilde{\mathbf{R}}_1,
\end{equation*}
where $\mathbf{H}_i(j,k)$, $i = 1,2, \dots, p$, is the Householder matrix with normal vector $\mathbf{v}_i(j, k) \in \mathbb{R}^{p+m}$ defined in equation (\ref{eq:householder_generic}). Eventually, the upper triangular sub-matrix $\mathbf{R}_1^+ \in \mathbb{R}^{p \times p}$ is selected. 

\subsection{Proofs of $\mathbf{R}_{1}$ updates when adding columns}

Assume that a new column is added at the end of matrix $\mathbf{X}$, i.e. $\mathbf{X}^{+}= \begin{bmatrix} \mathbf{X} & \mathbf{x}_\star  \end{bmatrix}$ where $\mathbf{X}^{+} \in \mathbb{R}^{N \times (p+1)}$, then equation \eqref{eq:addonecol} becomes 
\begin{equation*}
\begin{bmatrix} \mathbf{Q}^\top_1 \\  \mathbf{Q}^\top_2 \end{bmatrix} \mathbf{X}^{+} = \begin{bmatrix} \mathbf{R}_1 & \mathbf{z}_{\star 1} \\   \mathbf{0}_{N-p, p} & \mathbf{z}_{\star 2} \end{bmatrix} = \widetilde{\mathbf{R}}^{+},
\end{equation*}
where $\mathbf{z}_{\star 1}= \mathbf{Q}^\top_{1} \mathbf{x}_{\star}$ and $\mathbf{z}_{\star 2}= \mathbf{Q}^\top_{2} \mathbf{x}_{\star}$. The QR factorization of $\mathbf{X}^+$ can be obtained by setting to zero the last $N-p-1$ elements of the last column of $\widetilde{\mathbf{R}}^{+}$ as follows
\begin{equation*}
\mathbf{R}^{+} = \mathbf{G}_p (p+1, p+2)^{\top} \times\cdots\times \mathbf{G}_p(N-1, N)^{\top} \widetilde{\mathbf{R}}^{+} = \begin{bmatrix} \mathbf{R}_1 & \mathbf{z}_{\star 1} \\   \mathbf{0}_{p}^{\top} & r_{\star 2} \\
\mathbf{0}_{N-p-1, p} & \mathbf{0}_{N-p-1} \end{bmatrix}. 
\end{equation*}
Hence $\mathbf{R}_{1}^{+}= \begin{bmatrix} \mathbf{R}_1 & \mathbf{z}_{\star 1} \\   \mathbf{0}_{p}^{\top} & r_{\star 2} \end{bmatrix} $, however, this procedure requires the evaluation of matrices $\mathbf{Q}_{1}$ and $\mathbf{Q}_{2}$. In order to avoid such computation, one can directly compute $\mathbf{z}_{\star 1}$ and $r_{\star 2}$ by exploiting the relation
\begin{equation*}
(\mathbf{X}^+)^\top \mathbf{X}^+ = \begin{bmatrix} \mathbf{X}^{\top} \\  \mathbf{x}_{\star}^{\top}\end{bmatrix} \begin{bmatrix} \mathbf{X}^{\top} &  \mathbf{x}_{\star}^{\top}\end{bmatrix}    = (\mathbf{R}_1^{+})^\top \mathbf{R}_1^+ = \begin{bmatrix} \mathbf{R}_1^{\top} & \mathbf{0}_{p} \\ \mathbf{z}_{\star 1}^{\top} & r_{\star 2} \end{bmatrix} \begin{bmatrix} \mathbf{R}_1 & \mathbf{z}_{\star 1} \\   \mathbf{0}_{p}^{\top} & r_{\star 2} \end{bmatrix},
\end{equation*}
which implies that $ \mathbf{z}_{\star 1}$ can be computed by solving the linear system $\mathbf{R}^\top_1 \mathbf{z}_{\star 1} = \mathbf{X}^\top \mathbf{x}_{\star}$, while $r_{\star 2}$ can be computed exploiting the relation $\mathbf{z}_{\star 1}^{\top} \mathbf{z}_{\star 1} + r_{\star 2}^{2}= \mathbf{x}^\top_{\star} \mathbf{x}_{\star}$, so
\begin{equation*}
\mathbf{R}^{+}_1 = \begin{bmatrix} \mathbf{R}_1 & \mathbf{R}^{-\top}_1 \mathbf{X}^\top \mathbf{x}_{\star} \\
\mathbf{0}_{p}^{\top} & (\mathbf{x}^\top_{\star}\mathbf{x}_{\star} - \mathbf{z}_{\star 1}^{\top} \mathbf{z}_{\star 1})^{1/2} \end{bmatrix}.
\end{equation*}
An analogous situation is faced when $m$ columns are added at the end of matrix $\mathbf{X}$, that is when updating $\mathbf{R}_1$ for $\mathbf{X}^{+} = \begin{bmatrix} \mathbf{X} & \mathbf{U}_{\star}  \end{bmatrix}$, where $\mathbf{U}_{\star} \in \mathbb{R}^{N \times m}$. Then
\begin{equation*}
\begin{bmatrix} \mathbf{Q}_1^\top \\  \mathbf{Q}_2^\top \end{bmatrix} \mathbf{X}^{+} = \begin{bmatrix} \mathbf{R}_1 & \mathbf{Z}_{\star 1} \\   \mathbf{0}_{N-p, p} & \mathbf{Z}_{\star 2} \end{bmatrix} = \widetilde{\mathbf{R}}^{+},
\end{equation*}
where $\mathbf{Z}_{\star 1}=\mathbf{Q}_{1}^\top \mathbf{U}_{\star}$ and $\mathbf{Z}_{\star 2}= \mathbf{Q}_{2}^\top \mathbf{U}_{\star}$. Matrix $\mathbf{R}^{+}$ can be obtained through triangularization of matrix $\mathbf{Z}_{\star 2}$ as
\begin{align*}
\mathbf{R}^{+}&=\mathbf{G}_{p+m} (p+m, p+m+1)^{\top}\times \cdots \times\mathbf{G}_{p+m} (N-1, N)^{\top}\times\nonumber\\
&\qquad\qquad\cdots \times \mathbf{G}_{p+1} (p+1, p+2)^{\top} \times \cdots \times\mathbf{G}_{p+1}(N-1, N)^{\top} \widetilde{\mathbf{R}}^{+},
\end{align*}
or alternatively using Householder reflections. Hence $\mathbf{R}_{1}^{+}= \begin{bmatrix} \mathbf{R}_1 & \mathbf{Z}_{\star 1} \\   \mathbf{0}_{m,p} & \mathbf{R}_{\star 2} \end{bmatrix} $, but again this procedure requires the evaluation of matrices $\mathbf{Q}_{1}$ and $\mathbf{Q}_{2}$. In order to avoid such computation, $\mathbf{Z}_{\star 1}$ can be determined by solving the linear system $\mathbf{R}_1^\top \mathbf{Z}_{\star 1} = \mathbf{X}^\top \mathbf{U}_{\star}$, while $\mathbf{R}_{\star 2}$ can be computed by solving $ \mathbf{U}_{\star}^{\top} \mathbf{U}_{\star} =  \mathbf{Z}_{\star 1}^{\top} \mathbf{Z}_{\star 1} + \mathbf{R}_{\star 2}^{\top}\mathbf{R}_{\star 2}$, see Algorithm \ref{alg:thinQRaddcols}.
%
\section{Algorithms}
\label{appB}
%
This supplementary section provides a comprehensive outline of the algorithms used to update the $\mathbf{Q}$ and $\mathbf{R}$ matrices in the QR factorization following the addition or removal of rows or columns from matrix $\mathbf{X}$. For completeness, we also include the algorithm that directly updates the square matrix $\boldsymbol{\Sigma} = (\mathbf{X}^\top\mathbf{X})^{-1}$ via the Sherman-Morrison rank-one update. Additionally, we describe the algorithms for Givens rotations and Householder reflections, included here to detail the exact computational costs presented in Section \ref{appC}.\par
%
%
\begin{algorithm}[H]
\caption{Rank one update of $\mathbf{B}=(\mathbf{X}^{\top}\mathbf{X})^{-1}$  when one column is added at position $1\leq k\leq p+1$, $\mathbf{B}^{+}=\mathsf{onecolinvadd}\left(\mathbf{B}, \mathbf{X}, k,\mathbf{x}_k\right)$}
 \label{alg:onecolinv_update_add}
\begingroup
    \fontsize{9pt}{11pt}
 {\textbf{Input}: $\mathbf{B} \in \mathbb{R}^{p \times p}$, $\mathbf{X} \in \mathbb{R}^{N \times p}$,  $k \in \{1, \ldots, p+1\}$,  $\mathbf{x}_k \in \mathbb{R}^{N}$\;}
	$\mathbf{u}_1 = \mathbf{X}^{\top}\mathbf{x}_k$\;
	$\mathbf{u}_2 = \mathbf{B}\mathbf{u}_1$\;
	$ d = 1 / (\mathbf{x}_k^{\top}\mathbf{x}_k - \mathbf{u}_1^{\top}\mathbf{u}_2)$\;
	$\mathbf{u}_3 =d\mathbf{u}_2$\;
	$\mathbf{F} = \mathbf{B} + \mathbf{u}_3\mathbf{u}^{\top}_2$\;
	$\mathbf{B}^+ = \begin{bmatrix} \mathbf{F} & -\mathbf{u}_3 \\ -\mathbf{u}_3^{\top} & d \end{bmatrix}$\;
	Permute last column and last row of $\mathbf{B}^+$ to position $k$\;
\textbf{return} $\mathbf{B}^+$\;
\endgroup
\end{algorithm}
%
\begin{algorithm}[H]
     \caption{Rank one update of $\mathbf{B}=(\mathbf{X}^{\top}\mathbf{X})^{-1}$  when one column is deleted at position $1\leq k\leq p$, $\mathbf{B}^{-}=\mathsf{onecolinvdel}\left(\mathbf{B},k,\mathbf{x}_k\right)$} 
 \label{alg:onecolinv_update_del}
\begingroup
    \fontsize{9pt}{11pt}
 {\textbf{Input}: $\mathbf{B} \in \mathbb{R}^{p \times p}$,  $k \in \{1, \ldots, p\}$,  $\mathbf{x}_k \in \mathbb{R}^{N}$\;}
	Permute $k$-th column and $k$-th row of $\mathbf{B}$ to position $p$\;
	$\mathbf{F} = \mathbf{B} \left[ 1:(p-1), \; 1:(p-1) \right] $\;
	$ d = \mathbf{B}\left[p, \; p\right]$\;
	$\mathbf{u}_3 =  \mathbf{B} \left[ 1:(p-1), \;p \right]$\;
	$\mathbf{u}_2 = \mathbf{u}_3 / d$\;
	$\mathbf{B}^- = \mathbf{F}- \mathbf{u}_3 \mathbf{u}_2^{\top}$\;
\textbf{return} $\mathbf{B}^-$\;
\endgroup
\end{algorithm}
%
\begin{algorithm}[H]
 	\caption{Householder reflection, $(\tau,\mathbf{v}, \mu )=\textsf{householder}\left(a,\mathbf{x}\right)$} 
\label{alg:householder_reflection}
	\begingroup
    	\fontsize{9pt}{11pt}
 	{\textbf{Input}: $a \in \mathbb{R}$, $\mathbf{x} \in \mathbb{R}^N$}\;
 	$ s = \Vert \mathbf{x} \Vert_2^2, \quad \mathbf{v} = \begin{bmatrix}1 & \mathbf{x}^\top\end{bmatrix}^\top $\;
 	\uIf(){$(s == 0) \; \& \; (a >= 0)$} {
		$\tau = 0$\;
		$\mu=a$\;}
 	\uElseIf(){$(s == 0) \; \& \; (a < 0)$} {
 	 	$\tau = 2$\;
		$\mu=-a$\;
 	 }   
	 \Else(){ 
 		$\mu = \sqrt{s + a^2} $\;
  		\eIf() {$(a \le 0 )$} {
  			$\mathbf{v}\left[1\right] = a - \mu$\; 
		}  {
  		        $\mathbf{v}\left[1\right] = - s / (a + \mu ) $\; 
		}
 	
 $b=(\mathbf{v}[1])^{2}$, $ \tau = 2b / ( s + b), \quad \mathbf{v} = \begin{bmatrix}1 & \big(\mathbf{v}[2:(N+1)]/\mathbf{v}[1]\big)^\top\end{bmatrix}^\top$\;
 }
\textbf{return} $ (\tau,\mathbf{v}, \mu ) $\;
\endgroup
\end{algorithm}
%
%
%
\begin{algorithm}[H]
 \caption{QR decomposition with Householder reflections, $ (\mathbf{R}, \mathbf{Q}) =\textsf{householderQR}\left(\mathbf{X}\right)$} 
\label{alg:qr_householder}
\begingroup
    \fontsize{9pt}{11pt}
 {\textbf{Input}: $\mathbf{X} \in \mathbb{R}^{N \times p}$}\;
 $ \mathbf{R} = \mathbf{X}, \quad \mathbf{Q} = \mathbf{I}_N $\;
 \For(){$\left(i = 1\mbox{\; }i \le p\mbox{\; } i ++\right)$} {
 	 $(\tau,\mathbf{v}, \mu) =\textsf{householder} (\mathbf{R}\left[i, \; i \right],\mathbf{R}\left[ (i+1) : N, \; i \right] )$ as in Algorithm \ref{alg:householder_reflection}\;
 	\tcp{update $\mathbf{R}$}
 	$ \mathbf{R}\left[ i, \; i \right]=\mu $\;
 	$ \mathbf{R}\left[(i+1):N, \; i \right] = \mathbf{0} $\;
 	\If() {$(i < p)$} {
 		$ \mathbf{R}\left[ i:N, \; (i+1):p \right] = \mathbf{R}\left[ i:N, \; (i+1):p \right] - \tau\times(\mathbf{v} \mathbf{v}^{\top} \mathbf{R}\left[ i:N, \; (i+1):p \right] ) $\;
	}
	 	\tcp{update $\mathbf{Q}$}
	$ \mathbf{Q}\left[1:N, \; i:N \right] = \mathbf{Q}\left[1:N, \; i:N \right] - \tau \times ( \mathbf{Q}\left[1:N, \; i:N \right] \mathbf{v} \mathbf{v}^{\top} ) $\;
} 
\textbf{return} $ (\mathbf{R}, \mathbf{Q} ) $\;
\endgroup
\end{algorithm}
%
%
\begin{algorithm}[H]
 \caption{Givens rotation, $ (c,s ) =\textsf{givens}\left(a,b\right)$} 
\label{alg:givens_rotation}
\begingroup
    \fontsize{9pt}{11pt}
 {\textbf{Input}: $a \in \mathbb{R}$, $b \in \mathbb{R}$}\;
 \eIf(){$(b == 0)$}
{$c = 1$\; 
$s = 0$\; } { 
 \eIf{$(\left| b \right| > \left| a \right|)$} { 
 	$r = - a / b$\; 
 	$s = 1 / \sqrt{ 1 + r^2 } $\; 
   	$c = s * r $\;
	\lIf() {$(b> 0)$} {$c=-c$, $s=-s$}
} 
{
  	$ r = -b / a $\; 
	$ c = 1 / \sqrt{ 1 + r^2 } $\; 
	$ s = c * r $\; 
	\lIf() {$(a< 0)$} {$c=-c$, $s=-s$}
	}
  }
\textbf{return} $ (c,s ) $\;
\endgroup
\end{algorithm}
%
%
\begin{algorithm}[H]
\label{alg:givens_QR}
\begingroup
    \fontsize{9pt}{11pt}
 \caption{Givens QR, $ (\mathbf{R}, \mathbf{Q}) =\mathsf{givensQR}\left(\mathbf{X}\right)$} 
 {\textbf{Input}: $\mathbf{X} \in \mathbb{R}^{N \times p}$}\;
 $ \mathbf{R} = \mathbf{X}, \quad \mathbf{Q} = \mathbf{I}_N $\;
 \For(){$\left(j = 1\mbox{\; }j \le p\mbox{\; } j ++\right)$} {
	 \For(){$\left(i = N\mbox{\; }i > j\mbox{\; } i -\right)$} {
		$(c, \; s ) = \textsf{givens} (R\left[i-1, j \right], \; R\left[i, j \right] )$ as in Algorithm \ref{alg:givens_rotation}\;
 		$ \mathbf{G} = \begin{bmatrix} c & s \\ -s & c \end{bmatrix} $\;
		\tcp{update $\mathbf{R}$}
		$\mathbf{R}\left[i-1:i, j:p\right] = \mathbf{G}^\top \mathbf{R}\left[i-1:i, j:p\right]$\;
		\tcp{update $\mathbf{Q}$}
		$\mathbf{Q}\left[1:N, i-1:i\right] = \mathbf{Q}\left[1:N, i-1:i\right] \mathbf{G}$\;
	}
}
\textbf{return} $ (\mathbf{R}, \mathbf{Q}) $\;
\endgroup
\end{algorithm}
\begin{algorithm}[H]
 \caption{QR update when one row is added at position $1\leq k \leq N+1$, $\left(\mathbf{R}^{+},\mathbf{Q}^{+} \right)=\mathsf{qraddrow}\left(\mathbf{Q},\mathbf{R},k,\mathbf{x}_k\right)$} 
\label{alg:qr_add_one_row}
\begingroup
    \fontsize{9pt}{11pt}
 {\textbf{Input}: $\mathbf{Q} \in \mathbb{R}^{N \times N}$, $\mathbf{R} \in \mathbb{R}^{N \times p}$, $k \in \{1, \ldots, N+1\}$, $\mathbf{x}_k \in \mathbb{R}^{p}$}\;
 \eIf{$k==N+1$} {$\mathbf{Q}^+ = \begin{bmatrix} \mathbf{Q} & 0 \\ 0 & 1 \end{bmatrix} $;}
 {	$\mathbf{Q}^+ = \begin{bmatrix} \mathbf{Q} \left[ 1:(k-1), \; \right] & 0 \\ 0 & 1 \\  \mathbf{Q} \left[ k:N, \; \right] & 0 \end{bmatrix} $;}
	 \For(){$\left( i = 1\mbox{\; } i \le p\mbox{\; } i ++ \right)$} {
 	$(c, \; s ) = \textsf{givens} (\mathbf{R}\left[i, \; i \right], \; \mathbf{x}_k\left[  i \right] )$ as in Algorithm \ref{alg:givens_rotation}\;
		\tcp{update $\mathbf{R}$ and $\mathbf{x}_k$ }
 	$ \mathbf{R}\left[ i, \; i \right] = c*\mathbf{R}\left[ i, \; i \right] - s * \mathbf{x}_k\left[  i \right] $\;
 	\If() {$(i < p)$} {
 		$ \mathbf{R}\left[ i, \; (i+1):p \right] = c * \mathbf{R}\left[ i, \; (i+1):p \right] - s*\mathbf{x}_k\left[  (i+1):p \right]   $\;
		$ \mathbf{x}_k\left[  (i+1):p \right] = s*\mathbf{R}\left[ i, \; (i+1):p \right] + c* \mathbf{x}_k\left[  (i+1):p \right]$\;
	}
		\tcp{update $\mathbf{Q}$}
	$ \mathbf{Q}^+\left[, \; i \right] = c* \mathbf{Q}^+\left[, \; i \right] - s* \mathbf{Q}^+\left[, \; N+1 \right]$\;
	$ \mathbf{Q}^+\left[, \; N+1 \right] = s* \mathbf{Q}^+\left[, \; i \right] + c* \mathbf{Q}^+\left[, \; N+1 \right]$\;
	}
	$\mathbf{R}^+ = \begin{bmatrix} \mathbf{R} \\ \mathbf{0}_{1 \times p} \end{bmatrix}$\;
	\textbf{return} $ (\mathbf{R}^{+}, \;\; \mathbf{Q}^{+} ) $\;
\endgroup
\end{algorithm}
%
\begin{algorithm}[H]
 \caption{QR update when one row is removed at position $1\leq k\leq N$, $\left(\mathbf{R}^{-},\mathbf{Q}^{-} \right)=\mathsf{qrdelrow}\left(\mathbf{Q},\mathbf{R},k\right)$} 
\label{alg:qr_delete_one_row}
\begingroup
    \fontsize{9pt}{11pt}
 {\textbf{Input}: $\mathbf{Q} \in \mathbb{R}^{N \times N}$, $\mathbf{R} \in \mathbb{R}^{N \times p}$, $k \in \{1, \ldots, N\}$}\;
	$\mathbf{q} = \mathbf{Q}\left[ k, \;\; \right] $\; 
	\lIf(){$(k \ne 1)$} {
		$ \mathbf{Q}\left[ 2:k, \;  \right] =  \mathbf{Q}\left[ 1:(k-1), \; \right] $
	}
	 \For(){$ \left( i = N\mbox{\; } i > 1\mbox{\; } i - - \right)$} {
 		$(c, \; s ) = \textsf{givens} (\mathbf{q}\left[i-1 \right], \mathbf{q}\left[  i  \right] )$ as in Algorithm \ref{alg:givens_rotation}\;
 		$ \mathbf{G} = \begin{bmatrix} c & s \\ -s & c \end{bmatrix} $\;
 			\tcp{update $ \mathbf{q} $} 
		$ \mathbf{q}\left[ i-1 \right] = c*\mathbf{q}\left[i-1\right] - s* \mathbf{q}\left[  i  \right]$\;
 			\tcp{update $\mathbf{R}$}
 		\lIf() {$(i-1 \le p)$} {
 			$ \mathbf{R}\left[ (i-1) : i, \; (i-1):p \right] = \mathbf{G}^{\top} \mathbf{R}\left[ (i-1) : i, \; (i-1):p \right]$
		}
			\tcp{update $\mathbf{Q}$}
		\lIf() {$(i-1 > 1)$} {
 			$ \mathbf{Q}\left[ 2 : N, \; (i-1):i \right] = \mathbf{Q}\left[ 2 : N, \; (i-1):i \right]\mathbf{G} $
		}
	}
	$\mathbf{R}^{-} = \mathbf{R}\left[ 2:N, \;\; \right], \quad \mathbf{Q}^{-} = \mathbf{Q}\left[ 2:N, \;2:N \right]$\;
	\textbf{return} $ (\mathbf{R}^{-},\mathbf{Q}^{-} ) $\;
\endgroup
\end{algorithm}
%
%
\begin{algorithm}[H]
 \caption{QR update when one column is added at position $1\leq k\leq p+1$, $\left(\mathbf{R}^{+},\mathbf{Q}^{+} \right)=\mathsf{qraddcol}\left(\mathbf{Q},\mathbf{R},k,\mathbf{x}_k\right)$} 
\label{alg:qr_add_one_col}
\begingroup
    \fontsize{9pt}{11pt}
 {\textbf{Input}: $\mathbf{Q} \in \mathbb{R}^{N \times N}$, $\mathbf{R} \in \mathbb{R}^{N \times p}$, $k \in \{1, \ldots, p+1\}$, $\mathbf{x}_k \in \mathbb{R}^{N}$}\;
	$\mathbf{v}_k = \mathbf{Q}^{\top} \mathbf{x}_k$\;
	\eIf{$k==p+1$}{$\mathbf{R}^+ = \begin{bmatrix} \mathbf{R}  &\mathbf{v}_k   \end{bmatrix} $\;}{$\mathbf{R}^+ = \begin{bmatrix} \mathbf{R} \left[, \;\; 1:(k-1) \right] &\mathbf{v}_k & \mathbf{R} \left[, \; \;k:p\right]  \end{bmatrix} $\;}
	 \For(){$\left( i = N\mbox{\; } i > k\mbox{\; } i - - \right)$} {
 		$(c, \; s ) = \textsf{givens} (\mathbf{v}_k\left[i-1 \right], \; \mathbf{v}_k\left[  i  \right] )$ as in Algorithm \ref{alg:givens_rotation}\;
 		$ \mathbf{G} = \begin{bmatrix} c & s \\ -s & c \end{bmatrix} $\;
			\tcp{update $\mathbf{v}_k$}
		$\mathbf{v}_k\left[ i-1 \right] = c*\mathbf{v}_k\left[ i-1 \right] - s * \mathbf{v}_k\left[ i  \right] $\;
		$\mathbf{v}_k\left[ i \right] = 0$\;
 			\tcp{update $\mathbf{R}$}
 		\lIf() {$(i-1 \le p)$} {
 			$ \mathbf{R}\left[ (i-1) : i, \; (i-1):p \right] = \mathbf{G}^{\top} \mathbf{R}\left[ (i-1) : i, \; (i-1):p \right]$	}
		\tcp{update $\mathbf{Q}$}
	$ \mathbf{Q}\left[, \; (i-1):i \right] = \mathbf{Q}\left[, \; (i-1):i \right]\mathbf{G}$\;
	}
	$\mathbf{R}^{+} \left[, \;\; k \right] = \mathbf{v}_k$\;
	\textbf{return} $ (\mathbf{R}^{+}, \;\; \mathbf{Q} ) $\;
\endgroup
\end{algorithm}
%
%
\begin{algorithm}[H]
 \caption{QR update when one column is deleted at position $1\leq k\leq p$, $\left(\mathbf{R}^{-},\mathbf{Q}^{-} \right)=\mathsf{qrdelcol}\left(\mathbf{Q},\mathbf{R},k\right)$} 
\label{alg:qr_delete_one_col}
\begingroup
    \fontsize{9pt}{11pt}
 {\textbf{Input}: $\mathbf{Q} \in \mathbb{R}^{N \times N}$, $\mathbf{R} \in \mathbb{R}^{N \times p}$, $k \in \{1, \ldots, p\}$}\;
	\eIf(){$(k == p)$} {
		\textbf{return} $( \mathbf{R}^{-} =  \mathbf{R}\left[ ,\; \;1:(p-1) \right]) $\;
	} {
		$ \mathbf{R}\left[ , \;k:(p-1) \right] = \mathbf{R}\left[ , \;(k+1):p \right] $\;
	}
	 \For(){$\left(i = k\mbox{\; } i < p\mbox{\; } i + + \right)$} {
 		$(c, \; s ) =  \textsf{givens} (\mathbf{R}\left[i, \;i \right], \mathbf{R}\left[  i+1, \; i  \right] )$ as in Algorithm \ref{alg:givens_rotation}\;
 		$ \mathbf{G} = \begin{bmatrix} c & s \\ -s & c \end{bmatrix} $\;
  			\tcp{update $\mathbf{R}$}
		$ \mathbf{R}\left[ i, \;i \right] = c*\mathbf{R}\left[i , \;i \right] - s*\mathbf{R}\left[i+1 , \;i \right]$\;
		$ \mathbf{R}\left[i+1 , \; i \right] = 0 $\;
		\If(){$i<p-1$}{
		$ \mathbf{R}\left[ i : i+1, \; (i+1):(p-1) \right] = \mathbf{G}^{\top} \mathbf{R}\left[ i : i+1, \; (i+1):(p-1) \right]$\;
		}
			\tcp{update $\mathbf{Q}$}
 		$ \mathbf{Q}\left[ , \;i : (i+1) \right] = \mathbf{Q}\left[, \; i:(i+1) \right]\mathbf{G} $\;
	}

	$\mathbf{R}^{-} = \mathbf{R}\left[, \;\; 1:(p-1) \right]$\;

	\textbf{return} $ (\mathbf{R}^{-}, \;\; \mathbf{Q}) $\;
\endgroup
\end{algorithm}
%
\begin{algorithm}[H]
 \caption{QR update when a block of $m\geq 2$ rows is added from position $1\leq k\leq N+1$ to position $k+m-1$, $ (\mathbf{R}^{+},\mathbf{Q}^{+} ) =\mathsf{qraddblockrows}\left(\mathbf{Q},\mathbf{R},k,\mathbf{U}_k\right)$} 
\label{alg:qrupdate_m_rows_added}
\begingroup
    \fontsize{9pt}{11pt}
 {\textbf{Input}: $\mathbf{Q} \in \mathbb{R}^{N \times N}$, $\mathbf{R} \in \mathbb{R}^{N \times p}$, $k \in \{1, \ldots, N+1\}$, $\mathbf{U}_k \in \mathbb{R}^{m \times p}$ }\;
	 \eIf{$k==N+1$} {$\mathbf{Q}^+ = \begin{bmatrix} \mathbf{Q} & \mathbf{0}_{N\times m} \\  \mathbf{0}_{m\times N} & \mathbf{I}_m  \end{bmatrix} $;}
 {	$\mathbf{Q}^+ = \begin{bmatrix} \mathbf{Q} \left[ 1:(k-1), \; \right] & \mathbf{0}_{(k-1)\times m} \\ \mathbf{0}_{m\times N} & \mathbf{I}_m \\  \mathbf{Q} \left[ k:N, \; \right] & \mathbf{0}_{(N-k+1)\times m} \end{bmatrix} $\;}
	\For(){$\left( i = 1\mbox{\; } i \le p\mbox{\; } i ++ \right)$} {
	
	\eIf(){$(i < p)$} {
 	$(\tau, \; \mathbf{v}, \; \mu ) = \textsf{householder} (\mathbf{R}\left[i, \; i \right], \; \mathbf{U}_k\left[  , \; i \right] )$, as in Algorithm \ref{alg:householder_reflection}\;
		\tcp{update $\mathbf{R}$}
	$\mathbf{v}_1 = \tau * \mathbf{v}[2:(m+1)]$\;
 	$ \mathbf{R}\left[ i, \; i \right] = \mu$\;
 		\tcp{update $\mathbf{U}_k$}
		$\mathbf{r} = \tau * \left(\mathbf{R}\left[ i, \; (i+1):p \right] + \mathbf{v}[2:(m+1)]^{\top} \mathbf{U}_k \left[ , \; (i+1):p \right]\right)$\;
		 $ \mathbf{R}\left[ i, \; (i+1):p \right] = \mathbf{R}\left[ i, \; (i+1):p \right] - \mathbf{r}$\;
		$\mathbf{U}_k\left[, \; (i+1):p \right] = \mathbf{U}_k\left[, \; (i+1):p \right] - \mathbf{v}[2:(m+1)]  \mathbf{r}$
		
	}{
	$ \mathbf{R}\left[ i, \; i \right] = \sqrt{(\mathbf{R}[i,i])^2 + \Vert \mathbf{U}_k\left[, i\right]\Vert_2^2}$\;
	}
	\tcp{update $\mathbf{Q}$}
	$\mathbf{q} =  \tau * \left( \mathbf{Q}^+\left[, \;\; i \right] +  \mathbf{Q}^+\left[, \; (N+1):(N+m) \right] \mathbf{v}[2:(m+1)]\right)$\;
	$ \mathbf{Q}^{+}\left[, \;\; i \right] = \mathbf{Q}^+\left[, \;\; i \right] -\mathbf{q}$\;
	$ \mathbf{Q}^{+}\left[, \; (N+1):(N+m) \right] = \mathbf{Q}^+\left[, \; (N+1):(N+m)\right] - \mathbf{q}\mathbf{v}[2:(m+1)]^{\top}$\;
	}
Compute:	$\mathbf{R}^+ = \begin{bmatrix} \mathbf{R} \\ \mathbf{0}_{m \times p} \end{bmatrix}$\;
	\textbf{return} $ (\mathbf{R}^{+},\mathbf{Q}^{+} ) $\;	
\endgroup
\end{algorithm}
%
%
%
\begin{algorithm}[H]
\caption{QR update when a block of $2\leq m<N$ rows is deleted from position $1\leq k\leq N-m+1$ to $k+m-1$, $ (\mathbf{R}^{+},\mathbf{Q}^{+} ) =\mathsf{qrdelblockrows}\left(\mathbf{Q},\mathbf{R},k,m\right)$} 
\label{alg:qrupdate_m_rows_deleted}
	\begingroup
    	\fontsize{9pt}{11pt}
	 {\textbf{Input}: $\mathbf{Q} \in \mathbb{R}^{N \times N}$, $\mathbf{R} \in \mathbb{R}^{N \times p}$, $k \in \{1, \ldots, N-m+1\}$, $m \in \{2, \ldots N-1\}$}\;
	
	Set $\mathbf{W} = \mathbf{Q}\left[k:(k+m-1), \right]$\;
	\lIf(){$(k \ne 1)$} {$\mathbf{Q} \left[ (m+1):(k+m-1), \right] = \mathbf{Q} \left[ 1:(k-1), \right]$}
	
	 \For(){$ \left( j = 1\mbox{\; } j \le m\mbox{\; } j++ \right)$} {
	 	\For(){$ \left( i = N-1\mbox{\; } i \ge j\mbox{\; } i- \right)$} {
	 		$(c, \; s ) =  \mbox{\textsf{givens}} (\mathbf{W}\left[j, \; i \right], \mathbf{W}\left[  j ,\; i+1\right] )$ as in Algorithm \ref{alg:givens_rotation}\;
 			$ \mathbf{G} = \begin{bmatrix} c & s \\ -s & c \end{bmatrix} $\;
 			\tcp{update $ \mathbf{W} $} 
			$ \mathbf{W}\left[ j, \; i \right] = c * \mathbf{W}\left[ j, \; i \right] - s * \mathbf{W}\left[ j, \; i+1 \right]$\;
			\If{$(j < m)$} {$\mathbf{W}\left[ (j+1):m, \; i:(i+1) \right] = \mathbf{W}\left[ (j+1):m, \; i:(i+1) \right] \mathbf{G}$}
			\tcp{update $\mathbf{R}$}
 			\If(){$(i \le p + j - 1)$} {
 				$\mathbf{R} \left[ i:(i+1), \; (i-j+1):p\right] = \mathbf{G}^{\top} \mathbf{R} \left[ i:(i+1), \; (i-j+1):p\right]$\;
 			}
				\tcp{update $\mathbf{Q}$}
			$\mathbf{Q} \left[ (m+1):N, \; i:(i+1)\right] = \mathbf{Q} \left[ (m+1):N, \; i:(i+1)\right]  \mathbf{G}$\;
	 	}
	}
	$\mathbf{R}^{-} = \mathbf{R}\left[ (m+1):N, \; \right], \quad \mathbf{Q}^{-} = \mathbf{Q}\left[ (m+1):N, \;(m+1):N \right]$\;
	\textbf{return} $ (\mathbf{R}^{-},\mathbf{Q}^{-} ) $\;
	\endgroup
\end{algorithm}
%
\begin{algorithm}[H]
 	\caption{QR update when a block of $m\geq 2$ columns is added from position $1\leq k\leq p+1$ to $k+m-1$, $(\mathbf{R}^{+},\mathbf{Q}^{+}) =\mathsf{qraddblockcols}\left(\mathbf{Q},\mathbf{R},k,\mathbf{U}_k\right)$} 
\label{alg:qraddblockcolumns}
	\begingroup
	\fontsize{9pt}{11pt}
 	{\textbf{Input}: $\mathbf{Q} \in \mathbb{R}^{N \times N}$, $\mathbf{R} \in \mathbb{R}^{N \times p}$, $k \in \{ 1, \ldots, p+1\}$, $\mathbf{U}_k \in \mathbb{R}^{N \times m}$}\;
 	
 	\tcp{add columns}
	Set $\mathbf{V}_k = \mathbf{Q}^{\top} \mathbf{U}_k$\;
	\For(){$\left( j = 1\mbox{\; } j \le m\mbox{\; } j + + \right)$}{
		\For(){$\left( i = N\mbox{\; } i \geq k+j\mbox{\; } i - - \right)$} {
 			$(c,s ) = \textsf{givens} (\mathbf{V}\left[ i-1 , \; j \right], \;\mathbf{V}\left[ i , \; j \right])$ as in Algorithm \ref{alg:givens_rotation}\;
 			$ \mathbf{G} = \begin{bmatrix} c & s \\ -s & c \end{bmatrix}$\;
				\tcp{update $\mathbf{V}_k$}
			$\mathbf{V}\left[ i-1 , \; j \right] = c*\mathbf{V}\left[ i-1 , \; j\right] - s * \mathbf{V}\left[ i ,\; j\right]$\;
			$\mathbf{V}\left[ i, \; j\right] = 0$\;
 			\lIf() {$(j < m)$} {
 				$ \mathbf{V}\left[ (i-1) : i, \; (j+1):m \right] = \mathbf{G}^{\top} \mathbf{R}\left[ (i-1) : i, \; (j+1):m \right]$}
				\tcp{update $\mathbf{R}$}
			\If(){$ (i \le p + j)$} {
 				$\mathbf{R} \left[ (i-1):i, \; (i-j):p\right] = \mathbf{G}^{\top} \mathbf{R} \left[ (i-1):i, \; (i-j):p\right]$\;
			}
			\tcp{update $\mathbf{Q}$}
			$ \mathbf{Q}\left[, \; (i-1):i \right] = \mathbf{Q}\left[, \; (i-1):i \right]\mathbf{G}$\;
		}
	}
	$\mathbf{R}^{+} = \begin{bmatrix} \mathbf{R} \left[, \; 1:(k-1) \right] &\mathbf{V}_k & \mathbf{R} \left[, \; k:p\right]  \end{bmatrix}$\;
	\textbf{return} $ (\mathbf{R}^{+},\mathbf{Q} ) $\;
	\endgroup
\end{algorithm}
%
\begin{algorithm}[H]
 \caption{QR update when a block of $2\leq m< p$ columns is deleted from position $1\leq k\leq p-m+1$ to $k+m-1$, $ (\mathbf{R}^{-},\mathbf{Q}^{-} ) =\mathsf{qrdelblockcols}\left(\mathbf{Q},\mathbf{R},k,m\right)$} 
\label{alg:qrdeleteblockcolumns}
\begingroup
    \fontsize{9pt}{11pt}    
 {\textbf{Input}: $\mathbf{Q} \in \mathbb{R}^{N \times N}$, $\mathbf{R} \in \mathbb{R}^{N \times p}$, $k \in \{1, \ldots, p-m+1 \}$, $m \in \{2, \ldots, p-1 \}$}\;
	 	\tcp{remove columns}
	 \eIf{ $(k == p-m+1 )$ } {
	 	\textbf{return} $(\mathbf{R}^{-} = \mathbf{R} \left[, \; 1:(k-1) \right], \;\;\mathbf{Q})$\;
	} {
		$\mathbf{R} \left[, \; k:(p-m) \right] = \mathbf{R} \left[, \; (k+m):p \right]$\;
	}
	 \For(){$\left( i = k\mbox{\; } i \le p-m\mbox{\; } i ++ \right)$} {
 	$(\tau,\mathbf{v}, \mu ) = \textsf{householder} (\mathbf{R}\left[i, \; i \right], \; \mathbf{R}\left[  (i+1):(i+m), \; i \right] )$ as in Algorithm \ref{alg:householder_reflection}\;
		\tcp{update $\mathbf{R}$}
	$ \mathbf{R}\left[ i, \; i \right] = \mu$\;
	$ \mathbf{R}\left[ (i+1):(i+m), \; i\right] = \mathbf{0}_{m}$\;
	$\mathbf{v}_1 = \tau * \mathbf{v}$\;
	\lIf(){$(i < p-m)$}{
		$\mathbf{R}\left[ i:(i+m), \; (i+1):(p-m)\right] =  \mathbf{R}\left[ i:(i+m), \; (i+1):(p-m)\right] -\mathbf{v}_1(\mathbf{v}^\top\mathbf{R}\left[ i:(i+m), \; (i+1):(p-m)\right])$
	}
 		\tcp{update $\mathbf{Q}$}
	$ \mathbf{Q}\left[, \; i:(i+m) \right] = \mathbf{Q}\left[, \; i:(i+m) \right] - (  \mathbf{Q}\left[, \; i:(i+m) \right] \mathbf{v}_1) \mathbf{v}^{\top}$\;
	}
	$\mathbf{R}^{-} = \mathbf{R}\left[, \; 1:(p-m) \right]$\;
	\textbf{return} $ (\mathbf{R}^{-},\mathbf{Q} ) $\;
	\endgroup
\end{algorithm}
%
\begin{algorithm}[H]
	\begingroup
    	\fontsize{9pt}{11pt}
 	 \caption{Apply either Givens rotation or Householder reflection to column $i$, $\left(\mathbf{R}, \mathbf{Q}\right) = \mathsf{qrstep}\left(\mathbf{Q}, \mathbf{R}, i, a\right)$} 
	\label{alg:QRstep}
 	{\textbf{Input}: $\mathbf{Q} \in \mathbb{R}^{N \times N}$, $\mathbf{R} \in \mathbb{R}^{N \times l}$, $i \in \left\{1, \dots, l \right\}$, $a \in \{1, \ldots, N-i\}$}\;
	
	\eIf(){$\left( a > 1 \right)$} {
		\tcp{Householder reflection}
		$\left(\tau,\mathbf{v},  \mu \right) = \textsf{householder} \left(\mathbf{R} \left[i, i\right], \mathbf{R} \left[(i+1):(i+a), i\right]\right)$ as in Algorithm \ref{alg:householder_reflection}\;
		
		\tcp{update $\mathbf{R}$}
		$\mathbf{v}_s = \tau * \mathbf{v}$\;
		$\mathbf{R}_1\left[i, i \right] = \mu$\;
 		$\mathbf{R}_1\left[(i+1):(i+a), i \right] = \mathbf{0}_a$\;
 		\lIf(){$\left( i < l \right)$}{
 			$\mathbf{R}\left[i:(i+a), (i+1):l \right] = \mathbf{R}\left[i:(i+a), (i+1):l \right] - \mathbf{v}_s\left(\mathbf{v}^\top\mathbf{R}\left[i:(i+a), (i+1):l \right]\right)$}
 			
 		\tcp{update $\mathbf{Q}$}
 		$\mathbf{Q} \left[, i:(i+a)\right] = \mathbf{Q} \left[, i:(i+a)\right] - \mathbf{Q} \left[, i:(i+a)\right] \mathbf{v}_s \mathbf{v}^\top$\;
 	
 	} {	
 		$(c,  s ) =  \textsf{givens} \left(\mathbf{R}\left[i, i\right], \mathbf{R}\left[  i+1, i  \right] \right)$ as in Algorithm \ref{alg:givens_rotation}\;
 		$ \mathbf{G} = \begin{bmatrix} c & s \\ -s & c \end{bmatrix} $\;
 		
 		\tcp{update $\mathbf{R}$}
		$ \mathbf{R}\left[ i, i \right] = c*\mathbf{R}\left[i , i \right] - s*\mathbf{R}\left[i+1 , i \right]$; \\
		$ \mathbf{R} \left[i+1 , i \right] = 0 $\;

		$ \mathbf{R} \left[ i : (i+1), (i+1):l \right] = \mathbf{G}^{\top} \mathbf{R}\left[ i : (i+1), (i+1):l \right]$\;
		
		\tcp{update $\mathbf{Q}$}
 		$\mathbf{Q} \left[, i:(i+1)\right] = \mathbf{Q} \left[, i:(i+1)\right] \mathbf{G}$\;

	}
	\textbf{return} $\left( \mathbf{R}, \mathbf{Q}\right) $\;
	\endgroup
\end{algorithm}
%
\begin{algorithm}[H]
	\begingroup
    	\fontsize{9pt}{11pt}
 	 \caption{Delete $m$ non-adjacent columns, $\left(\mathbf{R}^{-}, \mathbf{Q}^{-}\right) = \mathsf{qrdelblockcols\_nonadj}\left(\mathbf{Q}, \mathbf{R}, \mathbf{k}\right)$} 
	\label{alg:QRdeletenonadjcols}
 	{\textbf{Input}: $\mathbf{Q} \in \mathbb{R}^{N \times N}$, $\mathbf{R} \in \mathbb{R}^{N \times p}$, $\mathbf{k}[i] \in \{1, \ldots, p\}, i=1, \ldots, m, \, k[i] < k[j] \, \forall \, i<j, i,j =1, \ldots, m$}\;
		\tcp{delete $m$ columns}
		\lIf(){$\left(m = 1\right)$}{ \textbf{return} $ \textsf{qrdelcol}\left(\mathbf{Q},\mathbf{R}, \mathbf{k}\left[1\right]\right)$}
	\lIf(){$\left(\left(\mathbf{k}\left[m\right] - \mathbf{k}\left[1\right]\right) = (m-1)\right)$}{ \textbf{return} $\textsf{qrdelblockcols}\left(\mathbf{Q}, \mathbf{R}, \mathbf{k}\left[1\right], m\right)$}
	$\mathbf{e} = 1:p$\;
	$\overline{\mathbf{k}} = \mathbf{e} \setminus \mathbf{k}$\;
	$l = \overline{\mathbf{k}}[p-m]$\;
	$q=m-(p-l)$\;
	$\mathbf{k}=\mathbf{k}[1:q]$\;
	$\mathbf{R} = \mathbf{R} \left[1:l, 1:l\right]$\;
	\lIf(){$\left(q = 1\right)$}{ \textbf{return} $ \textsf{qrdelcol}\left(\mathbf{Q},\mathbf{R}, \mathbf{k}\left[1\right]\right)$}
	\lIf(){$\left(\left(\mathbf{k}\left[q\right] - \mathbf{k}\left[1\right]\right) = (q-1)\right)$}{ \textbf{return} $\textsf{qrdelblockcols}\left(\mathbf{Q}, \mathbf{R}, \mathbf{k}\left[1\right], q\right)$}
	
		\tcp{delete columns}
	$\mathbf{R} = \mathbf{R} \left[, \overline{\mathbf{k}}\right]$\;	
	\tcp{compute $\mathbf{a}[1]$}
	$\overline{\mathbf{k}} = \overline{\mathbf{k}}\left[\mathbf{k}\left[1\right]:(l-q)\right]$\;
	$\mathbf{a}[1] = \overline{\mathbf{k}}[1] - \mathbf{k}[1]$\;
	
	\tcp{update $\mathbf{Q}$ and $\mathbf{R}$}
	\For(){$\left( i = 1 \mbox{\; } i \le (l-q-\mathbf{k}[1]+1) \mbox{\; } i ++ \right)$} {		
		$\left(\mathbf{Q}, \mathbf{R}\right) = \textsf{qrstep}\left(\mathbf{Q}, \mathbf{R}, i+\mathbf{k}[1]-1, \mathbf{a}[i]\right)$\;
		\lIf(){$\left(i < (l-q-\mathbf{k}[1]+1)\right)$}{$\mathbf{a}[i+1] = \mathbf{a}[i] + \left(\overline{\mathbf{k}}[i+1] - \overline{\mathbf{k}}[i]\right) - 1$}
	}

	\textbf{return} $\left(\mathbf{R}, \mathbf{Q} \right) $\;
	\endgroup
\end{algorithm}
%
\begin{algorithm}[H]
\fontsize{9pt}{11pt}
\caption{R update when a row is added at position $k=N+1$, $\mathbf{R}_1^{+} = \mathsf{thinqraddrow}\left(\mathbf{R}_1, \mathbf{u}\right)$}
\label{alg:thinQRaddrow}{\textbf{Input}: $\mathbf{R}_1 \in \mathbb{R}^{p \times p}$, $\mathbf{u} \in \mathbb{R}^{p}$ \;}
	
	\tcp{add one row}
	\For(){$\left( i = 1\mbox{\; } i \le p\mbox{\; } i ++ \right)$} {
		$(c, s ) = \textsf{givens} (\mathbf{R}_1\left[i, i \right], \mathbf{u} \left[  i \right] )$ as in Algorithm \ref{alg:givens_rotation}\;
 	
 		\tcp{update $\mathbf{R}_1$ and $\mathbf{u}$}
		$ \mathbf{R}_1\left[ i, i \right] = c*\mathbf{R}_1\left[ i, i \right] - s * \mathbf{u}\left[  i \right] $\;
 		\If() {$(i < p)$} {
 			$ \mathbf{r}_i =  \mathbf{R}_1\left[ i, (i+1):p \right]$\;
 			$ \mathbf{R}_1\left[ i, (i+1):p \right] = c * \mathbf{r}_i - s*\mathbf{u} \left[  (i+1):p \right]   $\;
 			$ \mathbf{u} \left[  (i+1):p \right] = s* \mathbf{r}_i + c* \mathbf{u} \left[  (i+1):p \right]$\;
		}
	}
	\textbf{return} $ \mathbf{R}_1 $\;
\end{algorithm}
%
%
\begin{algorithm}[H]
    	\fontsize{9pt}{11pt}
 	\caption{R update when a row is deleted at position $1\leq k\leq N$, $\mathbf{R}_1^{-} = \mathsf{thinqrdelrow}\left(\mathbf{R}_1, \mathbf{u}_k\right)$} 
 	\label{alg:thinQRdeleterow}
 	{\textbf{Input}: $\mathbf{R}_1 \in \mathbb{R}^{p \times p}$, $\mathbf{u}_k \in \mathbb{R}^{p}$}\;
 	\tcp{delete one row}
 	
 	\For(){$\left( i = 1\mbox{\; } i \le p\mbox{\; } i ++ \right)$} {
 		\tcp{update $\mathbf{R}_1$ and $\mathbf{u}_k$}
 		$ r = \mathbf{R}_1 \left[i, i\right] $\; 
 		$ \mathbf{R}_1 \left[i, i\right]  = \sqrt{ \left\vert (\mathbf{R}_1\left[i, i\right])^2 - (\mathbf{u}_k\left[i\right])^2\right\vert}$\;
 		
 		\If() {$\left(i < p \right)$}  {
 			$c = \mathbf{R}_1\left[i, i\right] / r$\;
 			$s = -\mathbf{u}_k\left[i\right] / r$\;
 		
 			$\mathbf{R}_1 \left[i, (i+1):p \right] = \left( \mathbf{R}_1\left[i, (i+1):p\right] + s * \mathbf{u}_k\left[(i+1):p\right] \right) / c$\;
 			$\mathbf{u}_k \left[(i+1):p\right] = s *\mathbf{R}_1 \left[i, i:p \right] + c * \mathbf{u}_k \left[(i+1):p\right]$\;
 		}
	}
	\textbf{return} $ \mathbf{R}_1 $\;
\end{algorithm}
%
%
\begin{algorithm}[H]
    	\fontsize{9pt}{11pt}
 	\caption{R update when a column is added at position $k=p+1$, $\mathbf{R}_1^{+} = \mathsf{thinqraddcol}\left(\mathbf{R}_1, \mathbf{X}, \mathbf{u}\right)$} 
 	\label{alg:thinQRaddcol}
 	{\textbf{Input}: $\mathbf{R}_1 \in \mathbb{R}^{p \times p}$, $\mathbf{X} \in \mathbb{R}^{N \times p}$, $\mathbf{u} \in \mathbb{R}^{N}$}\;
 	\tcp{add one column}
 	
 	 Solve $\mathbf{R}_1^\top \mathbf{r}_{12} = \mathbf{X}^\top \mathbf{u}$ with respect to $\mathbf{r}_{12}$ with forward substitution algorithm\;
 	 $\mathbf{R}_1 = \begin{bmatrix} \mathbf{R}_1 & \mathbf{r}_{12} \\ \mathbf{0}_{1 \times p} & 0 \end{bmatrix} $\;  
 	 \tcp{update $\mathbf{R}_1$}
 	 $\mathbf{R}_1 \left[p+1, p+1\right] = \Vert \mathbf{u} \Vert_{2}^{2} - \Vert \mathbf{r}_{12} \Vert_{2}^{2}$\;
	$\mathbf{R}_1 \left[p+1, p+1\right] = \sqrt{\vert \mathbf{R}_1 \left[p+1, p+1\right]\vert}$\;
	\textbf{return} $ \mathbf{R}_1 $\;
\end{algorithm}
%
%
\begin{algorithm}[H]
    	\fontsize{9pt}{11pt}
 	\caption{R update when a column is deleted at position $1\leq k\leq p$, $\mathbf{R}_1^{-} = \mathsf{thinqrdelcol}\left(\mathbf{R}_1, k \right)$} 
 	\label{alg:thinQRdeletecol}
 	{\textbf{Input}: $\mathbf{R}_1 \in \mathbb{R}^{p \times p}$, $k \in \{1, \ldots, p\}$}\;
 	\tcp{delete one column}
 	\lIf() {$(k = p)$} {\textbf{return}  $\mathbf{R}_1 \left[1:(p-1),  1:(p-1) \right] $}
	$\mathbf{R}_1\left[, k:(p-1) \right] = \mathbf{R}_1 \left[, (k+1):p \right]$\;
	\For(){$\left(i = k \mbox{\; } i < p\mbox{\; } i + + \right)$} {
 		$(c,  s ) =  \textsf{givens} \left(\mathbf{R}_1\left[i, i\right], \mathbf{R}_1\left[  i+1, i  \right] \right)$ as in Algorithm \ref{alg:givens_rotation}\;
 		$ \mathbf{G} = \begin{bmatrix} c & s \\ -s & c \end{bmatrix} $\;
 		\tcp{update $\mathbf{R}_1$}
		$ \mathbf{R}_1\left[ i, i \right] = c*\mathbf{R}_1\left[i , i \right] - s*\mathbf{R}_1\left[i+1 , i \right]$\;
		$ \mathbf{R}_1\left[i+1 , i \right] = 0 $\;
		\If(){$(i < p-1)$}{
		$ \mathbf{R}_1\left[ i : i+1, (i+1):(p-1) \right] = \mathbf{G}^{\top} \mathbf{R}_1\left[ i : i+1, (i+1):(p-1) \right]$\;}
	}
	\textbf{return} $\mathbf{R}_1= \mathbf{R}_1\left[1:(p-1),  1:(p-1) \right] $\;
\end{algorithm}
%
%
\begin{algorithm}[H]
    	\fontsize{9pt}{11pt}
 	\caption{R update when $m\geq 2$ rows are added from position $N+1$ to $N+m$, $\mathbf{R}_1^{+} = \mathsf{thinqraddblockrows}\left(\mathbf{R}_1, \mathbf{U}\right)$} 
 	\label{alg:thinQRaddrows}
 	{\textbf{Input}: $\mathbf{R}_1 \in \mathbb{R}^{p \times p}$, $\mathbf{U} \in \mathbb{R}^{m \times p}$ }\;
 	\tcp{add $m$ rows}
 	\For(){$\left( i = 1\mbox{\; } i \le p-1\mbox{\; } i ++ \right)$} {
		 \tcp{update $\mathbf{R}_1$ and $\mathbf{U}$}
			$(\tau, \mathbf{v}, \mu ) = \textsf{householder} (\mathbf{R}_1\left[i, i \right], \mathbf{U} \left[  , i \right] )$ as in Algorithm \ref{alg:householder_reflection}\;
			$\mathbf{R}_1\left[i, i \right] = \mu$\;
			$ \mathbf{l} =  \tau * \left( \mathbf{R}_1\left[ i, (i+1):p\right]+ (\mathbf{v}\left[2:(m+1)\right])^\top \mathbf{U} \left[ ,  (i+1):p \right]\right)^\top $\;
 			$ \mathbf{R}_1\left[ i,  (i+1):p \right] = \mathbf{R}_1\left[ i, (i+1):p \right] - \mathbf{l}^\top $\;

			$\mathbf{U}\left[,  (i+1):p \right] = \mathbf{U}\left[,  (i+1):p \right] - \mathbf{v}\left[2:(m+1)\right] \mathbf{l}^\top$
}
	\tcp{update $\mathbf{R}_1\left[p, p \right]$}
		$\mathbf{R}_1\left[p, p \right] = \sqrt{(\mathbf{R}_1\left[p, p \right])^2 + \Vert \mathbf{U} \left[  , p \right]\Vert_{2}^{2}}$\;
	\textbf{return} $ \mathbf{R}_1 $\;
\end{algorithm}
%
%
\begin{algorithm}[H]
    	\fontsize{9pt}{11pt}
 	\caption{R update when $m\geq 2$ rows are deleted from position $k$ to $k+m-1$, $\mathbf{R}_1^{-} = \mathsf{thinqrdelblockrows}\left(\mathbf{R}_1, \mathbf{U}_k\right)$} 
 	\label{alg:thinQRdeleterows}
 	{\textbf{Input}: $\mathbf{R}_1 \in \mathbb{R}^{p \times p}$, $\mathbf{U}_k \in \mathbb{R}^{m \times p}$ }\;
 	\tcp{delete $m$ rows}
 	\For(){$\left( i = 1\mbox{\; } i \le p\mbox{\; } i ++ \right)$} {
 	 	\tcp{update $\mathbf{R}_1\left[i, i\right]$}
 	 	$s =\Vert \mathbf{U}_k\left[, i\right]\Vert_2^2$\;
 	 	$r = \mathbf{R}_1 \left[i, i\right]$\;
 		$\mathbf{R}_1 \left[i, i\right] = \sqrt{\vert (\mathbf{R}_1\left[i, i\right])^2 - s\vert}$\;  
 		
 		\If() {$\left(i < p\right)$} {
 			$\mathbf{v} = \begin{bmatrix} (\mathbf{R}_1 \left[i, i\right] - r) & \mathbf{U}_k \left[, i\right]^\top \end{bmatrix}^\top$, \qquad $b=(\mathbf{v}\left[1\right])^2$\;
 			$\tau = 2 * b / \left(b + s\right)$\;
			$\mathbf{v} = \begin{bmatrix}1 & \mathbf{v}[2:(m+1)]^{\top} / \mathbf{v}\left[1\right] \end{bmatrix}^{\top}$\;
			
			\tcp{update $\mathbf{R}_1$ and $\mathbf{U}_k$}
 			$\mathbf{w}= \tau *  \mathbf{U}_k\left[, (i+1):p\right]^\top\mathbf{v}\left[2:(m+1)\right]$\;
 			$\mathbf{R}_1\left[i, (i+1):p\right] = \left(\mathbf{R}_1\left[i, (i+1):p\right] + \mathbf{w}^\top\right) / \left(1 - \tau \right)$\;
			$\mathbf{U}_k \left[, (i+1):p\right] = \mathbf{U}_k \left[, (i+1):p\right] - \mathbf{v}\left[2:(m+1)\right]  \big( \tau * \mathbf{R}_1\left[i, (i+1):p\right] + \mathbf{w}^\top \big)$\;
		}
	}
	\textbf{return} $ \mathbf{R}_1 $\;
\end{algorithm}
%
%
\begin{algorithm}[H]
    	\fontsize{9pt}{11pt}
\caption{R update when $m\geq 2$ columns are added from position $k=p+1$ to $k+m-1$, $\mathbf{R}_1^{+} = \mathsf{thinqraddblockcols}\left(\mathbf{R}_1, \mathbf{X}, \mathbf{U}\right)$} 
 	\label{alg:thinQRaddcols}
 	{\textbf{Input}: $\mathbf{R}_1 \in \mathbb{R}^{p \times p}$, $\mathbf{X} \in \mathbb{R}^{N \times p}$, $\mathbf{U} \in \mathbb{R}^{N \times m}$}\;
 	\tcp{add $m$ columns}
	\tcp{compute $\mathbf{R}_{12}$}
 	 Solve $\mathbf{R}_1^\top \mathbf{R}_{12} = \mathbf{X}^\top \mathbf{U}$ with respect to $\mathbf{R}_{12}$ with forward substitution algorithm\;
 
 	\tcp{compute $\mathbf{R}_{22}$}
	$\mathbf{R}_{22} = \mathbf{0}_{m\times m}$\;
	$\mathbf{R}_{22}\left[1, 1\right] = \sqrt{\vert \Vert\mathbf{U}[, 1]\Vert_2^2 - \Vert\mathbf{R}_{12}[, 1]\Vert_2^2 \vert}$\;
	$\mathbf{R}_{22} \left[1, 2:m\right] = \left(\mathbf{U} \left[, 1\right]^\top \mathbf{U} \left[,2:m\right] - \mathbf{R}_{12} \left[, 1\right]^\top \mathbf{R}_{12} \left[, 2:m\right]\right)/\mathbf{R}_{22}\left[1, 1\right]$\;

 	\For(){$\left( i = 2\mbox{\; } i \le m\mbox{\; } i ++ \right)$} {
		$\mathbf{R}_{22} \left[i, i\right] = \sqrt{\vert\Vert \mathbf{U} \left[, i\right] \Vert_{2}^{2} -  \Vert \mathbf{R}_{12} \left[, i\right] \Vert_{2}^{2} - \Vert \mathbf{R}_{22}\left[1:(i-1), i\right] \Vert_{2}^{2}\vert}$\;
		\lIf {$\left(i < m\right)$} {
			$\mathbf{R}_{22}\left[i, (i+1):m\right] =  \big(\mathbf{U} \left[, i\right]^\top \mathbf{U} \left[,(i+1):m\right] - \mathbf{R}_{12} \left[, i\right]^\top \mathbf{R}_{12} \left[, (i+1):m\right] 
			- \mathbf{R}_{22} \left[1:(i-1), i\right]^\top \mathbf{R}_{22} \left[1:(i-1), (i+1):m\right] \big)/\mathbf{R}_{22}\left[i, i\right]$}
	}
	$\mathbf{R}_1 = \begin{bmatrix} \mathbf{R}_{1} & \mathbf{R}_{12} \\ \mathbf{0}_{m \times p} & \mathbf{R}_{22} \end{bmatrix} $\;  
	\textbf{return} $ \mathbf{R}_1 $\;
\end{algorithm}
%
%
\begin{algorithm}[H]
	\begingroup
    	\fontsize{9pt}{11pt}
 	 	\caption{R update when $2\leq m< p$ columns are deleted from position $1\leq k\leq p-m+1$ to $k+m-1$, $\mathbf{R}_1^{-} = \mathsf{thinqrdelblockcols}\left(\mathbf{R}_1, k, m\right)$} 
	\label{alg:thinQRdeletecols}
 	{\textbf{Input}: $\mathbf{R}_1 \in \mathbb{R}^{p \times p}$, $k \in \{1, \ldots, p-m+1\}$, $m \in \{2, \ldots, p+1-k\}$}\;
	\tcp{delete $m$ columns}
	\lIf() {$(k = p-m+1)$} {\textbf{return}  $\mathbf{R}_1 \left[1:(p-m),  1:(p-m) \right] $}
	
	\tcp{permute columns}
	$\mathbf{R}_1\left[, k:(p-m) \right] = \mathbf{R}_1 \left[, (k+m):p \right]$\;
	
	\For(){$\left( i = k\mbox{\; } i \le p-m-1\mbox{\; } i ++ \right)$} {
 			$(\tau,\mathbf{v}, \mu ) = \textsf{householder} (\mathbf{R}_1\left[i, i \right], \mathbf{R}_1\left[  (i+1):(i+m), i \right] )$ as in Algorithm \ref{alg:householder_reflection}\;
			$\mathbf{R}_1\left[i, i \right] = \mu$\;
 			$\mathbf{R}_1\left[(i+1):(i+m), i \right] = \mathbf{0}_m $\;
 	$ \mathbf{R}_1\left[ i:(i+m), (i+1):(p-m) \right] = \mathbf{R}_1\left[ i:(i+m), (i+1):(p-m) \right] -  (\tau * \mathbf{v}) \big(\mathbf{v}^\top \mathbf{R}_1\left[ i:(i+m), (i+1):(p-m) \right]\big)$\;
	}
		\tcp{update $\mathbf{R}_1\left[p-m, p-m\right] $}
	$\mathbf{R}_1\left[p-m, p-m\right] = \sqrt{\Vert\mathbf{R}_1\left[  (p-m):p, p \right]\Vert_2^2}$\;
	\textbf{return} $\mathbf{R}_1 \left[1:(p-m),  1:(p-m) \right] $\;
	\endgroup
\end{algorithm}
%
\begin{algorithm}[H]
	\begingroup
    	\fontsize{9pt}{11pt}
 	 \caption{Apply either Givens rotation or Householder reflection to column $i$, $\mathbf{R}_1 = \mathsf{thinqrstep}\left(\mathbf{R}_1, i, a\right)$} 
	\label{alg:thinQRstep}
 	{\textbf{Input}: $\mathbf{R}_1 \in \mathbb{R}^{p \times l}$, $i \in \left\{1, \dots, l-1\right\}$, $a \in \{1, \ldots, p-i\}$}\;
	
	\eIf(){$\left( a > 1 \right)$} {
	\tcp{Householder reflection}
 			$\left( \tau,\mathbf{v}, \mu \right) = \textsf{householder} \left(\mathbf{R}_1 \left[i, i\right], \mathbf{R}_1 \left[(i+1):(i+a), i\right]\right)$ as in Algorithm \ref{alg:householder_reflection}\;
			$\mathbf{R}_1\left[i, i \right] = \mu$\;
 			$\mathbf{R}_1\left[i:(i+a), (i+1):l \right] = \mathbf{R}_1\left[i:(i+a), (i+1):l \right] - \left(\tau * \mathbf{v}\right)\left(\mathbf{v}^\top\mathbf{R}_1\left[i:(i+a), (i+1):l \right]\right)$\;
 	 		$\mathbf{R}_1\left[(i+1):(i+a), i \right] = \mathbf{0}_a$\;
 	} {	
 		$(c,  s ) =  \textsf{givens} \left(\mathbf{R}_1\left[i, i\right], \mathbf{R}_1\left[  i+1, i  \right] \right)$ as in Algorithm \ref{alg:givens_rotation}\;
 		$ \mathbf{G} = \begin{bmatrix} c & s \\ -s & c \end{bmatrix} $\;
 		\tcp{update $\mathbf{R}_1$}
		$ \mathbf{R}_1\left[ i, i \right] = c*\mathbf{R}_1\left[i , i \right] - s*\mathbf{R}_1\left[i+1 , i \right]$; \\
		$ \mathbf{R}_1\left[i+1 , i \right] = 0 $\;
		$ \mathbf{R}_1\left[ i : (i+1), (i+1):l \right] = \mathbf{G}^{\top} \mathbf{R}_1\left[ i : (i+1), (i+1):l \right]$\;
	}
	\textbf{return} $\mathbf{R}_1 $\;
	\endgroup
\end{algorithm}
%
\begin{algorithm}[H]
	\begingroup
    	\fontsize{9pt}{11pt}
 	 \caption{R update when $m$ non-adjacent columns are deleted, $\mathbf{R}_1^{-} = \mathsf{thinqrdelblockcols\_nonadj}\left(\mathbf{R}_1, \mathbf{k}\right)$} 
	\label{alg:thinQRdeletenonadjcols}
	
 	{\textbf{Input}: $\mathbf{R}_1 \in \mathbb{R}^{p \times p}$, $\mathbf{k}[i] \in \{1, \ldots, p\}, i=1, \ldots, m, \, k[i] < k[j] \, \forall \, i<j, i,j =1, \ldots, m$}\;
		\tcp{delete $m$ columns}
	\lIf(){$\left(m = 1\right)$}{ \textbf{return} $ \textsf{thinqrdelcol}\left(\mathbf{R}_1, \mathbf{k}\left[1\right]\right)$ in Algorithm \ref{alg:thinQRdeletecol}}
	\lIf(){$\left(\left(\mathbf{k}\left[m\right] - \mathbf{k}\left[1\right]\right) = (m-1)\right)$}{ \textbf{return} $\textsf{thinqrdelblockcols}\left(\mathbf{R}_1, \mathbf{k}\left[1\right], m\right)$ in Algorithm \ref{alg:thinQRdeletecols}}
	$\mathbf{e} = 1:p$\;
	$\overline{\mathbf{k}} = \mathbf{e} \setminus \mathbf{k}$\;
	$l = \overline{\mathbf{k}}[p-m]$\;
	$q=m-(p-l)$\;
	$\mathbf{k}=\mathbf{k}[1:q]$\;
	$\mathbf{R}_1 = \mathbf{R}_1 \left[1:l, 1:l\right]$\;
	\lIf(){$\left(q = 1\right)$}{ \textbf{return} $ \textsf{thinqrdelcol}\left(\mathbf{R}_1, \mathbf{k}\left[1\right]\right)$ in Algorithm \ref{alg:thinQRdeletecol}}
	\lIf(){$\left(\left(\mathbf{k}\left[q\right] - \mathbf{k}\left[1\right]\right) = (q-1)\right)$}{ \textbf{return} $\textsf{thinqrdelblockcols}\left(\mathbf{R}_1, \mathbf{k}\left[1\right], q\right)$ in Algorithm \ref{alg:thinQRdeletecols}}
		\tcp{delete columns}
	$\mathbf{R}_1 = \mathbf{R}_1 \left[, \overline{\mathbf{k}}\right]$\;
	
	$\overline{\mathbf{k}} = \overline{\mathbf{k}}\left[\mathbf{k}\left[1\right]:(l-q)\right]$\;
	\tcp{compute $\mathbf{a}[1]$}
	$\mathbf{a} = \mathbf{0}_{l-q-\mathbf{k}[1]+1};  \quad
 \mathbf{a}[1] = \overline{\mathbf{k}}[1] - \mathbf{k}[1]$\;
	
	\tcp{update $\mathbf{R}_1$}
	\For(){$\left( i = 1 \mbox{\; } i \le (l-q-\mathbf{k}[1]) \mbox{\; } i ++ \right)$} {		
		$\mathbf{R}_1 = \textsf{thinqrstep}\left(\mathbf{R}_1, i+\mathbf{k}[1]-1, \mathbf{a}[i]\right)$ in Algorithm \ref{alg:thinQRstep}\;
		$\mathbf{a}[i+1] = \mathbf{a}[i] + \left(\overline{\mathbf{k}}[i+1] - \overline{\mathbf{k}}[i]\right) - 1$
	}
	$\mathbf{R}_1[l-q-\mathbf{k}[1]+1, l-q-\mathbf{k}[1]+1] = \sqrt{\Vert \mathbf{R}_{1}[(l-q-\mathbf{k}[1]+1):(l-\mathbf{k}[1]+1), l-q-\mathbf{k}[1]+1] \Vert^{2}_{2}}$\;
	\textbf{return} $\mathbf{R}_1[1:(l-q),]$\;
	\endgroup
\end{algorithm}
%
%
\section{Computational costs of updating algorithms}
\label{appC}
%
In this supplementary section, we detail the computational costs of all algorithms presented in Section \ref{appB}. As described in \citeSupp[\S 1.1.15]{golub_van_loan.2013}, a common way to quantify computational workload is by counting floating-point operations (FLOPS), where each flop represents an addition, subtraction, multiplication, or division of floating-point numbers. Although this count is not a precise measure of computation time, it provides a useful estimate of the algorithm computational demand by summing the required arithmetic operations. According to IEEE standards, the 64-bit double format allocates one bit for the sign of a floating-point number \cite[\S 2.7.2]{golub_van_loan.2013}. Since changing the sign is not considered an arithmetic operation, it is excluded from the count. Under the IEEE 754-2008 standard, a square root operation is considered a single floating-point operation, though it may take longer than other operations; for simplicity, we follow this convention here. To solve a linear system $\mathbf{R}\mathbf{x} = \mathbf{b}$ with respect to $\mathbf{x} \in \mathbb{R}^p$, where $\mathbf{R} \in \mathbb{R}^{p \times p}$ is triangular and $\mathbf{b} \in \mathbb{R}^p$, forward or backward substitution requires exactly $p^2$ FLOPS \citep[see, e.g.][]{golub_van_loan.2013, bjorck.2015}. Lastly, we use the indicator function $\mathbbm{1}_c$, which equals $1$ when condition $c$ is true and $0$ otherwise.
The exact computational costs of all the algorithms, and their leading-order term, are summarized in Tables \ref{tab:QRcostsfull} and \ref{tab:Rcosts} at the end of this Section.
For comparison purposes, before dealing with the costs of the QR and thin QR updating algorithms that represent the major contribution of this work, we address the computational cost of the naive algorithms that directly update the matrix  $\boldsymbol{\Sigma}=\big(\mathbf{X}^{\intercal} \mathbf{X}\big)^{-1}$ after either the addition of one column to the design matrix $\mathbf{X}\in\mathbb{R}^{N\times p}$ or the deletion of one column from the design matrix $\mathbf{X}\in\mathbb{R}^{N\times p}$. 
\begin{proposition}
\label{prop:onecolinv}
The exact number of flops needed to update the matrix $\boldsymbol{\Sigma}=\big(\mathbf{X}^{\intercal} \mathbf{X}\big)^{-1}$ after the addition of a column to the design matrix $\mathbf{X}$ at position $1\leq k\leq p+1$ using Algorithm \ref{alg:onecolinv_update_add} is $4p^2+2N(p+1)+p$ which is of order $\mathcal{O}\left(Np\right)$. The computational cost of updating the matrix $\boldsymbol{\Sigma}$ after the deletion of a column from the design matrix $\mathbf{X}$ at position $1\leq k\leq p$ using Algorithm \ref{alg:onecolinv_update_del} is $2p^2-3p+1$ which is of order $\mathcal{O}\left(p^2\right)$ independently of the number of rows $N$.
\end{proposition}
\begin{proof}
For $\mathbf{X}\in\mathbb{R}^{N\times p}$ the computational cost of updating $\boldsymbol{\Sigma}=\big(\mathbf{X}^{\intercal} \mathbf{X}\big)^{-1}$ by adding a column is the cost of performing the following few operations as detailed in Algorithm \ref{alg:onecolinv_update_add}:
\begin{itemize}
\item[{\it r3:}] one matrix-vector multiplication: $p(2N-1)$;
\item[{\it r4:}] one vector-matrix multiplication:  $p(2p-1)$;
\item[{\it r5:}] two inner products, a sum and a division: $2N+2p$;
\item[{\it r6:}] a scalar product: $p$;
\item[{\it r7:}] an outer product and a matrix sum: $2p^2$.
\end{itemize}
totalling $4p^2+2N(p+1)+p$.
Similarly, the computational cost of updating $\boldsymbol{\Sigma}$ by removing a column from $\mathbf{X}\in\mathbb{R}^{N\times p}$ is related of the computational cost of performing the following few operations as detailed in Algorithm \ref{alg:onecolinv_update_del}:
\begin{itemize}
\item[{\it r15:}] a scalar product: $p-1$;
\item[{\it r16:}] an outer product and a matrix sum: $2(p-1)^2$.
\end{itemize}
totalling $2(p-1)^2+p-1$, which completes the proof.
\end{proof}
%
%
\subsection{QR updating algorithms}
%
\begin{proposition}[Algorithm \ref{alg:qr_add_one_row}, add one row]
\label{prop:qrupdate_add_one_row}
The computational cost of adding a row at position $1\leq k\leq N+1$ using Algorithm \ref{alg:qr_add_one_row} is equal to $12p+6Np+3p^2$, which is of order $\mathcal{O}(Np)$.
\end{proposition}
%
%
\begin{proof}
The computational cost is 
\begin{itemize}
\item[{\it r8:}] $6p$;
\item[{\it r9:}] $3p$;
\item[{\it r11-12:}] $\sum_{i = 1}^{p-1}6(p-i)=6p(p-1)-6\frac{(p-1)p}{2}=3p(p-1)$;
\item[{\it r14:}] $3p(N+1)$;
\item[{\it r15:}] $3p(N+1)$,
\end{itemize}
totalling  $12p+6Np+3p^2$, which completes the proof.
\end{proof}
%
\begin{proposition}[Algorithm \ref{alg:qr_delete_one_row}, remove one row]
\label{prop:qrupdate_del_one_row}
The computational cost of deleting a row at position $1\leq k\leq N$ using Algorithm \ref{alg:qr_delete_one_row} is equal to $3(N-1)(2N-1)+3(p+2)(p-1)+6$, which is of order $\mathcal{O}(N^2)$.
\end{proposition}
%
%
\begin{proof}
The computational cost is 
\begin{itemize}
\item[{\it r5:}] $6(N-1)$;
\item[{\it r7:}] $3(N-1)$;
\item[{\it r8:}] $\sum_{i=2}^{p+1}6(p-i+2)=3(p+2)(p-1)+6$;
\item[{\it r9:}] $6(N-2)(N-1)$,
\end{itemize}
totalling  $3(N-1)(2N-1)+3(p+2)(p-1)+6$, which completes the proof.
\end{proof}
%
\begin{proposition}[Algorithm \ref{alg:qr_add_one_col}, add one column]
\label{prop:qrupdate_add_one_col}
The computational cost of adding a column at position $1\leq k\leq p+1$ using Algorithm \ref{alg:qr_add_one_col}  is $8N^2 + 8N-(6N+9)(p+1)$ if $k = p+1$, since line {\it r9} is not computed and $3(p-k)^2 + 9(p - 2k) + 8(N^2+N) - 6Nk + 6$ if $1\leq k\leq p$. Therefore, the computational cost  is of order $\mathcal{O}(N^2)$. The computational cost is decreasing as  $k$ increases for $1\leq k\leq p+1$.
\end{proposition}
%
\begin{proof}
If $k = p+1$ the computational cost is $8N^2 + 8N-(6N+9)(p+1)$ which is $\mathcal{O}(N^2)$, since line {\it r9} below is not computed. Otherwise, if $1 \le k \le p$ the computational cost is 
\begin{itemize}
\item[{\it r2:}] $2N^2-N$;
\item[{\it r9:}] $6(N-k)$;
\item[{\it r11:}] $3(N-k)$;
\item[{\it r13:}] $\sum_{i=k+1}^{p+1}6(p-i+2)=6(p+2)(p-k+1)+\frac{6(k+1)k}{2}-\frac{6(p+2)(p+1)}{2}=6(p+2)\left(\frac{p+1}{2}-k\right)+\frac{6(k+1)k}{2}$;
\item[{\it r14:}] $6N(N-k)$,
\end{itemize}
totalling  $8N^2 + 2N(4-3k) - 9k+6(p+2)\left(\frac{p+1}{2}-k\right)+\frac{6(k+1)k}{2}=3(p-k)^2 + 9(p - 2k) + 8(N^2+N) - 6Nk + 6$, which is of order $\mathcal{O}(N^2)$.
\end{proof}
%
\begin{proposition}[Algorithm \ref{alg:qr_delete_one_col}, delete one column]
\label{prop:qrupdate_del_one_col}
The computational cost of deleting a column at position $1\leq k\leq p$ using Algorithm \ref{alg:qr_delete_one_col}  is null if $k=p$ and  it is equal to $3(p-k)^{2}+6 (N+1) (p-k)$ if $1 \le k < p$, which is $\mathcal{O}(N(p-k))$.
\end{proposition}
%
%
\begin{proof}
If $k=p$, the computational cost is $0$, since the algorithm just deletes the last column. Otherwise, if $1 \le k < p$, the computational cost is 
\begin{itemize}
\item[{\it r8:}] $6(p-k)$;
\item[{\it r10:}] $3(p-k)$;
\item[{\it r13:}] $\sum_{i=k}^{p-2}6(p-i-1)=6(p-1)(p-k-1)-\frac{6(p-2)(p-1)}{2}+\frac{6k(k-1)}{2}$;
\item[{\it r15:}] $6N(p-k)$,
\end{itemize}
totalling $3(p-k)^{2}+6 (N+1) (p-k)$, which completes the proof. 
\end{proof}
%
\begin{proposition}[Algorithm \ref{alg:qrupdate_m_rows_added}, add a block of $m\geq 2$ rows]
\label{prop:qrupdate_add_block_rows}
The computational cost of adding a block of $m\geq 2$ rows from position $1\leq k\leq N+1$ to $k+m-1$ using Algorithm \ref{alg:qrupdate_m_rows_added} is equal to $p^{2}(2m+1)+2p\big( 2m(m+1)+N(2m+1)+4\big)-2m-7$, which is of order $\mathcal{O}(mNp)$.
\end{proposition}
%
%
\begin{proof}
The computational cost is 
\begin{itemize}
\item[{\it r9:}] $(p-1)(3m+9)$ for performing the Householder reflections;
\item[{\it r10:}] $(p-1)m$;
\item[{\it r12:}] $\sum_{i=1}^{p-1}(2m+1)(p-i)=(2m+1)\frac{(p-1)p}{2}$;
\item[{\it r13:}] $\sum_{i=1}^{p-1}p-i=\frac{(p-1)p}{2}$;
\item[{\it r14:}] $\sum_{i=1}^{p-1}2(p-i)m=m(p-1)p$;
\item[{\it r16:}] $2m+2$;
\item[{\it r18:}] $p(2m+1)(N+m)$;
\item[{\it r19:}] $(N+m)p$;
\item[{\it r20:}] $2(N+m)mp$,
\end{itemize}
totalling $p^{2}(2m+1)+2p\big( 2m(m+1)+N(2m+1)+4\big)-2m-7$, which completes the proof.
\end{proof}
%
\begin{proposition}[Algorithm \ref{alg:qrupdate_m_rows_deleted}, remove a block of $2\leq m<N$ rows]
\label{prop:qrupdate_del_block_rows}
The computational cost of deleting a block of $2\leq m<N$ rows from position $1\leq k\leq N-m+1$ to $k+m-1$ using Algorithm \ref{alg:qrupdate_m_rows_deleted} is equal to $3Nm (2N -2m+1)+m\big(2m^{2}+\frac{3}{2}m -\frac{7}{2}\big) + 3 m p (p+1)$ which is of order $\mathcal{O}(mN^2)$.
\end{proposition}
%
%
\begin{proof}
The computational cost is 
\begin{itemize}
\item[{\it r6:}] $6\sum_{k=1}^m(N-k)=6mN-3m(m+1)$;
\item[{\it r8:}] $3mN-3\frac{m(m+1)}{2}$;
\item[{\it r10:}] $m(m-1)(3N-m-1)$;
\item[{\it r13:}] $\sum_{j=1}^{m} \sum_{i=j}^{p+j-1} 6 (p+j-i) = 3 m p (p+1)$;
\item[{\it r15:}] $6(N-m)\Big(mN-\frac{(m+1)m}{2}\Big)$,
\end{itemize}
totalling  $3Nm (2N -2m+1)+m\big(2m^{2}+\frac{3}{2}m -\frac{7}{2}\big) + 3 m p (p+1)$, which completes the proof.
\end{proof}
%
\begin{proposition}[Algorithm \ref{alg:qraddblockcolumns}, add a block of $m\geq 2$ columns]
\label{prop:qrupdate_add_block_cols}
The computational cost of adding a block of $m\geq 2$ columns from position $1\leq k\leq p$ to $k+m-1$ using Algorithm \ref{alg:qraddblockcolumns} is equal to $2N(4mN+4m-3mk)+3m\big[p(p+3)+k(k-m-2p-5)+\frac{17}{6}\big]-m^2\big(m+\frac{3}{2}\big)$, which is of order $\mathcal{O}(mN^2)$ and is equal to $8mN^2+2mN-6mp(N+1)-m\big(m^2+\frac{9}{2}m+\frac{7}{2}\big)-3m^2p$ if $k=p+1$, which is of order $\mathcal{O}(mN^2)$.
\end{proposition}
%
%
\begin{proof}
The computational cost when $1\leq k\leq p$ is
\begin{itemize}
\item[{\it r2:}] $(2N-1)mN$;
\item[{\it r5:}] $6m(N-k+1)-6\frac{(m+1)m}{2}=6m(N-k+1)-3m(m+1)$;
\item[{\it r7:}] $3m(N-k+1)-3\frac{(m+1)m}{2}$;
\item[{\it r9:}] $\sum_{j=1}^{m-1} 6(m-j)(N-k-j+1)=m (m-1)\big[ 3 (N-k) -m+2\big]$;
\item[{\it r11:}] $\sum_{j=1}^m \sum_{i=k+j}^{p+j} 6 (p-i+j+1) =  3 m \big[ p(p+3) + k ( k-2p-3) + 2 \big]$;
\item[{\it r13:}] $6N\sum_{j=1}^{m}(N-k-j+1)=6Nm(N-k+1)-3Nm(m+1)$,
\end{itemize}
totalling  $2N(4mN+4m-3mk)+3m\big[p(p+3)+k(k-m-2p-5)+\frac{17}{6}\big]-m^2\big(m+\frac{3}{2}\big)$, while the computational cost for the case where $k=p+1$ is equal to the previous computational costs for rows ${\it r2}$-${\it r9}$ and ${\it r13}$, while ${\it r11}$ is not run,
totalling $8mN^2+2mN-6mp(N+1)-m\big(m^2+\frac{9}{2}m+\frac{7}{2}\big)-3m^2p$ which completes the proof.
\end{proof}
%
\begin{proposition}[Algorithm \ref{alg:qrdeleteblockcolumns}, remove a block of $2\leq m<p$ columns]
\label{prop:qrupdate_del_block_cols}
The computational cost of removing a block of $2\leq m<p$ columns from position $1\leq k\leq p-m+1$ to $k+m-1$ using Algorithm \ref{alg:qrdeleteblockcolumns} is null if $k = p-m+1$, and equal to $p^{2}\big(2m+\frac{3}{2} \big)+p \big[m(4N+3-4m-4k)+3(N-k)+\frac{23}{2} \big]+N \big[m(1-4m-4k)-3k+3\big]+m^{2} \big(2m+4k -\frac{9}{2} \big)-m \big(2k^{2}-3k-\frac{15}{2} \big)+\frac{3}{2}k^{2}-\frac{23}{2}k+10$ if $1 \leq k < p-m+1$, which is of order $\mathcal{O}(mNp)$.
\end{proposition}
%
%
\begin{proof}
The computational cost is 
\begin{itemize}
\item[{\it r8:}] $(p-m-k+1)(3m+9)$ for performing the Householder reflections;
\item[{\it r11:}] $(p-m-k+1) (m+1)$;
\item[{\it r12:}] $\sum_{i=k}^{p-m-1} \big[ (p-m-i)(4m+3) \big]=(p-m)(4m+3)\big(\frac{p-m+1}{2}-k\big)+(4m+3)k\big(\frac{k-1}{2}\big)$;
\item[{\it r13:}] $(p-m-k+1)\big(4Nm+3N\big)$,
\end{itemize}
totalling $p^{2}\big(2m+\frac{3}{2} \big)+p \big[m(4N+3-4m-4k)+3(N-k)+\frac{23}{2} \big]+N \big[m(1-4m-4k)-3k+3\big]+m^{2} \big(2m+4k -\frac{9}{2} \big)-m \big(2k^{2}-3k-\frac{15}{2} \big)+\frac{3}{2}k^{2}-\frac{23}{2}k+10$, which completes the proof.
\end{proof}
%
\begin{proposition}[Algorithm \ref{alg:QRstep}, choice between Givens rotation or Householder reflection]
\label{prop:qr_QRstep}
The computational cost of the updates related to column $i$ of matrices $\mathbf{R} \in \mathbb{R}^{N \times l}$  and $\mathbf{Q} \in \mathbb{R}^{N \times N}$ when a total of $a$ columns before column $i$ have been delated is $9+6N+ 6 (l-i)$ if $a=1$ and $4a+10+N(4a+3)+\mathbbm{1}_{i<l}\big((l-i)(4a+3) \big)$ if $a>1$.
\end{proposition}
%
%
\begin{proof}
The computational cost is 
\begin{itemize}
\item[{\it r3:}] $3a+9$;
\item[{\it r4:}] $a+1$;
\item[{\it r7:}] $(4a+3)(l-i)$;
\item[{\it r8:}] $N(4a+3)$;
\item[{\it r10:}] $6$;
\item[{\it r12:}] $3$;
\item[{\it r14:}] $6 (l-i)$;
\item[{\it r15:}] $6N$,
\end{itemize}
where rows 3 to 8 are performed when $a>1$, the others are performed otherwise. Hence the total is $4a+10+N(4a+3)+\mathbbm{1}_{i<l}\big((l-i)(4a+3) \big)$ if $a>1$ and $9+6N+6 (l-i)$ if $a=1$, which completes the proof.
\end{proof}
%
\begin{proposition}[Algorithm \ref{alg:QRdeletenonadjcols}, delete a block of $m\geq 1$ non-adjacent columns]
\label{prop:qr_QRdeletenonadjcols}
The computational cost of deleting a block of $m\geq 1$ non-adjacent columns at positions $\mathbf{k}=(k_1,\dots,k_m)$ using Algorithm \ref{alg:QRdeletenonadjcols} is equal to 
\begin{itemize}
\item[{\it (i)}] if $m=1$ and $k_1=p$, zero;
\item[{\it (ii)}] if $m=1$ and $1\leq k_1<p$, $3(p-k_1)^2+6(N+1)(p-k_1)$;
\item[{\it (iii)}] if $k_m-k_1=m-1$ and $k_{1}=p-m+1$, zero;
\item[{\it (iv)}] if $k_m-k_1=m-1$ and $k_{1} < p-m+1$, $p^{2}\big(2m+\frac{3}{2} \big)+p \big[m(4N+3-4m-4k_1)+3(N-k_1)+\frac{23}{2} \big]+N \big[m(1-4m-4k_1)-3k_1+3\big]+m^{2} \big(2m+4k_1 -\frac{9}{2} \big)-m \big(2k_1^{2}-3k_1-\frac{15}{2} \big)+\frac{3}{2}k_1^{2}-\frac{23}{2}k_1+10$;
\item[{\it (v)}]  if $q=1$, $3(l-k_1)^2+6(N+1)(l-k_1)+2$;
\item[{\it (vi)}]  if $k_q-k_1=q-1$, $q>1$, $l^{2}\big(2q+\frac{3}{2} \big)+l \big[q(4N+3-4q-4k_1)+3(N-k_1)+\frac{23}{2} \big]+N \big[q(1-4q-4k_1)-3k_1+3\big]+q^{2} \big(2q+4k_1 -\frac{9}{2} \big)-q \big(2k_1^{2}-3k_1-\frac{15}{2} \big)+\frac{3}{2}k_1^{2}-\frac{23}{2}k_1+12$;
\item[{\it (vii)}] if $k_q-k_1\neq q-1$, $q>1$, and $1 < k_2-k_1 \leq l-q-k_1$
\begin{align}
&\frac{3}{2}(l-q)^{2}+3(k_2-k_1)(l-q)-3k_1(k_1-k_2+3)+3(l-q-k_1+1)\nonumber\\
&\qquad-\frac{1}{2}k_2(3 k_2-1)+N(3(l+k_2-2k_1)+q)\nonumber \\
& \qquad -\frac{3}{2}q+\frac{11}{2}l+11+ 4(N+1+l-q) \sum_{i=k_2-k_1}^{l-q-k_1} a_i -4  \sum_{i=k_2-k_1}^{l-q-k_1} i \, a_i;\nonumber
\end{align}
\item[{\it (viii)}] if $k_q-k_1\neq q-1$, $q>1$, and $k_2-k_1 = l-q-k_1+1$, 
\begin{align*}
&6(N+l-q)(k_2-k_1)+12(k_2-k_1)-3(k_2-k_1)^{2}-3(N+2l)\nonumber \\
&\qquad +10q+4Nq+1+3(l-q-k_1+1);
\end{align*}
\item[{\it (ix)}] if $k_q-k_1\neq q-1$, $q>1$, and $k_2-k_1=1$ 
\begin{align}
&\frac{3}{2}(l-q)^{2}+N(q+3(l-k_1+1))-\frac{3}{2}\big(k_1\big)^{2}-\frac{17}{2}(k_1-l) -\frac{9}{2}q +10\nonumber \\
& \qquad +3(l-q-k_1+1)+ 4\left(N+1+l-q\right) \sum_{i=1}^{l-q-k_1} a_i  - 4 \sum_{i=1}^{l-q-k_1} i  \, a_i,\nonumber
\end{align}
\end{itemize}
which is of order $\mathcal{O}((p-m)^2)$.
\end{proposition}%
%
%
\begin{proof}
The computational cost is 
\begin{itemize}
\item[{\it r7:}] $2$;
\item[{\it r14:}] $1$;
\item[{\it r16:}] \begin{align*}
& \sum_{i=1}^{l-q-k_1+1} \left[\mathbbm{1}_{a_i=1} (9+6N+6(l-q-i)) + \right. \nonumber \\
& \quad  \left. +\mathbbm{1}_{a_i>1} \left( 4a_i+10+N(4a_i+3)+ \mathbbm{1}_{i+k_1-1<l-q}(l-q-i)(4a_i+3) \right)\right],
\end{align*}
its computation is detailed below;
\item[{\it r17:}] $3 (l-q-k_1)$.
\end{itemize}
The computational cost of performing row {\it r16} is:
\begin{itemize}
\item[{\it (i)}] if $1<k_2-k_1\leq l-q-k_1$
\begin{align*}
&(9+6N+6(l-q)) (k_2-k_1-1) -3 (k_2-k_1-1)(k_2-k_1)\nonumber \\
& \qquad +\sum_{i=k_2-k_1}^{l-q-k_1} \Big[ (4N+4)a_i+10+3N+(l-q-i)(4a_i+3) \Big]\nonumber\\
&\qquad+(4N+4)q + 10 + 3N\nonumber \\
&= (k_2-k_1-1) \left\{9+6N+6(l-q)-\frac{3}{2}(k_2-k_1) \right\}\nonumber\\
&\qquad+(10+3(N-l-q)) (l-q-k_2+1) \nonumber \\
& \qquad -\frac{3}{2}(l-q-k_1)((l-q-k_1+1)\nonumber\\
&\qquad+4(N+1+l-q) \sum_{i=k_2-k_1}^{l-q-k_1} a_i -4  \sum_{i=k_2-k_1}^{l-q-k_1} i \, a_i\nonumber \\
&= \frac{3}{2}(l-q)^{2}+3(k_2-k_1)(l-q)-3k_1(k_1-k_2+3)\nonumber\\
&\qquad-\frac{1}{2}k_2(3 k_2-1)+N(3(l+k_2-2k_1)+q)\nonumber \\
& \qquad -\frac{3}{2}q+\frac{11}{2}l+11+ 4(N+1+l-q) \sum_{i=k_2-k_1}^{l-q-k_1} a_i -4  \sum_{i=k_2-k_1}^{l-q-k_1} i \, a_i;
\end{align*}
\item[{\it (ii)}] if $k_2-k_1 =l-q-k_1+1$ 
\begin{align*}
& \sum_{i=1}^{k_2-k_1-1} \Big[9+6N+6(l-q-i)\Big] +(4N+4)q +10+3N\nonumber \\
&=6(N+l-q)(k_2-k_1)+12(k_2-k_1)\nonumber\\
&\qquad-3(k_2-k_1)^{2}-3(N+2l)+10q+4Nq+1;
\end{align*}
\item[{\it (iii)}] if $k_2-k_1=1$
\begin{align*}
& \sum_{i=1}^{l-q-k_1} \Big[ 4a_i+10+N(4a_i+3)+(l-q-i)(4a_i+3)\Big]+ 4 q +10 + N(4q+3)\nonumber \\
&=  \frac{3}{2}(l-q)^{2}+N(q+3(l-k_1+1))-\frac{3}{2}\big(k_1\big)^{2}-\frac{17}{2}(k_1-l) -\frac{9}{2}q +10\nonumber \\
& \qquad + 4\left(N+1+l-q\right) \sum_{i=1}^{l-q-k_1} a_i  - 4 \sum_{i=1}^{l-q-k_1} i  \, a_i,
\end{align*}
\end{itemize}
which completes the proof.
\end{proof}
%
\subsection{R updating algorithms}
%
%
\begin{proposition}[Algorithm \ref{alg:thinQRaddrow}, add one row]
\label{prop:thinqr_update_add_one_row}
The computational cost of adding a row at position $k=N+1$ using Algorithm \ref{alg:thinQRaddrow} is equal to $3p(p+2)$ which is of order $\mathcal{O}(p^2)$.
\end{proposition}
%
%
\begin{proof}
The computational cost is 
\begin{itemize}
\item[{\it r3:}] $6p$;
\item[{\it r4:}] $3p$;
\item[{\it r7:}] $\sum_{i=1}^{p-1} 3(p-i)= \frac{3}{2} p(p-1)$;
\item[{\it r8:}] $\sum_{i=1}^{p-1} 3(p-i)= \frac{3}{2} p(p-1)$,
\end{itemize}
totalling $3p(p+2)$, which completes the proof.
\end{proof}
%
\begin{proposition}[Algorithm \ref{alg:thinQRdeleterow}, delete one row]
\label{prop:thinqr_update_del_one_row}
The computational cost of deleting one row at position $1\leq k\leq N$ using Algorithm \ref{alg:thinQRdeleterow} is equal to $3p^{2}+3p-2$ which is of order $\mathcal{O}(p^2)$.
\end{proposition}
%
%
\begin{proof}
The computational cost is 
\begin{itemize}
\item [{\it r4:}] $4p$;
\item [{\it r6:}] $p-1$;
\item [{\it r7:}] $p-1$;
\item [{\it r8:}] $\sum_{i=1}^{p-1} 3(p-i)= \frac{3}{2} p(p-1)$;
\item [{\it r9:}] $\sum_{i=1}^{p-1} 3(p-i)= \frac{3}{2} p(p-1)$,
\end{itemize}
totalling  $3p^{2}+3p-2$, which completes the proof.
\end{proof}
%
%
\begin{proposition}[Algorithm \ref{alg:thinQRaddcol}, add one column]
\label{prop:thinqr_update_add_one_col}
The computational cost of adding a column at position $k=p+1$ using Algorithm \ref{alg:thinQRaddcol} is equal to $p^2+p(2N+1)+2N$ which is of order $\mathcal{O}(Np)$.
\end{proposition}
%
%
\begin{proof}
The computational cost is 
\begin{itemize}
\item[{\it r2:}] $p^{2}-p+2Np$;
\item[{\it r4:}] $2(N+p)-1$;
\item[{\it r5:}] 1,
\end{itemize}
totalling  $p^2+p(2N+1)+2N$, which completes the proof.
\end{proof}
%
\begin{proposition}[Algorithm \ref{alg:thinQRdeletecol}, delete one column]
\label{prop:thinqr_update_del_one_col}
The computational cost of deleting one column at position $1\leq k\leq p$ using Algorithm \ref{alg:thinQRdeletecol} is either zero if $k=p$, or $3(p-k)^{2}+6(p-k)$ otherwise, which is of order $\mathcal{O}((p-k)^2)$.
\end{proposition}
%
%
\begin{proof}
The computational cost is 
\begin{itemize}
\item[{\it r5:}] $6(p-k)$;
\item[{\it r7:}] $3(p-k)$;
\item[{\it r10:}] $\sum_{i=k}^{p-2}6(p-i-1)=6(p-1)(p-k-1)-\frac{6(p-2)(p-1)}{2}+\frac{6k(k-1)}{2}$,
\end{itemize}
totalling  $3(p-k)^{2}+6(p-k)$, which completes the proof.
\end{proof}
%
%
%
\begin{proposition}[Algorithm \ref{alg:thinQRaddrows}, add a block of $m\geq 2$ rows]
\label{prop:thinqr_update_add_block_rows}
The computational cost of adding a block of rows from position $N+1$ to $N+m$ using Algorithm \ref{alg:thinQRaddrows} is equal to $(2m+1)p^2+p(m+8)-m-7$ which is of order $\mathcal{O}(mp^2)$.
\end{proposition}
%
%
\begin{proof}
The computational cost is 
\begin{itemize}
\item[{\it r3:}] $(p-1)(3m+9)$;
\item[{\it r5:}] $\sum_{i=1}^{p-1} \big[(p-i)(2m-1) +2(p-i)\big]=(m+\frac{1}{2})p(p-1)$;
\item[{\it r6:}] $\sum_{i=1}^{p-1} (p-i)=\frac{p(p-1)}{2}$;
\item[{\it r7:}] $\sum_{i=1}^{p-1} 2m(p-i)=mp(p-1)$;
\item[{\it r9:}] $2m+2$,
\end{itemize}
totalling $(2m+1)p^2+p(m+8)-m-7$, which completes the proof.
\end{proof}
%
\begin{proposition}[Algorithm \ref{alg:thinQRdeleterows}, delete a block of $m\geq 2$ rows]
\label{prop:thinqr_update_del_block_rows}
The computational cost of deleting a block of $m\geq 2$ rows from position $1\leq k\leq N-m+1$ to $k+m-1$ using Algorithm \ref{alg:thinQRdeleterows} is equal to $2p^2(m+1)+p(6+m)-m-6$ which is of order $\mathcal{O}(mp^2)$.
\end{proposition}
%
%
\begin{proof}
The computational cost is 
\begin{itemize}
\item[{\it r3:}] $p(2m-1)$;
\item[{\it r5:}] $3p$;
\item[{\it r7:}] $2(p-1)$;
\item[{\it r8:}] $3(p-1)$;
\item[{\it r9:}] $(p-1)m$;
\item[{\it r10:}] $\sum_{i=1}^{p-1}(p-i)+\sum_{i=1}^{p-1}(2m-1)(p-i)=mp(p-1)$;
\item[{\it r11:}] $\sum_{i=1}^{p-1}2(p-i)+p-1=(p+1)(p-1)$;
\item[{\it r12:}] $\sum_{i=1}^{p-1}\big[2(p-i)+2m(p-i)\big]=p(m+1)(p-1)$,
\end{itemize}
totalling  $2p^2(m+1)+p(6+m)-m-6$, which completes the proof.
\end{proof}
%
%
\begin{proposition}[Algorithm \ref{alg:thinQRaddcols}, add a block of $m\geq 2$ columns]
\label{prop:thinqr_update_add_block_cols}
The computational cost of adding a block of $m\geq 2$ columns from position $k=p+1$ to $k+m-1$ using Algorithm \ref{alg:thinQRaddcols} is equal to $2Nmp+m p^{2} + m^{2}(N+p)+m(N+\frac{1}{3}m^2-\frac{1}{3})$ which is of order $\mathcal{O}(mNp)$.
\end{proposition}
%
%
\begin{proof}
The computational cost is 
\begin{itemize}
\item[{\it r2:}] $m(p^{2}-p+2Np)$;
\item[{\it r4:}] $2(N+p)$;
\item[{\it r5:}] $2(m-1)(N+p)$;
\item[{\it r7:}] $\sum_{i=2}^{m} \big[ (2N-1)+(2p-1)+(2(i-1)-1)+3\big]=2(m-1)(N+p-1)+ m(m+1)-2$;
\item[{\it r8:}] $\sum_{i=2}^{m-1}\big[ (2N-1)(m-i) + (2p-1)(m-i)+(2(i-1)-1)(m-i)+3(m-i)\big]=(N+p)(m^2-3m+2)+m(\frac{1}{3}m^{2}-m+\frac{2}{3})$;
\end{itemize}
totalling  $2Nmp+m p^{2} + m^{2}(N+p)+m(N+\frac{1}{3}m^2-\frac{1}{3})$, which completes the proof.
\end{proof}
%
%
\begin{proposition}[Algorithm \ref{alg:thinQRdeletecols}, delete a block of $2\leq m< p$ columns]
\label{prop:thinqr_update_del_block_cols}
The computational cost of deleting a block of $m\geq 2$ columns from position $1\leq k\leq p-m+1$ to $k+m-1$ using Algorithm \ref{alg:thinQRdeletecols} is equal to zero if $k=p-m+1$ and $p^2\big(2m+\frac{3}{2}\big)-p\big(4m^2+k(4m+3)-3m-\frac{23}{2}\big)+m\big(2m^2-\frac{9}{2}m-\frac{19}{2}\big)+mk(4m+2k-3)+\frac{3}{2}k^2-\frac{23}{2}k+2$ otherwise, which is of order $\mathcal{O}(mp^2)$ otherwise.
\end{proposition}
%
%
\begin{proof}
The computational cost is 
\begin{itemize}
\item[{\it r5:}] $(3m+9)(p-k-m)$;
\item[{\it r8:}] $\sum_{i=k}^{p-m-1}\big[m+1+(4m+3)(p-m-i)\big]=(4mp+3p-4m^2-2m+1)(p-m-k)- (4m+3) \frac{(p-m-1)(p-m)}{2}+ (4m+3)\frac{k(k-1)}{2}$;
\item[{\it r10:}] $2m+2$,
\end{itemize}
totalling $p^2\big(2m+\frac{3}{2}\big)-p\big(4m^2+k(4m+3)-3m-\frac{23}{2}\big)+m\big(2m^2-\frac{9}{2}m-\frac{19}{2}\big)+mk(4m+2k-3)+\frac{3}{2}k^2-\frac{23}{2}k+2$, which completes the proof.
\end{proof}
\begin{proposition}[Algorithm \ref{alg:thinQRstep}, choice between Givens rotation or Householder reflection]
\label{prop:thinqr_thinQRstep}
 When a total of $a$ columns previous to column $i$ of matrix $\mathbf{R}_{1} \in \mathbb{R}^{p \times l}$ are delated, the total cost for updating column $i$ is $9+ 6 (l-i)$ if $a=1$ and $(l-i)(4a+3)+4a+10$ if $a>1$.
\end{proposition}
%
%
\begin{proof}
The computational cost is 
\begin{itemize}
\item[{\it r3:}] $3a+9$;
\item[{\it r5:}] $(4a+3)(l-i)+a+1$;
\item[{\it r8:}] $6$;
\item[{\it r10:}] $3$;
\item[{\it r12:}] $6 (l-i)$,
\end{itemize}
where rows 3 and 5 are performed when $a>1$, the others are performed otherwise. Hence the total is $(l-i)(4a+3)+4a+10$ if $a>1$ and $9+6 (l-i)$ if $a=1$, which completes the proof. 
\end{proof}
%

\begin{proposition}[Algorithm \ref{alg:thinQRdeletenonadjcols}, delete  a block of $m\geq 1$ non-adjacent columns]
\label{prop:thinqr_thinQRdeletenonadjcols}
The computational cost of deleting a block of $m\geq 1$ non-adjacent columns at positions $\mathbf{k}=(k_1,\dots,k_m)$ using Algorithm \ref{alg:thinQRdeletenonadjcols} is equal to 
\begin{itemize}
\item[{\it (i)}] if $m=1$ and $k_1=p$, zero;
\item[{\it (ii)}] if $m=1$ and $1\leq k_1< p$, $3(p-k_1)^2+6(p-k_1)$;
\item[{\it (iii)}] if $k_m-k_1=m-1$ and $k_1=p-m+1$, zero;
\item[{\it (iv)}]  if $k_m-k_1=m-1$ and $k_1< p-m+1$, $p^2\big(2m+\frac{3}{2}\big)-p\big(4m^2+k_1(4m+3)-3m-\frac{23}{2}\big)+m\big(2m^2-\frac{9}{2}m-\frac{19}{2}\big)+mk_1(4m+2k_1-3)+\frac{3}{2}k_1^2-\frac{23}{2}k_1+2$;
\item[{\it (v)}] if $q=1$, $3(l-k_1)^2+6(l-k_1)+2$;
\item[{\it (vi)}]  if $k_q-k_1=q-1$, $q>1$, $l^2\big(2q+\frac{3}{2}\big)-l\big(4q^2+k_1(4q+3)-3q-\frac{23}{2}\big)+q\big(2q^2-\frac{9}{2}q-\frac{19}{2}\big)+qk_1(4q+2k_1-3)+\frac{3}{2}k_1^2-\frac{23}{2}k_1+4$;
\item[{\it (vii)}] if $k_q-k_1\neq q-1$, $q>1$, and $1 < k_2-k_1 \leq l-q-k_1$
\begin{align*}&\frac{3}{2}(l-q)^{2}+\frac{11}{2}(l-q)-3(k_1-k_2)(l-q)-3k_1(k_1-k_2+3) \nonumber \\
&\qquad-\frac{1}{2}k_2(3k_2-1)+1+4(l-q+1) \sum_{i=k_2-k_1}^{l-q-k_1} a_i - 4 \sum_{i=k_2-k_1}^{l-q-k_1} i a_i\nonumber\\
&\qquad+3(l-k_1)-q+5;
\end{align*}
\item[{\it (viii)}] if $k_q-k_1\neq q-1$, $q>1$, and $k_2-k_1 = l-q-k_1+1$, 
\begin{align*} -3(k_1-k_2)^{2}+6(k_2-k_1)(l-q+2)-6(l-q)-4+3(l-k_1)-q;
\end{align*}
\item[{\it (ix)}] if $k_q-k_1\neq q-1$, $q>1$, and $k_2-k_1=1$
\begin{align*}
&\frac{3}{2}(l-q)^{2} + \frac{17}{2}(l-q-k_1)-\frac{3}{2}k_1^{2}\nonumber\\
&\qquad\qquad+ 4 (l-q+1) \sum_{i=1}^{l-q-k_1} a_i-4 \sum_{i=1}^{l-q-k_1} ia_i+3(l-k_1)-q+5,
\end{align*}
\end{itemize}
which is of order $\mathcal{O}((p-m)^2)$.
\end{proposition}
%
%
\begin{proof}
The computational cost is 
\begin{itemize}
\item[{\it r7:}] $2$;
\item[{\it r14:}] $1$;
\item[{\it r16:}] $\sum_{i=1}^{l-q-k_1} \left[\mathbbm{1}_{a_i=1} \big[9+6(l-q-i)\big] + \mathbbm{1}_{a_i>1} \left( (l-q-i)(4a_i+3)+4a_i+10 \right) \right]$. The cost is detailed below;
\item[{\it r17:}] $3 (l-q-k_1)$;
\item[{\it r19:}] $2 (q+1)$.
\end{itemize}
The computational cost of row {\it r16} is:
\begin{itemize}
\item[{\it (i)}] if $1 < k_2-k_1 \leq l-q-k_1$ 
\begin{align*}
&(9+6(l-q)) (k_2-k_1-1) - 6\frac{(k_2-k_1-1)(k_2-k_1)}{2}\nonumber\\
&\qquad + \sum_{i=k_2-k_1}^{l-q-k_1}\left[ (l-q-i)(4a_i+3)+4a_i+10 \right]\nonumber \\
&=(9+6(l-q)) (k_2-k_1-1) - 6\frac{(k_2-k_1-1)(k_2-k_1)}{2}   \nonumber \\
&\qquad + 4(l-q+1) \times \sum_{i=k_2-k_1}^{l-q-k_1} a_i +(3(l-q)+10)(l-q-k_2+1)\nonumber \\
& \qquad -3\left(\frac{(l-q-k_1)(l-q-k_1+1)}{2}-\frac{(k_2-k_1)(k_2-k_1-1)}{2}\right)\nonumber\\
&\qquad-4 \sum_{i=k_2-k_1}^{l-q-k_1} i a_i\nonumber\\
&=(k_2-k_1-1) \left( 9+6 (l-q) -\frac{3}{2} (k_2-k_1)\right)\nonumber\\
&\qquad-\frac{3}{2} (l-q-k_1)(l-q-k_1+1)+ (3(l-q)+10)(l-q-k_2+1)\nonumber\\
&\qquad+ 4(l-q+1) \sum_{i=k_2-k_1}^{l-q-k_1} a_i - 4 \sum_{i=k_2-k_1}^{l-q-k_1} i a_i\nonumber \\
&=\frac{3}{2}(l-q)^{2}+\frac{11}{2}(l-q)-3(k_1-k_2)(l-q)-3k_1(k_1-k_2+3) \nonumber \\
&\qquad-\frac{1}{2}k_2(3k_2-1)+1+4(l-q+1) \sum_{i=k_2-k_1}^{l-q-k_1} a_i - 4 \sum_{i=k_2-k_1}^{l-q-k_1} i a_i.
\end{align*}
\item[{\it (ii)}] if $k_2-k_1 = l-q-k_1+1$
\begin{align*}
&(9+6(l-q)) (k_2-k_1-1)-3(k_2-k_1-1)(k_2-k_1)\nonumber \\
&\qquad=-3(k_1-k_2)^{2}+6(k_2-k_1)(l-q+2)-6(l-q)-9.
\end{align*}
\item[{\it (iii)}] if $k_2-k_1=1$
\begin{align*}
&\sum_{i=1}^{l-q-k_1} \left[(l-q-i)(4a_i+3)+4a_i+10 \right]\nonumber \\
&\qquad= (3(l-q)+10)(l-q-k_1)+ 4 (l-q+1) \sum_{i=1}^{l-q-k_1} a_i\nonumber\\
&\qquad\qquad-3\frac{(l-q-k_1)(l-q-k_1+1)}{2}-4 \sum_{i=1}^{l-q-k_1} ia_i\nonumber \\
&\qquad=\frac{3}{2}(l-q)^{2} + \frac{17}{2}(l-q-k_1)-\frac{3}{2}k_1^{2}\nonumber\\
&\qquad\qquad+ 4 (l-q+1) \sum_{i=1}^{l-q-k_1} a_i-4 \sum_{i=1}^{l-q-k_1} ia_i,
\end{align*}
\end{itemize}
which completes the proof.
\end{proof}
%

\begin{table}[h]
\caption{Arithmetic operations (FLOPS) required by QR updating algorithms for performing addition or deletion of rows or columns of an $N \times p$ matrix. ``Exact'' and  ``Most relevant term'', provide the exact computational complexity and the most relevant term of the computational cost, respectively.}
\centering
\begin{tabular}{p{2.2cm} p{0.60cm}p{6.5cm}p{1.5cm}}
\toprule
Descr. & Prop.	&	Exact	&	Most relevant term	\\
\midrule
add $1$ row	&	\ref{prop:qrupdate_add_one_row}	&	$12p + 6Np + 3p^2$ 	&	$6Np$		\\
add the $k$-th col.	&	\ref{prop:qrupdate_add_one_col}	&	$8N^2 + 8N-(6N+9)(p+1)$ if $k = p+1$ and $3(p-k)^2 + 9(p - 2k) + 8(N^2+N) - 6Nk + 6$ if $1\leq k\leq p$	&	$8N^2$		\\
rem. $1$ row	&	\ref{prop:qrupdate_del_one_row}	&	$3(N-1)(2N-1)+3(p+2)(p-1)+6$	&	$6N^2$		\\
rem. the $k$-th col.	&	\ref{prop:qrupdate_del_one_col}	&	null if $k=p$ and  it is equal to $3(p-k)^{2}+6 (N+1) (p-k)$ if $1 \le k < p$	&	$6N(p-k)$		\\															
add $m$ rows	&	\ref{prop:qrupdate_add_block_rows}	&	$p^{2}(2m+1)+2p\big( 2m(m+1)+N(2m+1)+4\big)-2m-7$	&	$4mNp$	\\
add $m$ col.s	&	\ref{prop:qrupdate_add_block_cols}	&	$8mN^{2}+2mN-6mp(N+1)-m(m^{2}+\frac{9}{2}m+\frac{7}{2})-3m^{2}p $ if $k=p+1$ and $2N(4mN+4m-3mk)+3m[p(p+3)+k(k-m-2p-5)+\frac{17}{6}]-m^{2}(m+\frac{3}{2})$ otherwise	&	$8mN^2$	\\
rem. $m$ rows	&	\ref{prop:qrupdate_del_block_rows}	&	$3Nm (2N -2m+1)+m\big(2m^{2}+\frac{3}{2}m -\frac{7}{2}\big) + 3 m p (p+1)$	&	$6mN^2$		\\
rem. $m$ col.s	&	\ref{prop:qrupdate_del_block_cols}	&	null if $k = p-m+1$, or $p^{2}\big(2m+\frac{3}{2} \big)+p \big[m(4N+3-4m-4k)+3(N-k)+\frac{23}{2} \big]+N \big[m(1-4m-4k)-3k+3\big]+m^{2} \big(2m+4k -\frac{9}{2} \big)-m \big(2k^{2}-3k-\frac{15}{2} \big)+\frac{3}{2}k^{2}-\frac{23}{2}k+10$ if $1 \leq k < p-m+1$	&	$4mNp$\\					
\bottomrule
\end{tabular}
\label{tab:QRcostsfull}
\end{table}%
%
\begin{table}[!h]
\caption{Arithmetic operations (FLOPS) required by R updating algorithms for performing addition or deletion of rows or columns of an $N \times p$ matrix. ``Exact'' and  ``Most relevant term'', provide the exact computational complexity and the most relevant term of the computational cost, respectively.}
\centering
\begin{tabular}{p{2.2cm} p{0.60cm}p{6.5cm}p{1.5cm}}
\toprule
Descr. & Prop.	&	Exact	&	Most relevant term	\\
\midrule
add $1$ row	&	\ref{prop:thinqr_update_add_one_row}	&	$3p(p+2)$	&	$3p^2$		\\
add the $(p+1)$-th column	&	\ref{prop:thinqr_update_add_one_col}	&	$p^2+p(2N+1)+2N$	&	$2Np$	\\
rem. $1$ row	&	\ref{prop:thinqr_update_del_one_row}	&	$3p^{2}+3p-2$	&	$3p^2$	\\
rem. the $k$-th column	&	\ref{prop:thinqr_update_del_one_col}	&	null if $k=p$, or $3(p-k)^{2}+6(p-k)$, otherwise	&	$3(p-k)^2$		\\							
add $m$ rows	&	\ref{prop:thinqr_update_add_block_rows}	&	$(2m+1)p^2+p(m+8)-m-7$	&	$2mp^2$		\\
add $m$ columns	&	\ref{prop:thinqr_update_add_block_cols}	&	$2Nmp+m p^{2} + m^{2}(N+p)+m(N+\frac{1}{3}m^2-\frac{1}{3})$	&	$2mNp$		\\
rem. $m$ rows	&	\ref{prop:thinqr_update_del_block_rows}	&	$2p^2(m+1)+p(6+m)-m-6$	&	$2mp^2$		\\
rem. $m$ columns	&	\ref{prop:thinqr_update_del_block_cols}	&	zero if $k=p-m+1$ and $p^2\big(2m+\frac{3}{2}\big)-p\big(4m^2+k(4m+3)-3m-\frac{23}{2}\big)+m\big(2m^2-\frac{9}{2}m-\frac{19}{2}\big)+mk(4m+2k-3)+\frac{3}{2}k^2-\frac{23}{2}k+2$ otherwise	&	$2mp^2$		\\
\bottomrule
\end{tabular}
\label{tab:Rcosts}
\end{table}

\section{Prior hyper-parameters settings}
\label{app:prior_hyperparameters}
%
The spike-and-slab approach in Sections \ref{sec:simulations}-\ref{sec:applications} requires the tuning of several prior hyper-parameters $(\nu,\lambda,\xi,\varphi,\upsilon_0)$ that control for the prior residual variance $(\nu,\lambda)$, the average number of covariates included prior to observing the data $(\xi,\varphi)$ and the variance of the slab component $\upsilon_0$. Those prior hyper-parameters are usually selected in order to be as non-informative as possible. Several alternatives have been discussed in literature \citep[][]{george_foster.2000,chipman_etal.2001,liang_etal.2008,cui_george.2008,narisetty_he.2014,ghosh_clyde.2011,biswas_etal.2022}. Let  $\mathbb{E}(\theta)=\mu_\theta$ and $\mathbb{V}(\theta)=\sigma_\theta^2$ be the prior mean and variance for the parameter $\theta$, respectively, we fix $(\xi,\varphi)$ in such a way that $\mathbb{P}\big(\sum_{j=1}^{p}\gamma_j>K\big)=10^{-1}$, for a pre-specified value of $K = \max\{K_0, \log(n)\}$ and $K_0=40$ and $\sigma_\theta=10^{-1}$ as suggested by \cite{narisetty_he.2014}. Moreover, we set $\nu=0.5$ as in \cite{biswas_etal.2022} and we set 
\begin{equation*}
\lambda=\begin{cases}
5&\text{if }p<10^3\\10 & \text{if }10^3\leq p<10^4\\ 15 & \text{if } p\geq10^4.
\end{cases}
\end{equation*}
As concerns $\upsilon_0$, we again adapt the approach  developed by \cite{narisetty_he.2014} and fix
\begin{equation*}
\upsilon_0=\widehat{\sigma}_y^2\max\Big\{\frac{p^{2.1}}{100n},\log(n)\Big\},
\end{equation*}
where $\widehat{\sigma}_y^2$ is the sample variance of $y$ as default option in our simulation experiments. Alternatively, we propose a data-driven approach to select $\upsilon_0$ that consists in adapting the approach based on log predictive score minimization of \cite{lamnisos_etal.2012}. Further details are provided in Section \ref{sec:model_prior_appl}.\par 
Before concluding, it is also worth detailing the prior distribution and the prior hyper-parameters selected for the alternative models compared in the simulation and real data examples. The Rao-Blackwellized stochastic search variable selection algorithm of \cite{ghosh_clyde.2011} considers a hierarchical prior as in equation \eqref{eq:spike_slab_prior_general_def_1}, with
\begin{equation*}
\boldsymbol{\Sigma}_{\beta_{\gamma}}=\frac{\sigma^2}{\zeta_j}\mathbf{I}_{p_\gamma},\qquad\zeta_j\sim\mathrm{G}(\alpha/2,\alpha/2),\qquad j=1,\dots,p_\gamma,
\end{equation*}
independent Bernoulli priors are then placed on the model indicator $\boldsymbol{\gamma}$, i.e. $p(\boldsymbol{\gamma}\vert\boldsymbol{\theta})\sim\prod_{j=1}^{p_\gamma}\mathrm{Ber}(\theta_j)$ and $p(\sigma^2)\propto\big(\sigma^2\big)^{-1}$ for the scale parameter. The Authors set $\theta_j=0.5$ for $j=1,\dots,p_\gamma$. The same prior distributions and hyper-parameters are selected by \cite{biswas_etal.2022} for their scalable spike-and-slab approach. The prior for the scale parameter $\sigma^2$ is an Inverse Gamma distribution with $(\nu=0.5,\lambda=0.5)$.
\subsection{Calibration of $\upsilon_0$}
\label{sec:model_prior_appl}
%
\noindent 
Selecting the $\upsilon_0$ hyper-parameter is critical for achieving reliable posterior inference across the model space. Whereas prior uncertainty in $\theta$ and $\sigma^2$ can be addressed through noninformative priors, as discussed in the previous section, the choice of $\upsilon_0$ demands particular attention. This parameter plays a central role in governing the trade-off between model fit and regularization, and thus has a pronounced impact on inference quality. Empirical Bayes approaches  \citep{cui_george.2008}, for instance, offer adaptive settings for $\upsilon_0$ that may outperform vague priors often used by default. This approach enables us to concentrate on the relevant aspects of the model while still enforcing regularization, though it departs from the usual Bayesian approach of accounting for prior uncertainty across all parameters.\par
%
%
The value of $\upsilon_0$ can be chosen based on the recommendations of \cite{narisetty_he.2014} which rely on asymptotic considerations (see previous discussion), by empirical Bayes methods \citep[][]{george_foster.2000}, or alternatively, it can be selected by minimizing a predictive criterion, as proposed by \cite{lamnisos_etal.2012}. 
%
%
Here, we briefly discuss the approach proposed by \cite{lamnisos_etal.2012}, which we adopt because the availability of fast posterior model sampling algorithms substantially reduces computational burden and allows for an efficient implementation of the method.\par
\cite{lamnisos_etal.2012} propose a cross-validation (CV) approach to calibrate fixed hyper-parameters that relies on the minimization of the log-predictive score: 
\begin{equation*}
\mathcal{S}(\upsilon_0)=-\frac{1}{n} \sum_{i=1}^n \log \pi(y_i \vert \mathbf{y}_{-\kappa(i)}, \upsilon_0),
\end{equation*}
where $\kappa(i) \in\{1, \ldots, k\}$ represents the partition to which $y_i$ is allocated, $\mathbf{y}_{-\kappa(i)}$ are the observations from the remaining partitions and $\pi(y_i \vert \mathbf{y}_{-\kappa(i)}, \upsilon_0)$ denotes the posterior predictive distribution of $y_i$ for $i=1,\dots,n$. For each CV set $\kappa(i)$, the posterior predictive of $y_i$ can be obtained as:
\begin{align}
\label{eq:log_predictive_score_1}
\pi(y_i \vert \mathbf{y}_{-\kappa(i)},\mathbf{x}_i, \upsilon_0)&=  \sum_{\boldsymbol{\gamma}} \int \pi(y_i \vert\mathbf{x}_{i}, \boldsymbol{\beta}_{\gamma},\sigma^2, \boldsymbol{\gamma}) \pi(\boldsymbol{\beta}_{\gamma},\sigma^2, \boldsymbol{\gamma} \vert \mathbf{y}_{-\kappa(i)},\mathbf{X}_{-\kappa(i)}^\gamma, \upsilon_0) d \boldsymbol{\beta}_{\gamma} d\sigma^2\nonumber\\
%
%
&= \sum_{\boldsymbol{\gamma}} \mathbb{E}_{m(\boldsymbol{\gamma}\vert  \mathbf{y}_{-\kappa(i)},\mathbf{X}^{\gamma})}\Big[\pi(y_i \vert \mathbf{x}_{i}, \boldsymbol{\gamma})\Big],
\end{align}
where $\pi(\boldsymbol{\beta}_{\gamma},\sigma^2, \boldsymbol{\gamma} \vert \mathbf{y}_{-\kappa(i)},\mathbf{X}_{-\kappa(i)}^{\gamma}, \upsilon_0)$ is the joint posterior distribution, $\pi(y_i \vert \mathbf{x}_{i}, \boldsymbol{\gamma})$ is the marginal predictive which is analytically available for the Gaussian linear model, and $m(\boldsymbol{\gamma}\vert  \mathbf{y},\mathbf{X}^{\gamma})$ is the marginal likelihood in equation \eqref{eq:marginal_posterior}.  Any proper score functions \citep[][]{gneiting_raftery.2007} could replace the logarithmic score function in equation \eqref{eq:log_predictive_score_1}. When $\kappa(i)=i$ for all $i=1,\dots,n$, then the $K$-fold cross-validation becomes the leave-one-out cross-validation (LOOCV).\par
Both the asymptotic approach of \cite{narisetty_he.2014} and the method proposed by \cite{lamnisos_etal.2012} offer structured ways to determine $\upsilon_0$ balancing between model fit and regularization. However, while the former requires only the calculation of a specific formula, the latter involves simulating from the model posterior for different data partitions. Specifically, the cross-validation approach necessitates a simulation-based estimator of the log-predictive score in equation \eqref{eq:log_predictive_score_1}, which evaluates the model predictive performance across these partitions:
\begin{equation}
\label{eq:log_predictive_score_2}
\widehat{\pi}(y_i \vert \mathbf{y}_{-\kappa(i)},\mathbf{x}_i, \upsilon_0)
= \frac{1}{B}\sum_{b=1}^B\pi(y_i \vert \mathbf{x}_{i}, \boldsymbol{\gamma}),\quad \boldsymbol{\gamma}^{(b)}\sim m(\boldsymbol{\gamma}\vert  \mathbf{y}_{-\kappa(i)},\mathbf{X}^{\gamma}),\quad b=1,\dots,B.
\end{equation}
\cite{lamnisos_etal.2012} introduce alternative estimators designed to alleviate the computational burden by incorporating importance densities. However, the availability of a fast and efficient method for simulating from the marginal posterior allows us to directly utilize equation \eqref{eq:log_predictive_score_2}. More significantly, for a given fold $\kappa(i)$, the techniques for rapidly updating the R factor when modifying the columns of the design matrix prove especially beneficial. Most importantly, these algorithms can update the R factor even when the fold changes, not only when columns are altered. Additionally, the computational cost of modifying the R factor after removing a row remains constant, regardless of the row position (see Section \ref{sec:computational_costs} for a detailed discussion).
%
\section{Additional simulation results}
\label{appD}
%
\subsection{Independent covariates}
We present additional results for the first simulation scenario, where covariates were simulated under the assumption of independence. While the AUC and F1 score are discussed in Section 5 of the main paper, here we include Figures \ref{fig:TPR_Sim1}, \ref{fig:FDR_Sim1} and \ref{fig:MSE_Sim1}, which illustrate the true positive rate (TPR), false discovery rate (FDR), and mean square error (MSE) of the $\boldsymbol{\beta}$ estimates, respectively. When $n$ is much smaller than $p$, the RJ algorithm using the R update tends to include covariates outside the true model, though some correct variables are selected, resulting in a nonzero true positive rate. In contrast, the SSVS algorithm fails to select any covariates when $p_0 > 10$, leading to a true positive rate of zero and an undefined false discovery rate.
%
\begin{figure}[H]
\begin{center}
\includegraphics[width=0.8\textwidth]{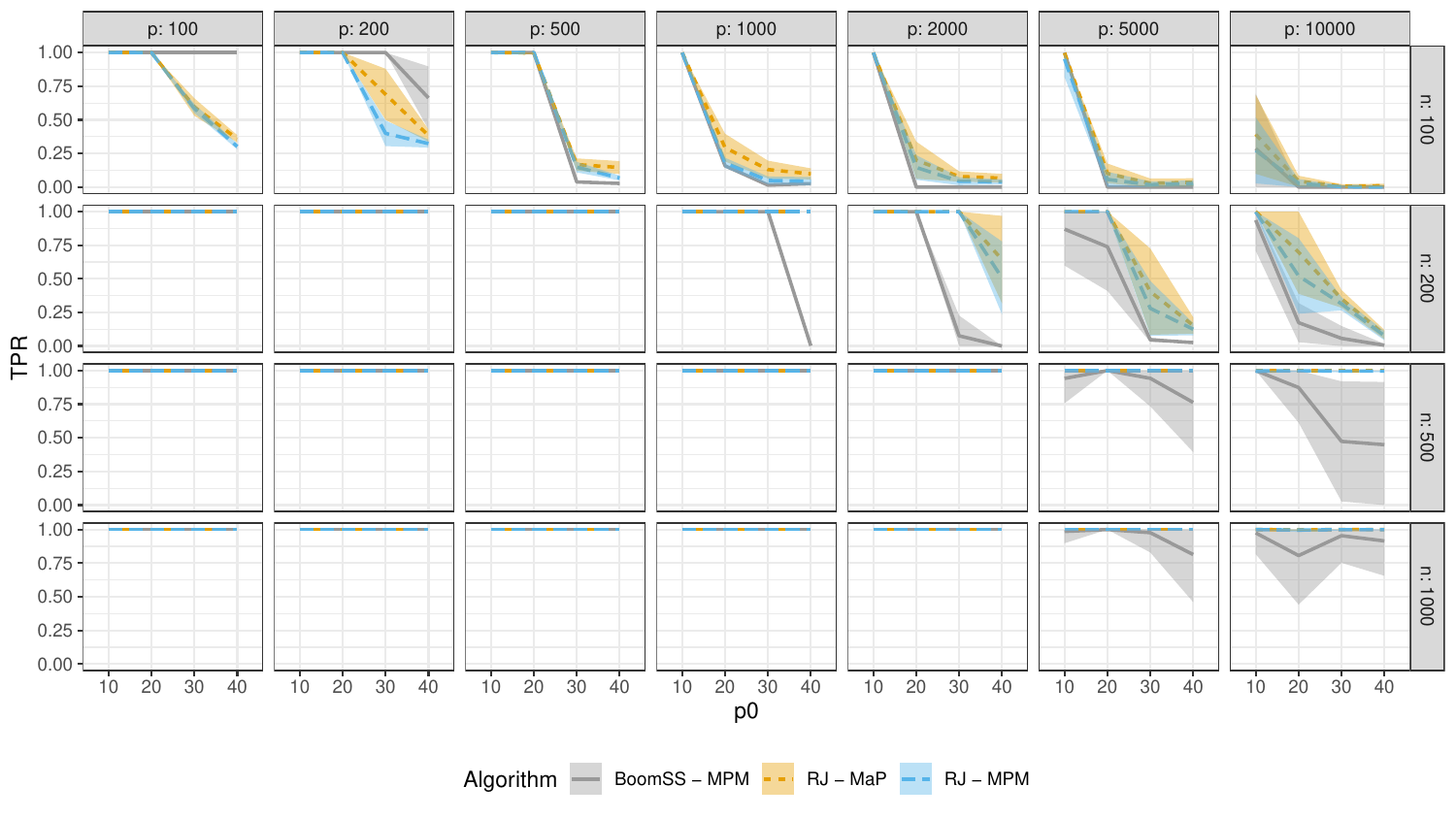}
\caption{Mean true positive rate (TPR) with 1 standard error bands of the MPM and the MaP model computed by the RJ (with R update) and the BoomSS algorithms. 40 repetitions for each setting with independent covariates.}
\label{fig:TPR_Sim1}
\end{center}
\end{figure}
\begin{figure}[H]
\begin{center}
\includegraphics[width=0.8\textwidth]{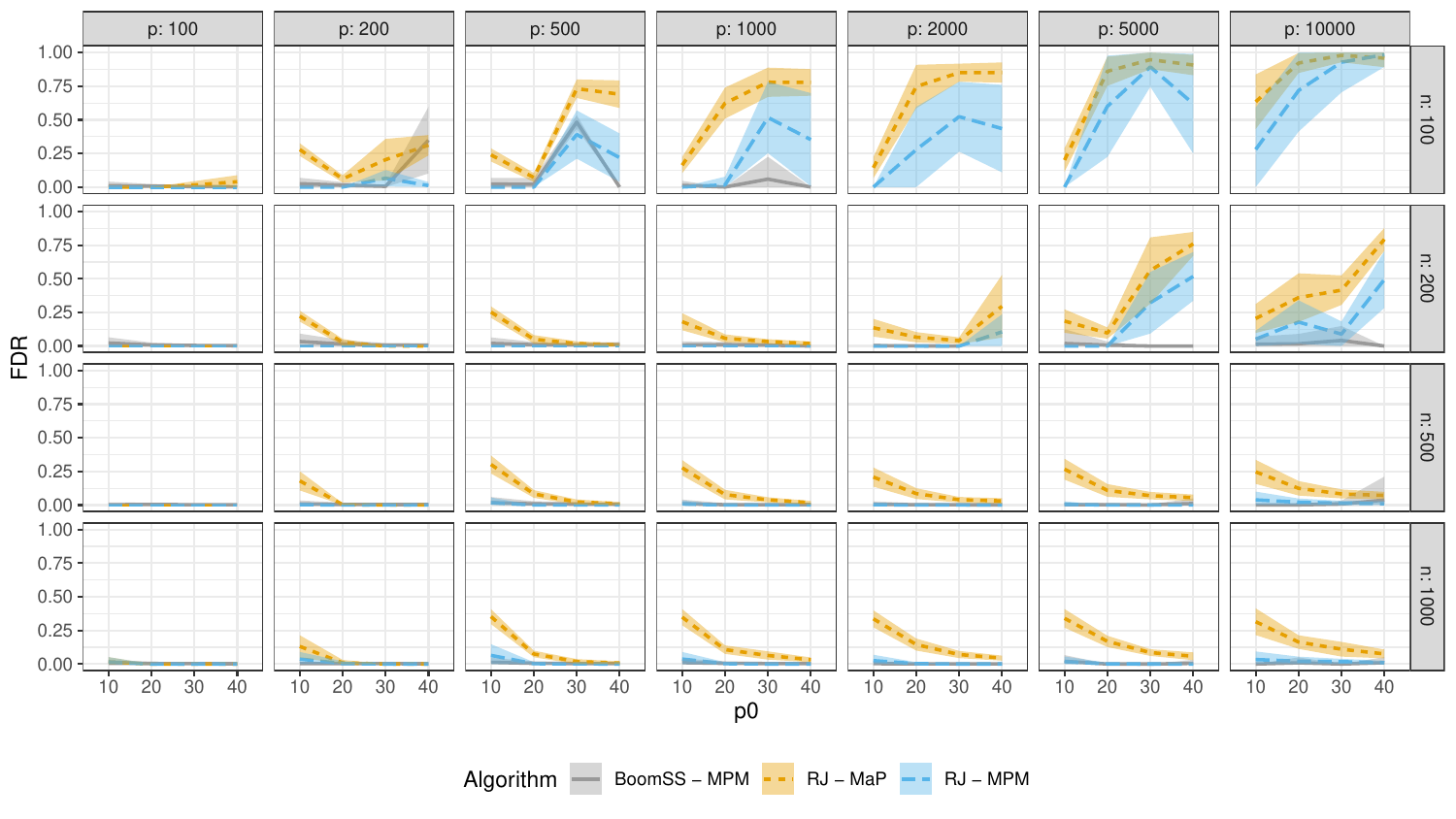}
\caption{Mean false discovery rate (FDR) with 1 standard error bands of the MPM and the MaP model computed by the RJ (with R update) and the BoomSS algorithms. 40 repetitions for each setting with independent covariates.}
\label{fig:FDR_Sim1}
\end{center}
\end{figure}
\begin{figure}[H]
\begin{center}
\includegraphics[width=0.8\textwidth]{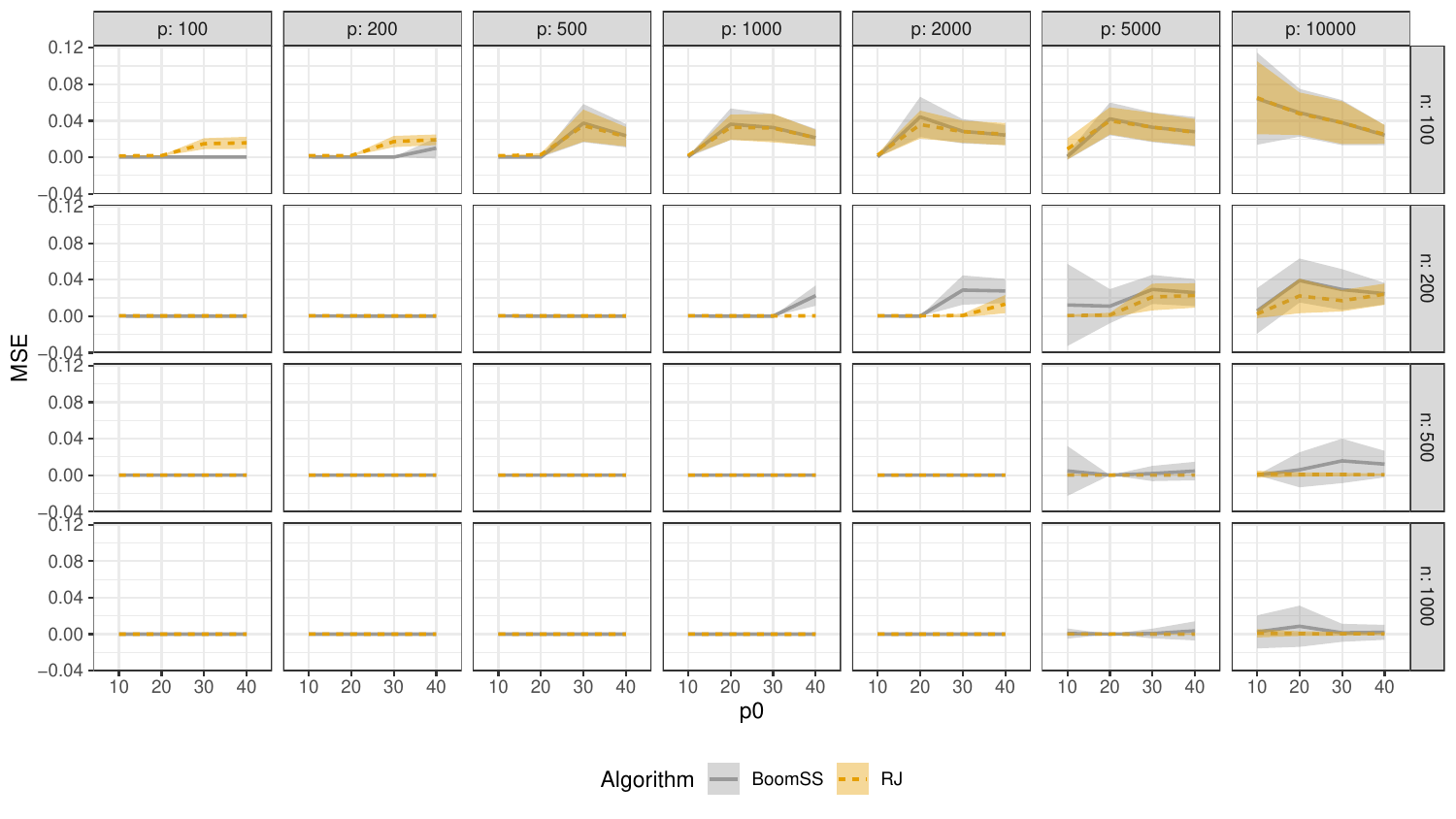}
\caption{Mean square error with 1 standard error bands of the estimates of the MPM computed by the RJ (with R update) and the BoomSS algorithms. 40 repetitions for each setting with independent covariates.}
\label{fig:MSE_Sim1}
\end{center}
\end{figure}
%
\subsection{Equicorrelated covariates}
Figures \ref{fig:auc_Sim2}--\ref{fig:MSE_Sim2} display the outcomes for the scenario where covariates are simulated with equicorrelated dependence. The trends observed here closely resemble those in the independent covariate scenario, particularly in terms of model performance metrics like AUC, TPR, and MSE. However, while the overall behavior remains similar, equicorrelated dependence introduces subtle changes. For instance, the correlation among covariates slightly influences the selection process, which may lead to minor shifts in FDR and TPR, especially in settings where $p$ is large. This similarity in performance despite the dependency structure highlights the robustness of the model, suggesting that the algorithms can manage both independent and equicorrelated structures with comparable effectiveness.
%
\begin{figure}[H]
\begin{center}
\includegraphics[width=0.8\textwidth]{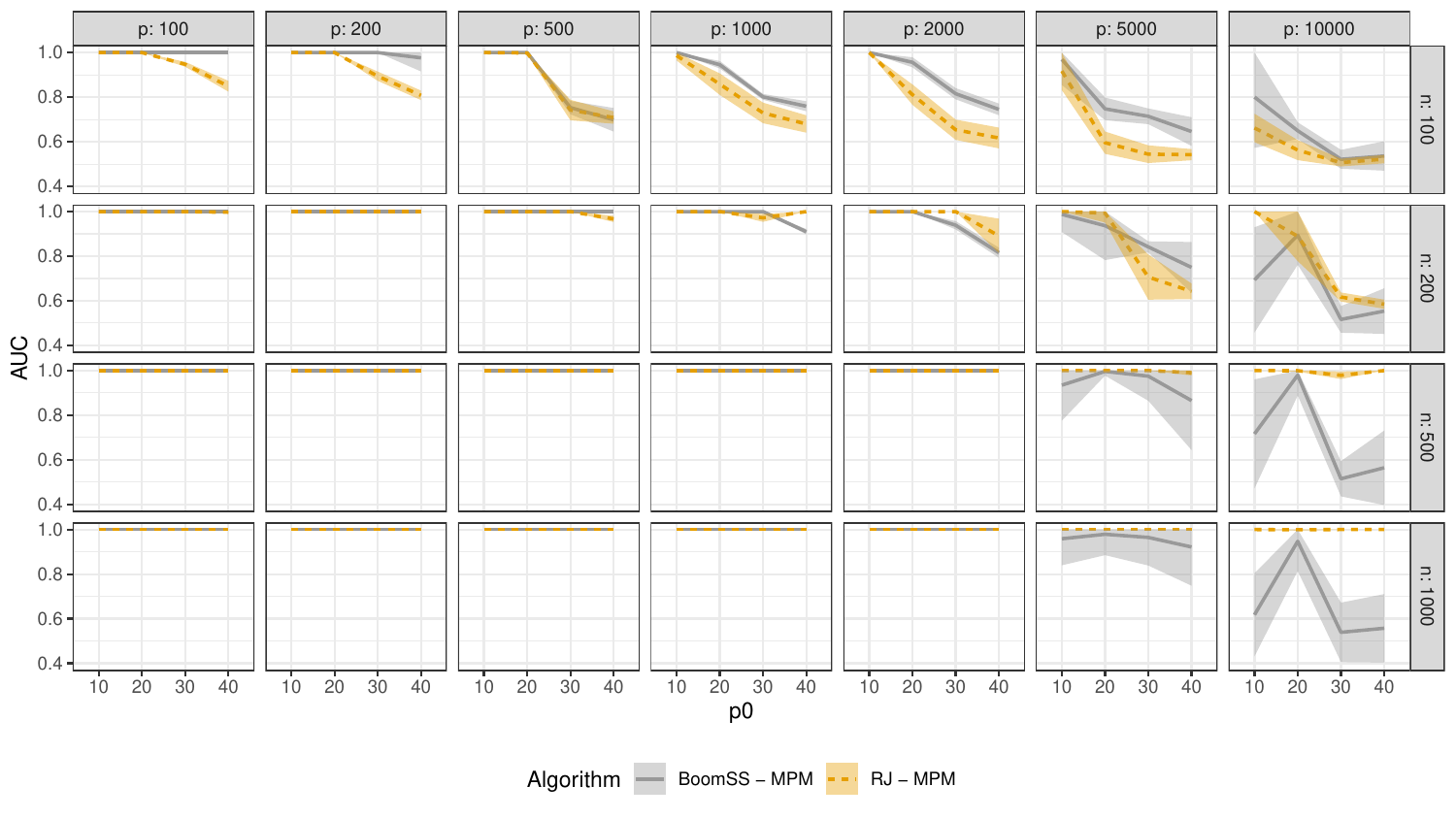}
\caption{Mean AUC with 1 standard error bands of the MPM computed by the RJ (with R update) and the BoomSS algorithms. 40 repetitions for each setting with equicorrelated covariates.}
\label{fig:auc_Sim2}
\end{center}
\end{figure}
\begin{figure}[H]
\begin{center}
\includegraphics[width=0.8\textwidth]{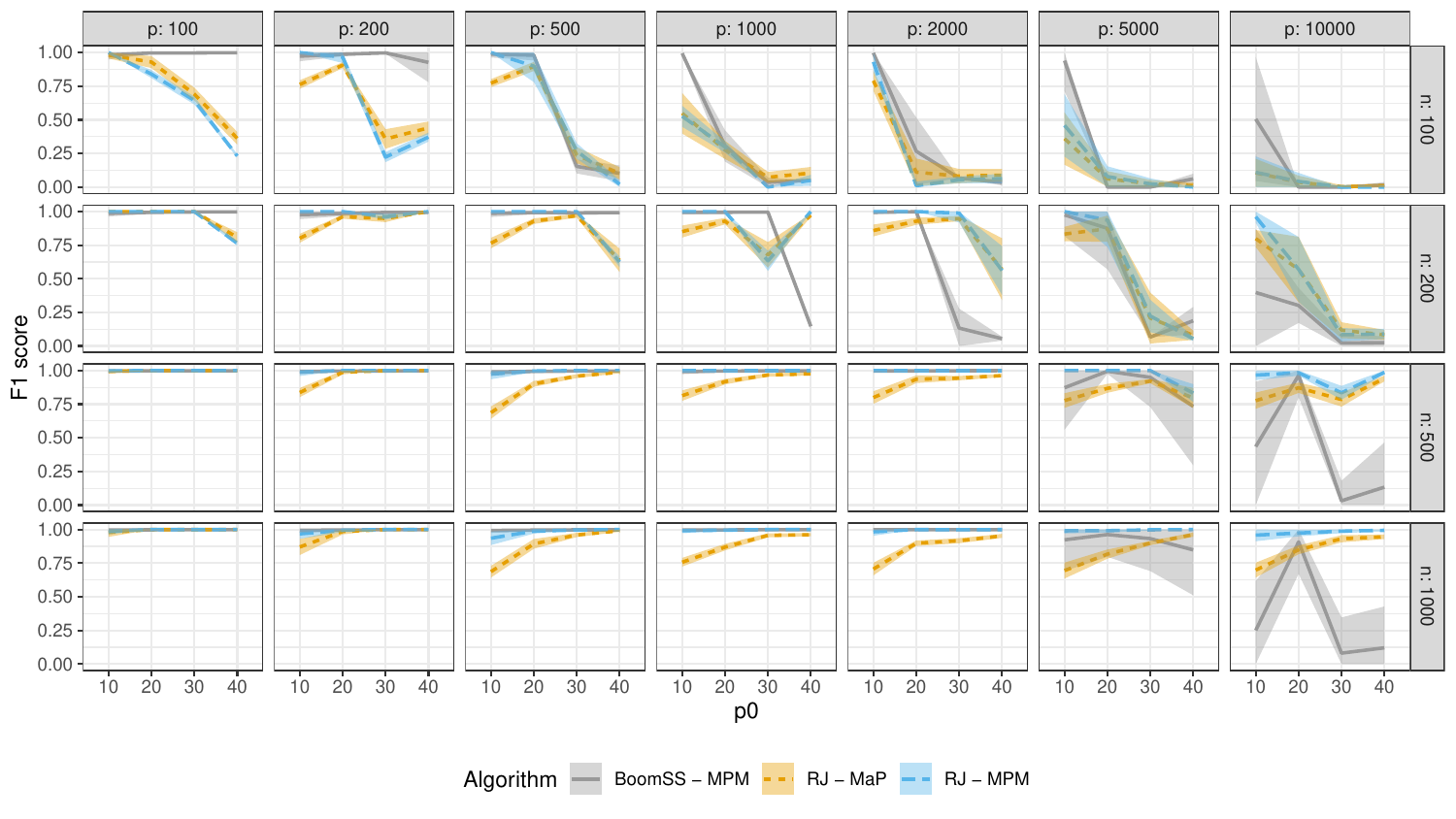}
\caption{Mean F1 score with 1 standard error bands of the MPM and the MaP model computed by the RJ (with R update) and the BoomSS algorithms. 40 repetitions for each setting with equicorrelated covariates.}
\label{fig:F1_Sim2}
\end{center}
\end{figure}
\begin{figure}[H]
\begin{center}
\includegraphics[width=0.8\textwidth]{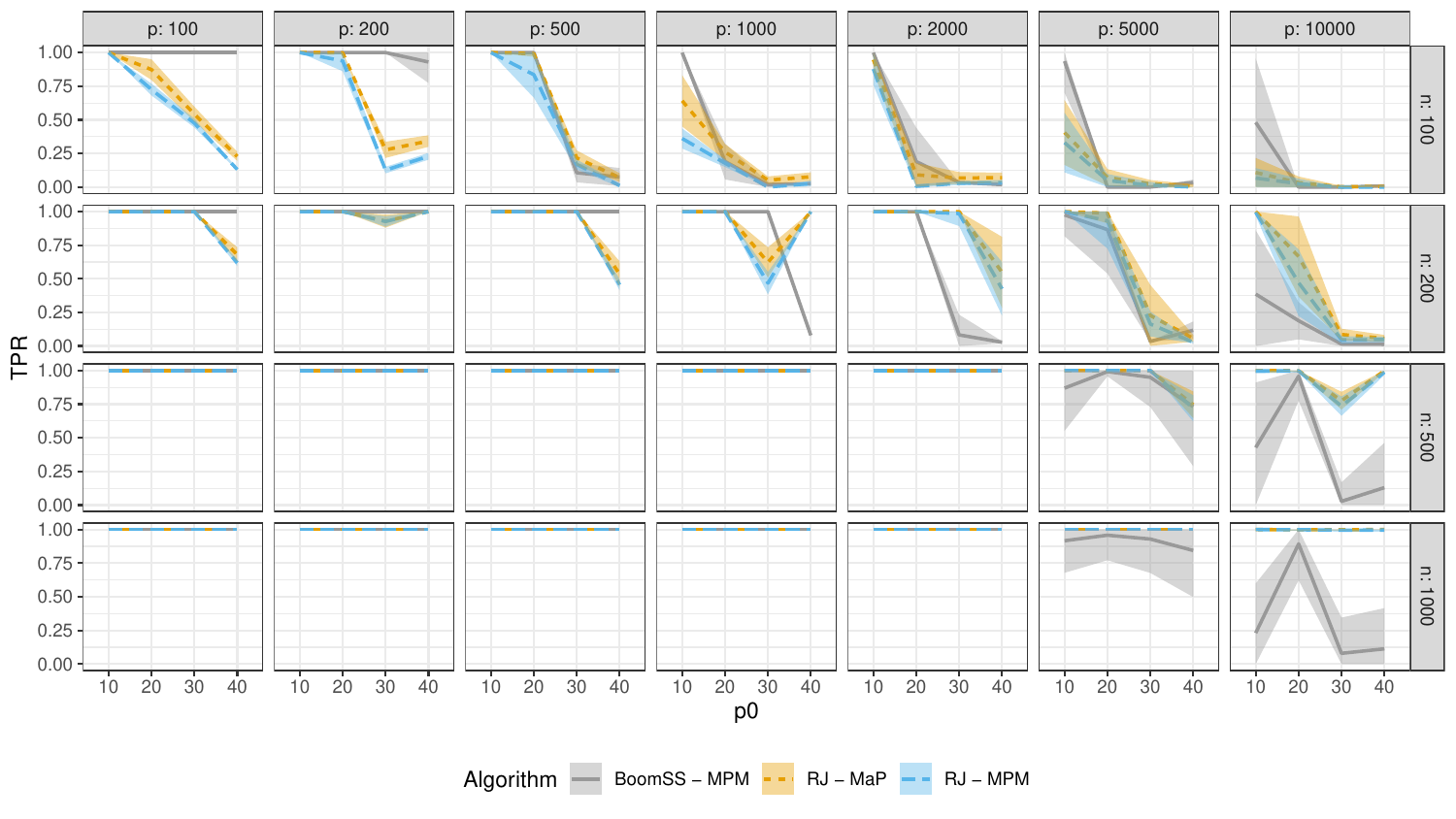}
\caption{Mean TPR with 1 standard error bands of the MPM and the MaP model computed by the RJ (with R update) and the BoomSS algorithms. 40 repetitions for each setting with equicorrelated covariates.}
\label{fig:TPR_Sim2}
\end{center}
\end{figure}
\begin{figure}[H]
\begin{center}
\includegraphics[width=0.8\textwidth]{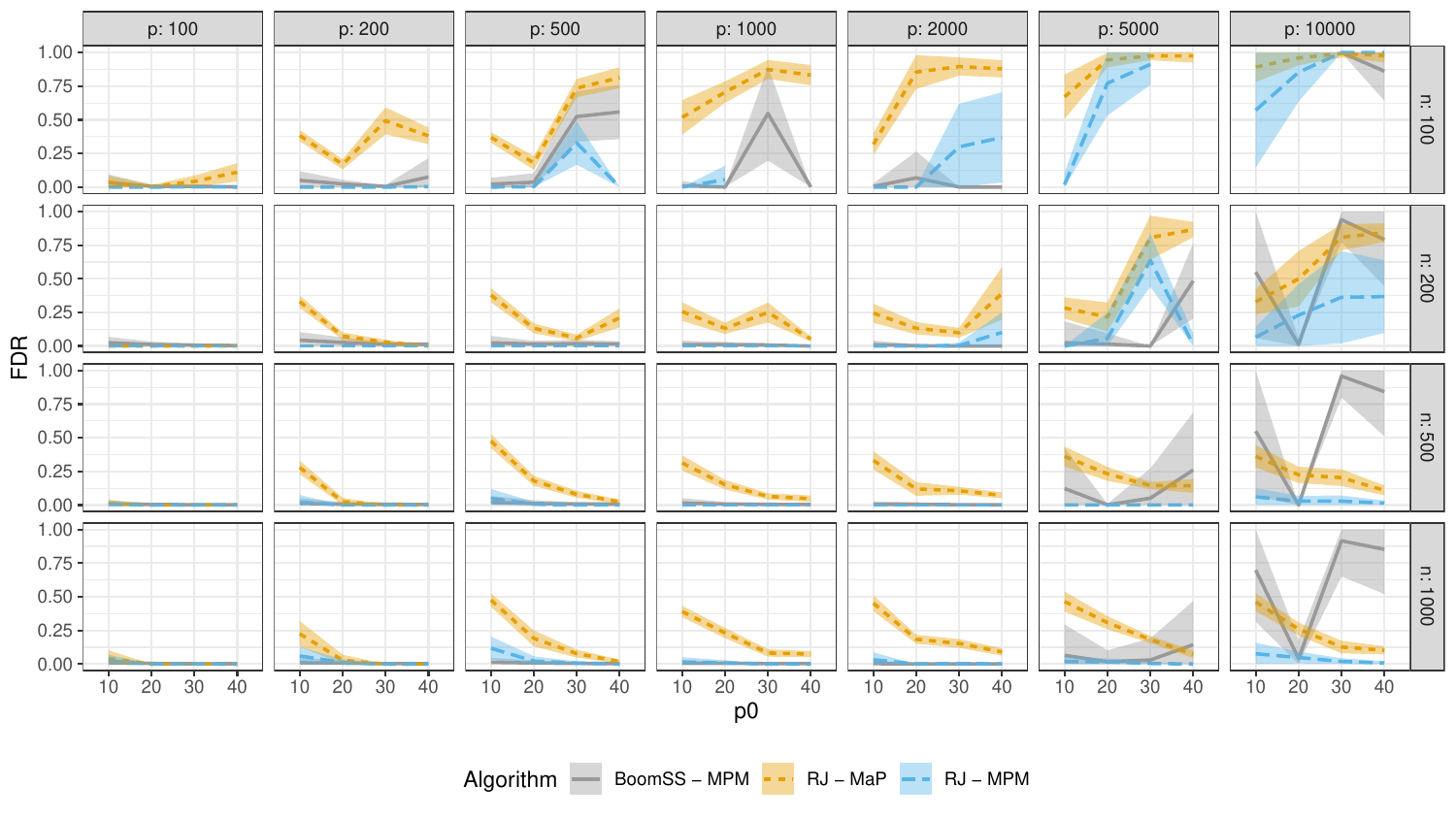}
\caption{Mean FDR with 1 standard error bands of the MPM and the MaP model computed by the RJ (with R update) and the BoomSS algorithms. 40 repetitions for each setting with equicorrelated covariates.}
\label{fig:FDR_Sim2}
\end{center}
\end{figure}
\begin{figure}[H]
\begin{center}
\includegraphics[width=0.8\textwidth]{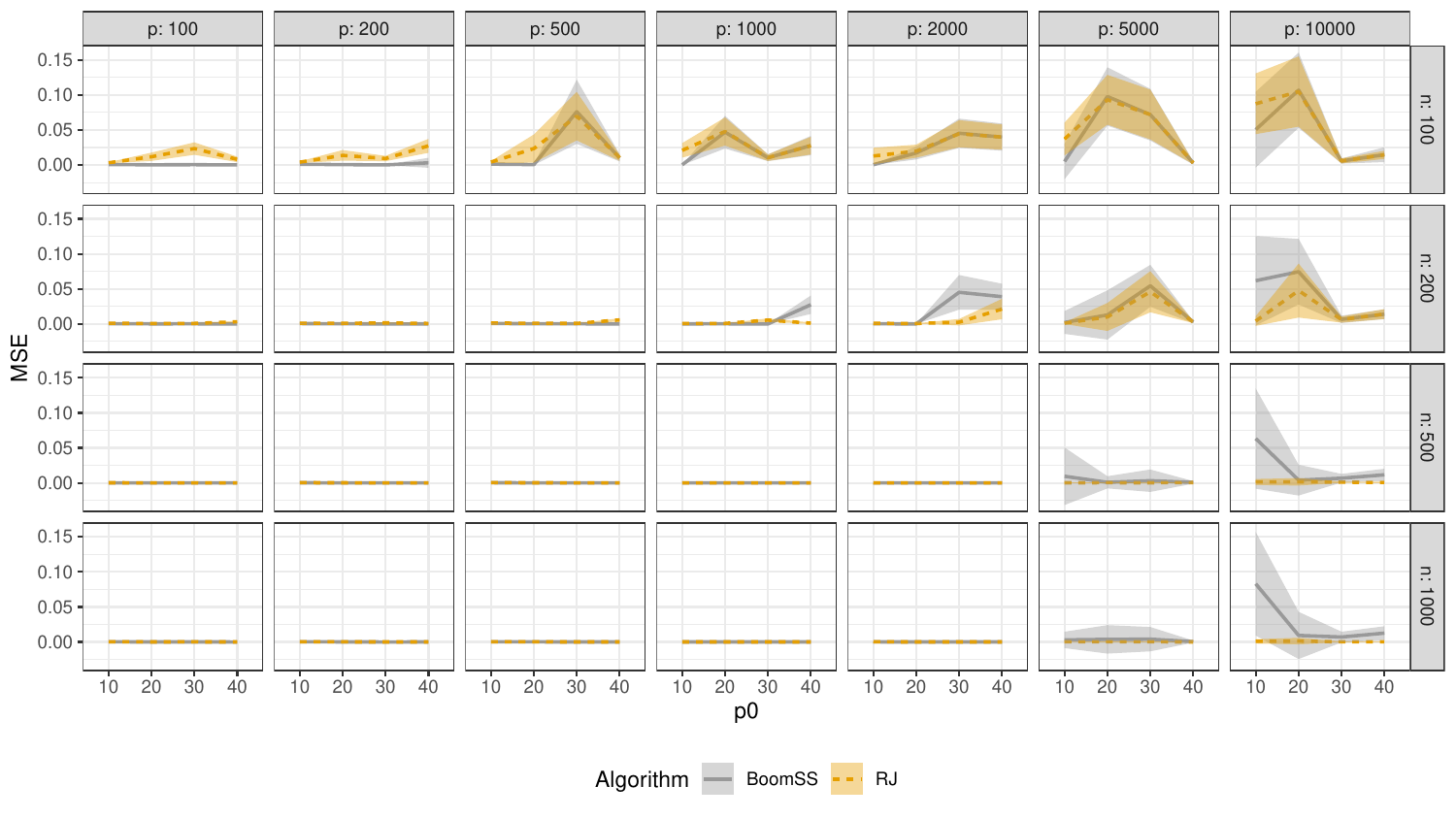}
\caption{Mean square error with 1 standard error bands of the MPM computed by the RJ (with R update) and the BoomSS algorithms. 40 repetitions for each setting with equicorrelated covariates.}
\label{fig:MSE_Sim2}
\end{center}
\end{figure}
%
\subsection{Decreasingly correlated covariates}
Figures \ref{fig:auc_Sim3}--\ref{fig:MSE_Sim3} present results for the scenario in which covariates exhibit a decreasing correlation structure.%
The model performance remains consistent with the other scenarios, indicating that the model and selection algorithms are resilient to varying levels of correlation among covariates. This suggests that the algorithm accuracy in estimating $\boldsymbol{\beta}$ and selecting relevant covariates is not significantly impacted by a gradual increase in correlation. Furthermore, the similarity in results across different correlation structures underscores the algorithm robustness and adaptability to changes in dependency patterns among covariates, an encouraging property for applications where covariate correlations may vary.
%
\begin{figure}[H]
\begin{center}
\includegraphics[width=0.8\textwidth]{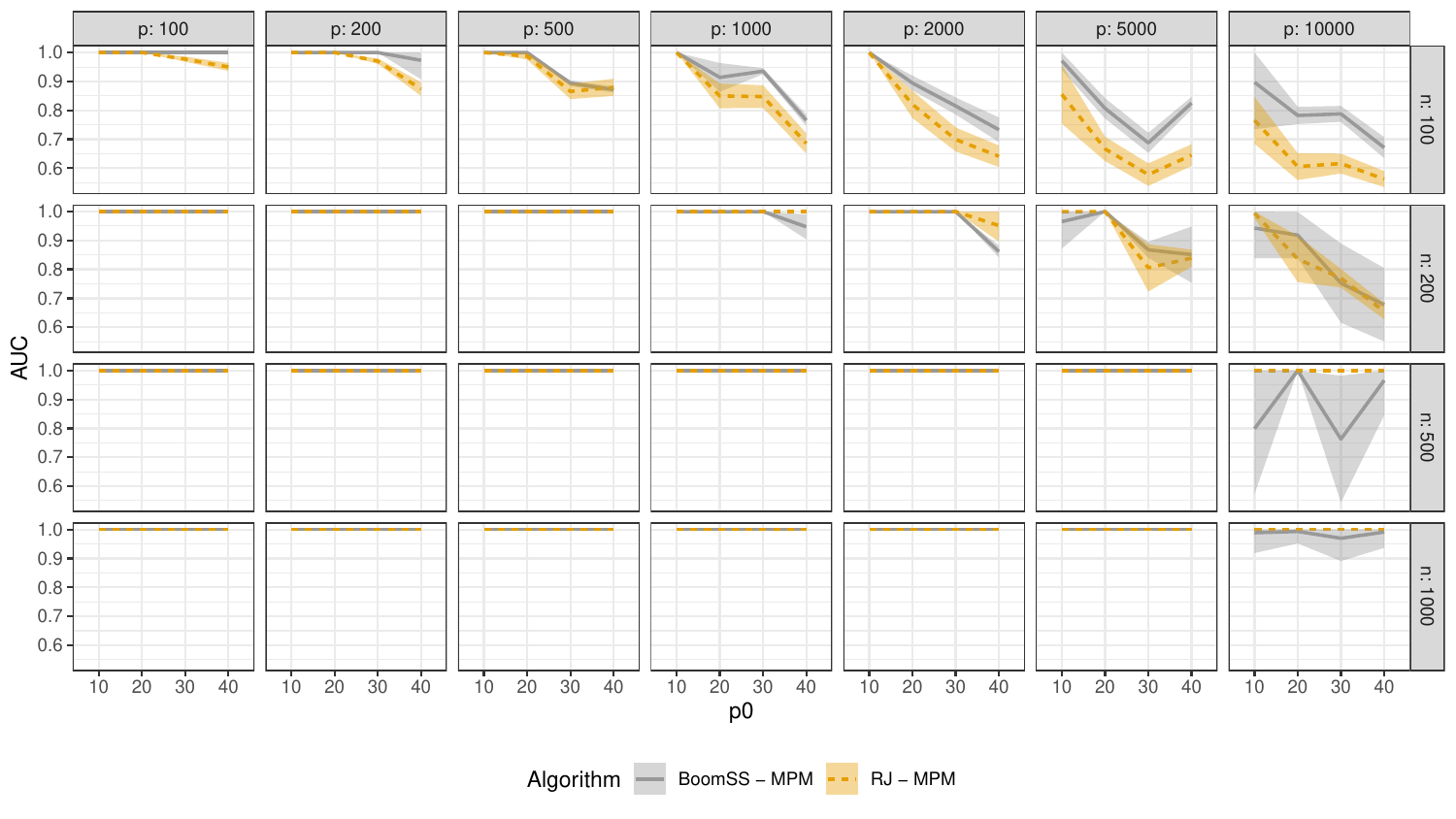}
\caption{Mean AUC with 1 standard error bands of the median probability model computed by the RJ algorithm (with R update) and the SSVS algorithm. 40 repetitions for each setting with decreasingly dependent covariates.}
\label{fig:auc_Sim3}
\end{center}
\end{figure}
\begin{figure}[H]
\begin{center}
\includegraphics[width=0.8\textwidth]{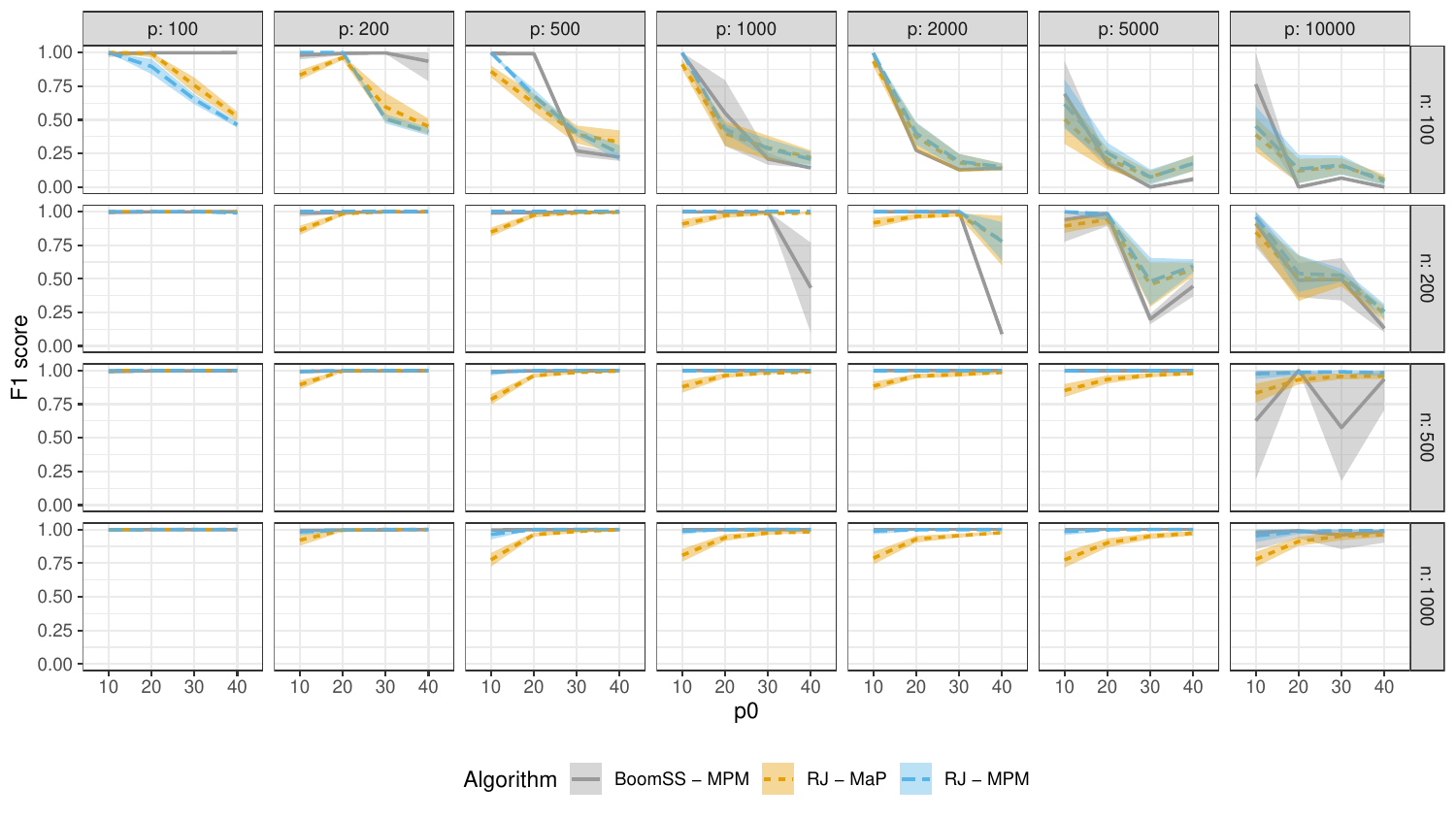}
\caption{Mean F1 score with 1 standard error bands of the MPM and the MaP model computed by the RJ (with R update) and the BoomSS algorithms. 40 repetitions for each setting with decreasingly dependent covariates.}
\label{fig:F1_Sim3}
\end{center}
\end{figure}
\begin{figure}[H]
\begin{center}
\includegraphics[width=0.8\textwidth]{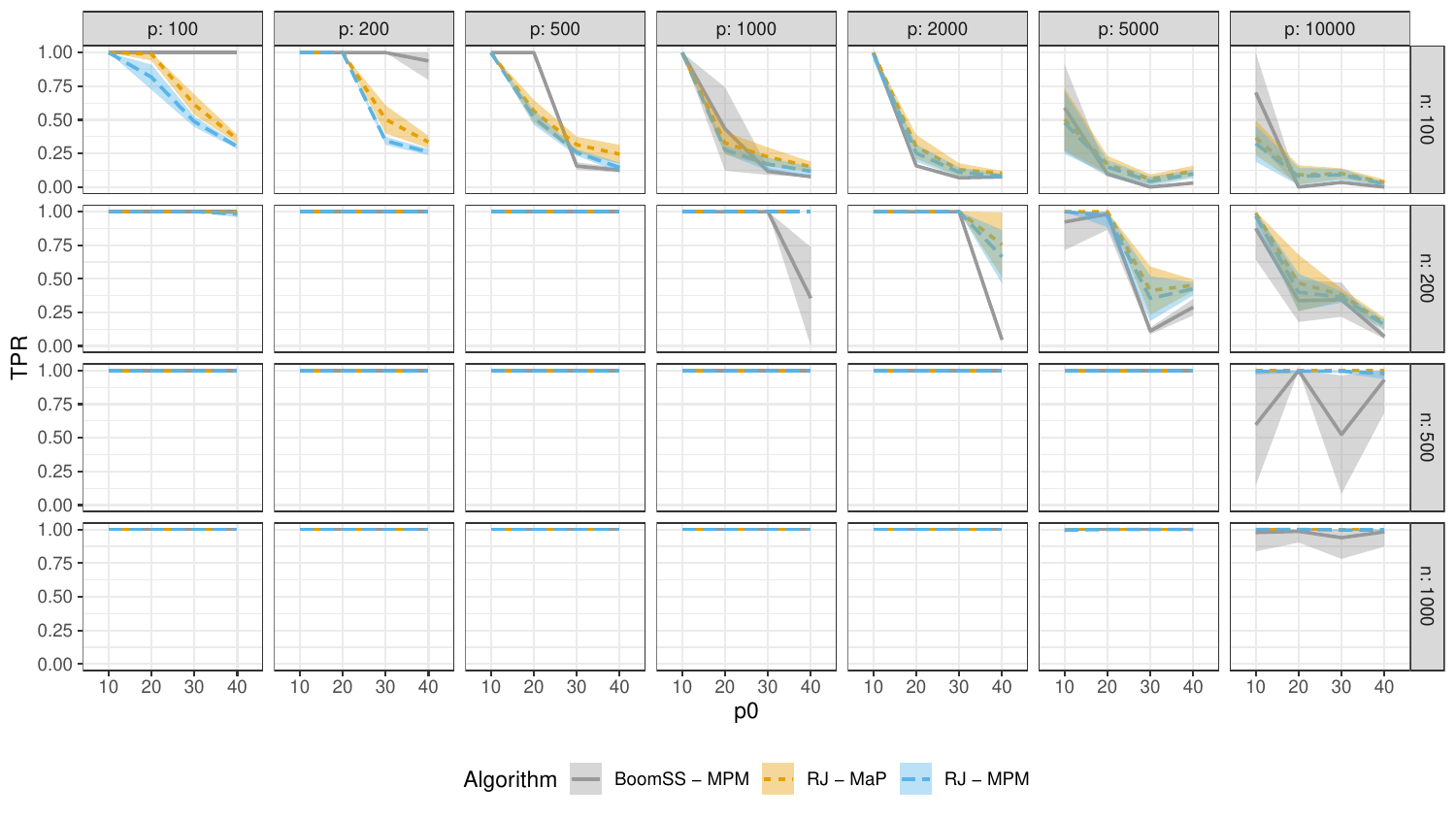}
\caption{Mean TPR with 1 standard error bands of the MPM and the MaP model computed by the RJ (with R update) and the BoomSS algorithms. 40 repetitions for each setting with decreasingly dependent covariates.}
\label{fig:TPR_Sim3}
\end{center}
\end{figure}
\begin{figure}[H]
\begin{center}
\includegraphics[width=0.8\textwidth]{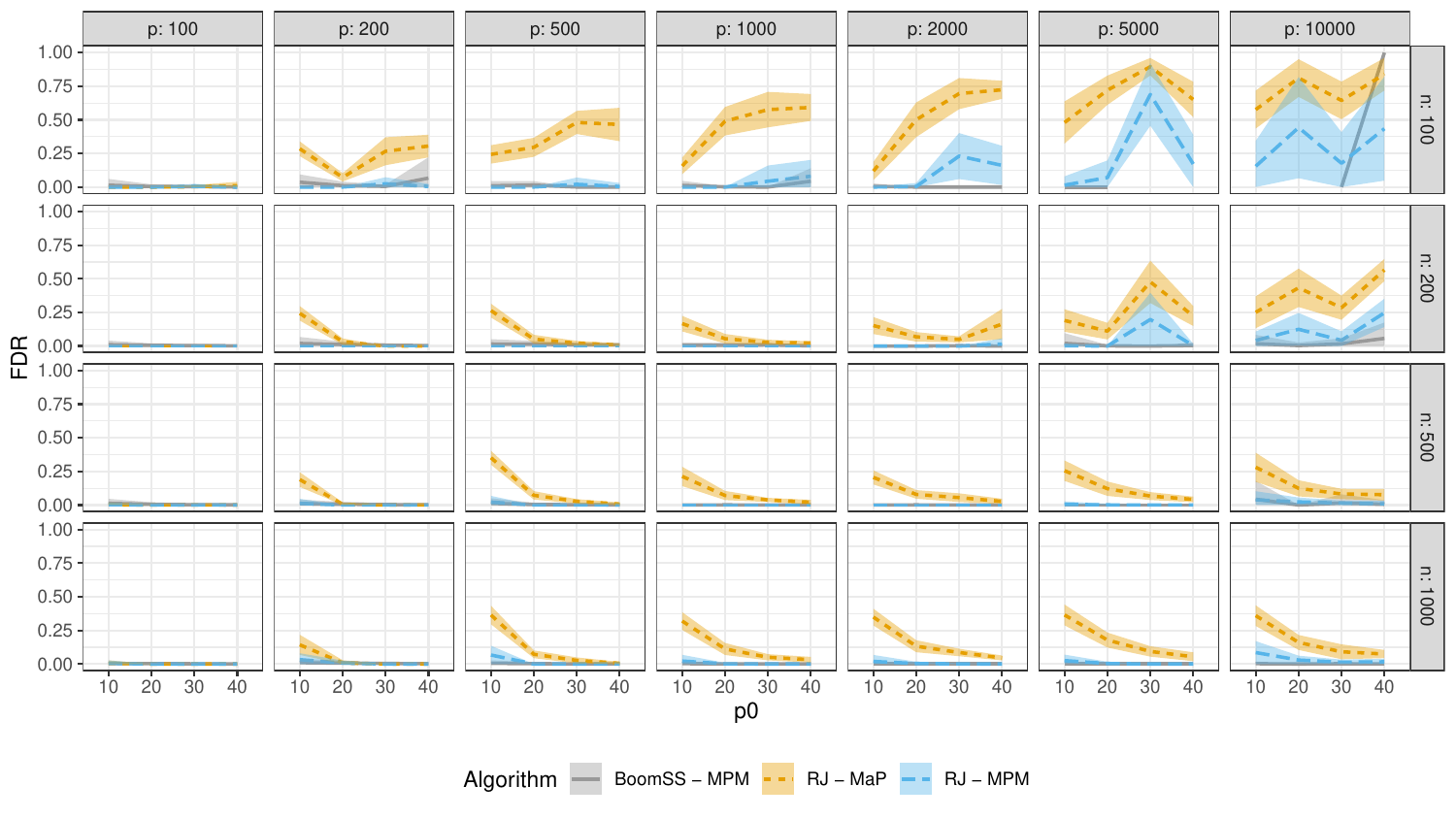}
\caption{Mean FDR with 1 standard error bands of the MPM and the MaP model computed by the RJ (with R update) and the BoomSS algorithms. 40 repetitions for each setting with decreasingly dependent covariates.}
\label{fig:FDR_Sim3}
\end{center}
\end{figure}
\begin{figure}[H]
\begin{center}
\includegraphics[width=0.8\textwidth]{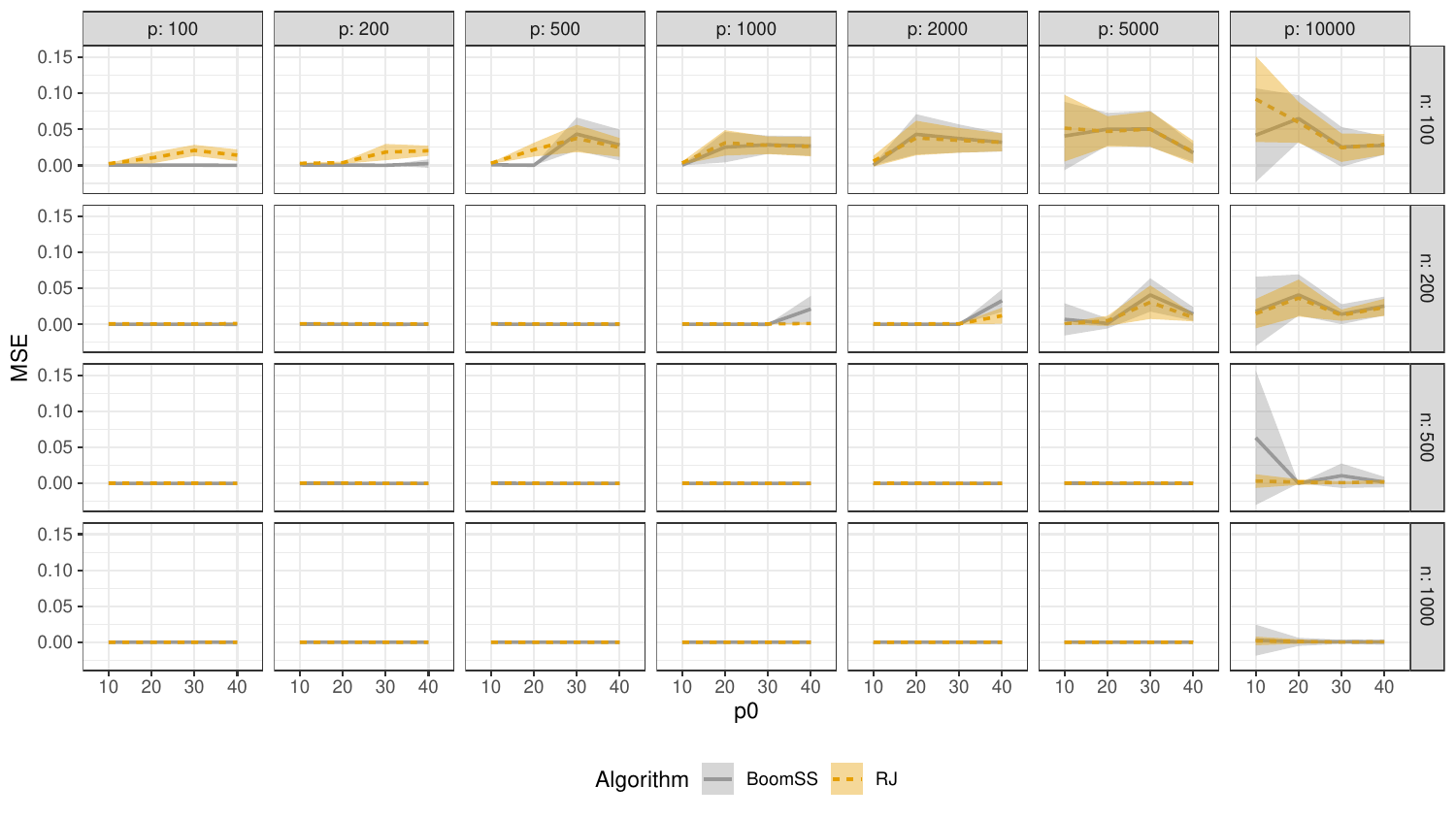}
\caption{Mean square error with 1 standard error bands of the MPM computed by the RJ (with R update) and the BoomSS algorithms. 40 repetitions for each setting with decreasingly dependent covariates.}
\label{fig:MSE_Sim3}
\end{center}
\end{figure}
%
\section{Inflation data description}
\label{sec:dataset_description}
%
We consider the quarterly changes in the Consumer Price Index (CPIAUCSL) as a measure of inflation. Inflation is predicted using quarterly data from several macroeconomic indicators downloaded from the FRED online database, see \cite{bernardi_etal.2016}. Table \ref{tab:US_inflation_data} provides a complete and detailed description of the dataset and the transformations applied to make the data stationary. In this example, we consider quarterly observations for the period from 1991-Q3 to 2023-Q4. Further details on the variables used and their sources can be found in the data appendix of \cite{bernardi_etal.2016}. 
%
\begin{table}[h]
\setlength{\tabcolsep}{5pt}
\caption{Inflation dataset. The column tcode denotes the following data transformation for a series $x$: (1) no transformation; (2) $\Delta x_t$; (3) $\log(x_t)$; (4) $\Delta\log (x_t)$; (5) Percentage change $\Delta x_t/x_{t-1}$. The FRED column gives mnemonics in FRED followed by a short description.} 
%
\centering
\resizebox{0.95\columnwidth}{!}{
\begin{tabular}{lllllll}\\
\toprule
Id& Name&Type& tcode & FRED & Description\\
\cmidrule(lr){1-1}\cmidrule(lr){2-2}\cmidrule(lr){3-3}\cmidrule(lr){4-4}\cmidrule(lr){5-5}\cmidrule(lr){6-6}
%
1	&	\textrm{DATE}	&	--	&--&--&	date	\\
\multirow{2}{*}{2}	 &	\multirow{2}{*}{\textrm{CPIAUCSL}}&\multirow{2}{*}{Index}& \multirow{2}{*}{5}& \multirow{2}{*}{\textrm{CPIAUCSL}}		&	Consumer Price Index for All Urban Consumers:	\\
&&&&& All Items in U.S. City Average\\
\multirow{2}{*}{3}	&	\multirow{2}{*}{\textrm{CPILFESL}} &\multirow{2}{*}{Index}&\multirow{2}{*}{5}& \multirow{2}{*}{\textrm{CPILFESL}}		&	Consumer Price Index for All Urban Consumers: 	\\
&&&&& All Items Less Food and Energy in U.S. City Average \\
4	&	\textrm{UNRATE} &Percent	&1&\textrm{UNRATE}	&		Unemployment Rate\\
5	&	\textrm{EC}	&	Index	&5&	\textrm{DPCERA3M086SBEA}	&	Real Personal Consumption Expenditures\\
6	&	\textrm{PRFI}	&	Level	&5&	\textrm{PRFI}	&	Private Residential Fixed Investment\\
7	&	\textrm{GDPC1}	&	Level	&5&	\textrm{GDPC1}	&	Real Gross Domestic Product 	\\
8	&	\textrm{HOUST}	&	Level	&3&	\textrm{HOUST}	&		New Privately-Owned Housing Units Started: Total Units\\
9	&	\textrm{USPRIV}&	Level	&5&	\textrm{USPRIV}	&	Employees, Total Private	\\
10	&	\textrm{TB3MS}	&	Percent	&1&	\textrm{TB3MS}	&	3-Month Treasury Bill Secondary Market Rate 	\\
\multirow{2}{*}{11}	&	\multirow{2}{*}{\textrm{GS10}}	&	\multirow{2}{*}{Percent}	&\multirow{2}{*}{1}&	\multirow{2}{*}{\textrm{GS10}}	&	Market Yield on U.S. Treasury 	\\
&&&&& Securities at 10-Year Constant Maturity \\
\multirow{2}{*}{12}	&	\multirow{2}{*}{\textrm{T10Y3MM}} &	\multirow{2}{*}{Percent}	&\multirow{2}{*}{1}&	\multirow{2}{*}{\textrm{T10Y3MM}}	&	10-Year Treasury Constant Maturity Minus 	\\
&&&&& 3-Month Treasury Constant Maturity 		\\
\multirow{2}{*}{13}	&	\multirow{2}{*}{\textrm{T10YFFM}} &	\multirow{2}{*}{Percent}	&\multirow{2}{*}{1}&	\multirow{2}{*}{\textrm{T10YFFM}}	&	10-Year Treasury Constant Maturity Minus  	\\
&&&&& Federal Funds Rate		\\
14	&	\textrm{M1SL}	&	Level	&5&\textrm{M1SL}	&		Money supply - M1\\
15	&	\textrm{MICH}	&	Percent	&1&\textrm{MICH}	&		University of Michigan: Inflation Expectation\\
16	&	\textrm{PPIACO}	&	Index	&2&	\textrm{PPIACO}	&		Producer Price Index by Commodity: All Commodities \\
17	&	\textrm{DJIA}	&	Index	&5&\textrm{DJIA}&		Dow Jones Industrial Average Index\\
18	&	\textrm{PMI}	&	Index	&2&\textrm{PMI}	&	Purchasing Manager's composite index (IMS)\\
19	&	\textrm{NAPMSDI}&	Index	&2&	\textrm{NAPMSDI}	&		NAPM vendor deliveries index\\
20	&	\textrm{OILPRICE}&	Index	&4&	\textrm{WTISPLC}	&		Spot Crude Oil Price: West Texas Intermediate (WTI)  \\
21	&	\textrm{GASPRICE}&	Index	&4&	\textrm{GASREGCOVM}	&		US Regular Conventional Gas Price \\
%
\bottomrule 
\end{tabular}
} 
\label{tab:US_inflation_data}
\end{table} 
%
\section{Additional results for real data applications}
\label{appG}
%
Table \ref{tab:Bardet-Biedl_summary_stat} presents the posterior summary for the five best models generated by the RJ algorithm. This table offers a comprehensive overview of gene expression related to this syndrome, which is instrumental in identifying potentially crucial genes involved in its development or manifestation. Notably, only three genes, Defb1, Actl7b, and Tmem230, are included across all models, with two of these genes exhibiting higher marginal inclusion probabilities. An intriguing aspect of the results is the marked disparity in the estimated probability for model $\mathcal{M}_1$, which stands at $0.10$. In contrast, the remaining models, $\mathcal{M}_2$ to $\mathcal{M}_5$ exhibit similar probabilities of approximately $0.05$ (refer to Table \ref{tab:inflation_posterior_models}). Despite this variance in model probabilities, the first four models share roughly $60\%$ of the included genes, indicating a degree of overlap and consistency among them. However, model $\mathcal{M}_5$ is markedly different, as about $65\%$ of the genes it includes have not been identified in the other models.
%
%
%
When compared to other approaches outlined in Table  \ref{tab:inflation_posterior_models_RW}, specifically the adaptive lasso and elastic net, only a small number of probes from Table \ref{tab:Bardet-Biedl_summary_stat}  are identified as relevant regressors---just $1$ for the adaptive lasso and $2$ for the elastic net. This focused feature inclusion reflects the stringent selection criteria of these models, indicating a more conservative approach compared to other methods that may capture a broader set of variables.\par
Of the 25 probe sets uniquely selected by the top five models as significantly associated with trim32, 20 had identifiable gene symbols. These genes were dgat1, eif2b3, six3, defb1, nrp1, actl7b, tmem230, tmtc4, tmrc6c, nsrp1, zfp62, il17b, rbm47, asic1, ubl7, ino80c, kcne2, nrn1, galnt10 and adss. In particular, according to \texttt{https://www.genecards.org}, only for UBL7 (Ubiquitin-Like 7) there might be indirect functional overlap in protein modification pathways, but no specific direct interaction with trim32 has been documented. The other associations found may be useful for researchers in studying the genetic factors contributing to Bardet-Biedl syndrome.

%
\begin{table}[!hp]
\setlength{\tabcolsep}{16pt} 
\renewcommand{\arraystretch}{0.85} 
\caption{Bardet-Biedl syndrome gene expression study. Posterior summary statistics for the five best models $(\mathcal{M}_1,\dots,\mathcal{M}_5)$, as estimated by the RJ algorithm.  For each probe the reported summary statistics are the mean, the standard deviation (in italics), and the marginal inclusion probability (MIP). The summary statistics are calculated using only the post burn-in draws corresponding to the $j$-th model.}
%
\centering
\footnotesize
\resizebox{1\columnwidth}{!}{
\begin{tabular}{p{1cm} p{1.0cm}p{0.8cm}p{0.80cm}p{0.8cm} p{0.80cm}p{0.80cm}p{0.5cm}}
\toprule
& &\multicolumn{6}{c}{\textit{Summary statistics}}\\
\cmidrule(lr){3-8}
Probe & gene &  $\mathcal{M}_1$ & $\mathcal{M}_2$ & $\mathcal{M}_3$ & $\mathcal{M}_4$ & $\mathcal{M}_5$   & MIP \\ 
%
\cmidrule(lr){1-1}\cmidrule(lr){2-2}\cmidrule(lr){3-8}
\multirow{2}{*}{1369660\_at} & \multirow{2}{*}{Defb1} & -0.35 & -0.31 & -0.31 & -0.37 & -0.27&\multirow{2}{*}{0.76} \\ 
  & & {\it (0.08)} & {\it (0.06)} & {\it (0.08)} & {\it (0.08)} &{\it (0.08)} \\ 
\multirow{2}{*}{1376429\_at}& \multirow{2}{*}{Actl7b} &  -0.20 & -0.19 & -0.16 & -0.18 & -0.18 & \multirow{2}{*}{0.99}\\ 
  &&  {\it (0.07)} & {\it (0.07)} & {\it (0.07)} & {\it (0.08)} & {\it (0.09)} \\ 
\multirow{2}{*}{1389910\_at}&\multirow{2}{*}{Tmem230}& 0.64 & 0.64 & 0.59 & 0.57 & 0.55 & \multirow{2}{*}{0.91}\\ 
 &&   {\it (0.11)} & {\it (0.12)}& {\it (0.11)} &{\it (0.09)} & {\it (0.11)} \\ 
\midrule
\multirow{2}{*}{1368967\_at} &\multirow{2}{*}{Eif2b3}&  -0.10 & -0.04 & -0.02 & -0.01 &&\multirow{2}{*}{0.52}\\ 
&&    {\it (0.10)} & {\it (0.10)} & {\it (0.10)} &{\it (0.09)} \\ 
\multirow{2}{*}{1369028\_at }&\multirow{2}{*}{Six3}&  -0.18 & -0.11 & -0.15 & -0.18 &&\multirow{2}{*}{0.82}\\ 
&&    {\it (0.11)} & {\it (0.11)} &{\it (0.10)} & {\it (0.10)} \\ 
\multirow{2}{*}{1370570\_at} &  \multirow{2}{*}{Nrp1}&  -0.07 & -0.05 & -0.03 & -0.04&&\multirow{2}{*}{0.67} \\ 
&&   {\it (0.10)} & {\it (0.11)} & {\it (0.09)} & {\it (0.08)} \\ 
\multirow{2}{*}{1376386\_at} &&  -0.05 & -0.02 & -0.01 & -0.02&&\multirow{2}{*}{0.93} \\ 
&&    {\it (0.09)} & {\it (0.09)} & {\it (0.08)} & {\it (0.08)} \\ 
\multirow{2}{*}{1382674\_a\_at} &\multirow{2}{*}{Tnrc6c}&  -0.07 & -0.11 & -0.07 & -0.07&&\multirow{2}{*}{0.90} \\ 
 &&  {\it (0.09)} & {\it (0.09)} & {\it (0.09)} & {\it (0.08)} \\ 
\multirow{2}{*}{1382904\_at} &   \multirow{2}{*}{Nsrp1} &  0.18 & 0.14 & 0.11 & 0.18 &&\multirow{2}{*}{0.58}\\ 
 &&   {\it (0.09)} & {\it (0.07)} & {\it (0.08)} & {\it (0.08)} \\ 
\multirow{2}{*}{1385539\_at}&&  -0.22 & -0.14 && -0.15 & -0.18 &\multirow{2}{*}{0.73} \\ 
  &&  {\it (0.10)} & {\it (0.10)} && {\it (0.11)} & {\it (0.07)} \\
\midrule
\multirow{2}{*}{1367915\_at} &\multirow{2}{*}{Dgat1}&  0.06 & 0.12 & 0.11 &&&\multirow{2}{*}{0.30}\\ 
 &&   {\it (0.10)} & {\it (0.10)} & {\it (0.10)} \\ 
\multirow{2}{*}{1379541\_at}&\multirow{2}{*}{Tmtc4}&  -0.01 & -0.10 & -0.06 &&&\multirow{2}{*}{0.26}\\ 
&&    {\it (0.14)} & {\it (0.14)} & {\it (0.14)} \\ 
\midrule
\multirow{2}{*}{1375354\_at}&& &&& -0.17 & -0.16 &\multirow{2}{*}{0.64}\\ 
 &&  &&& {\it (0.07)} & {\it (0.07)} \\ 
\multirow{2}{*}{1379495\_at} && & 0.28 & 0.23 &&&\multirow{2}{*}{0.28}\\ 
&&    &{\it (0.11)} & {\it (0.13)} &&\\ 
\midrule
\multirow{2}{*}{1371045\_at} &   \multirow{2}{*}{Asic1} & &&&& -0.10& \multirow{2}{*}{0.16}\\ 
 &&  &&&& {\it (0.08)} \\ 
1371452\_at &\multirow{2}{*}{Ubl7}& &&&& -0.12& \multirow{2}{*}{0.50}\\ 
&&   &&&& {\it (0.09)} \\ 
\multirow{2}{*}{1374479\_at} &\multirow{2}{*}{Ino80c}& &&&& 0.14 &\multirow{2}{*}{0.19}\\ 
 &&  &&&& {\it (0.10)} \\ 
\multirow{2}{*}{1376445\_at} &\multirow{2}{*}{Il17b}& &&& -0.10&&\multirow{2}{*}{0.20} \\ 
&&   &&& {\it (0.08)}& \\ 
\multirow{2}{*}{1376728\_at} &\multirow{2}{*}{Rbm47}& &&& 0.03&&\multirow{2}{*}{0.19} \\ 
 &&  &&& {\it (0.07)} \\ 
\multirow{2}{*}{1379029\_at} &\multirow{2}{*}{Zfp62}& && 0.14 &&&\multirow{2}{*}{0.28}\\ 
&&   && {\it (0.10)} \\ 
\multirow{2}{*}{1386770\_x\_at} &\multirow{2}{*}{Kcne2}&&&& & -0.03 &\multirow{2}{*}{0.17}\\ 
 && &&& & {\it (0.08)} \\ 
\multirow{2}{*}{1386969\_at} &\multirow{2}{*}{Nrn1}& &&&& 0.07 &\multirow{2}{*}{0.46}\\ 
&&   &&&& {\it (0.08)} \\ 
\multirow{2}{*}{1387290\_at} &\multirow{2}{*}{Galnt10}& &&&& -0.00&\multirow{2}{*}{0.48} \\ 
&&   &&&& {\it (0.09)} \\ 
\multirow{2}{*}{1390394\_at}&& &&&& 0.16&\multirow{2}{*}{0.34} \\ 
 &&  &&&& {\it (0.10)} \\ 
\multirow{2}{*}{1399050\_at} &\multirow{2}{*}{Adss}&&&& & -0.12&\multirow{2}{*}{0.26} \\ 
&&   &&&& {\it (0.11)} \\ 
\bottomrule 
\end{tabular}
} 
\label{tab:Bardet-Biedl_summary_stat} 
\end{table} 
%

\clearpage

\bibliographystyle{apalike} 
\bibliography{refs}

\end{document}